\documentclass[a4paper]{article}
\usepackage[T1]{fontenc}
\usepackage{fullpage}  
\usepackage{color}
\usepackage{epsfig}
\usepackage{amssymb}
\usepackage{authblk}
\usepackage{color}
\usepackage{subcaption}
\usepackage{enumitem}

\newcommand{\mtimestamp}{compil\'e le  \ifnum\day<10 0\fi\the\day.\,%
\ifnum\month<10 0\fi\the\month.\,%
\the\year\ \`a \xxivtime\,h}

\usepackage{graphicx}
\usepackage{amsmath,amssymb,amsthm}

\newtheorem*{theorem*}{Theorem}

\newtheorem{problem}{Problem}

\newtheorem{theorem}{Theorem}

\newtheorem{definition}{Definition}

\newtheorem{corollary}{Corollary}
\newtheorem{lemma}{Lemma}

\newtheorem{claim}{Claim}
\newtheorem{subclaim}{Subclaim}[claim]

\newcommand{\qedclaim}{\hfill $\diamond$ \medskip}
\newcommand{\qedsubclaim}{\hfill $\circ$ \medskip}
\newenvironment{proofclaim}{\noindent{\em Proof.}}{\qedclaim}
\newenvironment{proofsubclaim}{\noindent{\em Proof.}}{\qedsubclaim}
\newenvironment{proofof}[1]{\medskip\noindent\emph{Proof of #1. }\ignorespaces}{\hfill\qed\medskip\par\noindent\ignorespacesafterend}

\newenvironment{keywords}{
	\list{}{\advance\topsep by0.35cm\relax\small
		\leftmargin=1cm
		\labelwidth=0.35cm
		\listparindent=0.35cm
		\itemindent\listparindent
		\rightmargin\leftmargin}\item[\hskip\labelsep
	\bfseries Keywords:]}
{\endlist}

\title{On computing tree and path decompositions with metric constraints on the bags}

\author[1,2]{Guillaume Ducoffe}
\author[3]{Sylvain Legay}
\author[2,1]{Nicolas Nisse}

\affil[1]{\small Univ. Nice Sophia Antipolis, CNRS, I3S, UMR 7271, 06900 Sophia Antipolis, France}
\affil[2]{\small INRIA, France}
\affil[3]{\small LRI, Univ Paris Sud, Universit\'e Paris-Saclay, 91405 Orsay, France}

\date{}

\begin{document}

\maketitle

\begin{abstract}
%Two fundamental problems in graph algorithms are to embed a graph into a tree or the infinite line with minimum distortion of the distances in the graph.	
%It has been proved recently these two problems can be well approximated if it is computed in the first place a Robertson and Seymour \emph{tree decomposition} (resp. a path decomposition) whose each bag has bounded diameter and radius.	
We here investigate on the complexity of computing the \emph{tree-length} and the \emph{tree-breadth} of any graph $G$, that are respectively the best possible upper-bounds on the diameter and the radius of the bags in a tree decomposition of $G$.
\emph{Path-length} and \emph{path-breadth} are similarly defined and studied for path decompositions.	
So far, it was already known that tree-length is NP-hard to compute.
We here prove it is also the case for tree-breadth, path-length and path-breadth.
Furthermore, we provide a more detailed analysis on the complexity of computing the tree-breadth.
In particular, we show that graphs with tree-breadth one are in some sense the hardest instances for the problem of computing the tree-breadth.
We give new properties of graphs with tree-breadth one.
Then we use these properties in order to recognize in polynomial-time all graphs with tree-breadth one that are planar or bipartite graphs.
On the way, we relate tree-breadth with the notion of \emph{$k$-good} tree decompositions (for $k=1$), that have been introduced in former work for routing.
As a byproduct of the above relation, we prove that deciding on the existence of a $k$-good tree decomposition is NP-complete (even if $k=1$).
All this answers open questions from the literature.
\end{abstract}


%\tableofcontents
%\newpage

\section{Introduction}
\label{sec:intro}

\paragraph{Context.}
It is a fundamental problem in metric graph theory~\cite{Bandelt2008} to embed a graph into a simpler metric space while minimizing the (multiplicative) distortion of the distances in the graph.
In particular, minimum distortion embeddings of a graph into a tree or a path have practical applications in computer vision~\cite{Tenenbaum2000}, computational chemistry and biology~\cite{Indyk2001} as well as in network design and distributed protocols~\cite{Herlihy2006}.
The two problems to embed a graph into a tree or a path with minimum distortion are NP-hard~\cite{Agarwala1998,Bǎdoiu2005,Monien1986}.
However, there exists a nice setting in order to approximate these two problems.
More precisely, a series of graph parameters has been introduced in recent work in order to measure how much the \emph{distance distribution} of a graph is close to a tree metric or a path metric~\cite{Dourisboure2007b,Dragan2014a,Dragan2014b}.
We refer to~\cite{Dragan2013,Dragan2014b} for details about the relationships between these parameters and the two above-mentioned embedding problems.
Here we study the complexity of computing these parameters, thereby solving open problems in the literature.

The parameters that are considered in this note can be defined using the terminology of Robertson and Seymour tree decompositions~\cite{Robertson1986}.
Informally, a tree decomposition is a dividing of a graph $G$ into ``bags'': that are overlapping subgraphs that can be pieced together in a tree-like manner (formal definitions will be given in the technical sections of the paper).   
The shortest-path metric of $G$ is ``tree-like'' when each bag of the tree decomposition has bounded diameter and bounded radius, where the distance taken between two vertices in a same bag is their distance in $G$.
The \emph{tree-length}~\cite{Dourisboure2007b} and the \emph{tree-breadth}~\cite{Dragan2014a} of $G$ are respectively the best possible upper-bounds on the diameter and the radius of the bags in a tree decomposition of $G$.
\emph{Path-length}~\cite{Umezawa2009} and \emph{path-breadth}~\cite{Dragan2014b} are defined in the same fashion as tree-length and tree-breadth for path decompositions.
In this paper, we focus on the complexity of computing the four parameters tree-length, tree-breadth, path-length and path-breadth.  

\medskip
Recent studies suggest that some classes of real-life networks -- including biological networks and social networks -- have bounded tree-length and tree-breadth~\cite{Abu2015}.
This metric tree-likeness can be exploited in algorithms.
For instance, bounded tree-length graphs admit compact distance labeling scheme~\cite{Dourisboure2007a} as well as a PTAS for the well-known Traveling Salesman problem~\cite{Krauthgamer2006}.
Furthermore, the diameter and the radius of bounded tree-length graphs can be approximated up to an additive constant in linear-time~\cite{Chepoi2008}.
In contrast to the above result, we emphasize that under classical complexity assumptions the diameter of general graphs \emph{cannot} be approximated up to an additive constant in subquadratic-time, that is prohibitive for very large graphs~\cite{Chechik2014}. 

Note that a large amount of the literature about tree decompositions rather seeks to minimize the \emph{size} of the bags than their diameter. 
The \emph{tree-width}~\cite{Robertson1986} of a graph $G$ is the best possible upper-bound on the size of the bags in a tree decomposition of $G$. 
However, tree-length and the other parameters that are considered in this paper can differ arbitrarily from tree-width;
we refer to~\cite{Coudert2014} for a global picture on the relations between tree-length and tree-width.
Furthermore, one aim of this paper is to complete the comparison between tree-width and path-width on one side, and tree-length, tree-breadth, path-length and path-breadth on the other side, from the complexity point of view.
Let us remind that computing the tree-width (resp. the path-width) is NP-hard~\cite{Arnborg1987,Kashiwabara1979}, however for every fixed $k \geq 1$ there is a linear-time algorithm to decide whether a graph has tree-width at most $k$ (resp., path-width at most $k$)~\cite{Bodlaender1996,Bodlaender1996b}.

\paragraph{Related work.}
The complexity of computing tree-length, tree-breadth, path-length and path-breadth has been left open in several works~\cite{Dourisboure2007b,Dragan2014a,Dragan2014b}.
So far, it has been solved only for tree-length, that is NP-hard to compute.

{\it Tree-length and tree-breadth.} 
It is NP-complete to decide whether a graph has tree-length at most $k$ for every constant $k \geq 2$~\cite{Lokshtanov2010}.
However, the reduction used for tree-length goes through weighted graphs and then goes back to unweighted graphs using rather elegant gadgets.
It does not seem to us these gadgets can be easily generalized in order to apply to the other parameters that are considered in this note.
On a more positive side, there exists a $3$-approximation algorithm for tree-length~\cite{Dourisboure2007b}.
In this aspect, it looks natural to investigate on the complexity of computing the tree-breadth, since any polynomial-time algorithm would imply an improved $2$-approximation algorithm for tree-length.

{\it Path-length and path-breadth.}
There exist constant-factor approximation algorithms for path-length and path-breadth~\cite{Dragan2014b}.
Recently, the minimum eccentricity shortest-path problem -- that is close to the computation of path-length and path-breadth -- has been proved NP-hard~\cite{Dragan2015}.
Let us point out that for every fixed $k$, it can be decided in polynomial-time whether a graph admits a shortest-path with eccentricity at most $k$~\cite{Dragan2015}.
Our results will show the situation is different for path-length and path-breadth than for the  minimum eccentricity shortest-path. 

\paragraph{Our contributions.}
On the negative side, we prove that tree-breadth, path-length and path-breadth are NP-hard to compute.
More precisely:
\begin{itemize}
 \item recognizing graphs with tree-breadth one is NP-complete;
 \item recognizing graphs with path-length two is NP-complete;
 \item recognizing graphs with path-breadth one is NP-complete.
\end{itemize} 
It is remarkable the last two results (for path-length and path-breadth) are obtained using \emph{the same reduction}.
Our reductions have distant similarities with the reduction that was used for tree-length.
However, they do not need any detour through weighted graphs.

\medskip
We next focus our work on tree-breadth (although part of the results may extend to the three other parameters that are considered in this note).
We give a more in-depth analysis on the complexity of computing this parameter.
In particular, we prove it is equally hard to compute tree-breadth as to recognize graphs with tree-breadth one.
Therefore, graphs with tree-breadth one are in some sense the hardest instances for the problem of computing the tree-breadth.
The latter partly answers an open question from~\cite{Dragan2014a}, where it was asked for a characterization of graphs with tree-breadth one.
%Indeed, our result implies that finding such a characterization is equivalent to characterize graphs with tree-breadth at most $k$ for every fixed $k$.
We also prove a few properties of graphs with tree-breadth one.
In particular, graphs with tree-breadth one are exactly those graphs admitting a \emph{$1$-good tree decomposition}, that is a tree decomposition whose each bag has a spanning star.
The more general notion of $k$-good tree decompositions was introduced in~\cite{Kosowski2015} to obtain new compact routing schemes.
Note that as a byproduct of the above relation between $1$-good tree decompositions and graphs with tree-breadth one, we obtain that deciding on the existence of a $k$-good tree decomposition is an NP-complete problem even when $k=1$. 

Finally, on the algorithmic side, we show how to recognize in polynomial time all graphs of tree-breadth one that are planar or bipartite. 
In particular, our recognition algorithm for planar graphs of tree-breadth one relies upon deep structural properties of these graphs.

\medskip
Definitions and useful notations are given in Section~\ref{sec:def}.
All our results of NP-completeness are listed and proved in Section~\ref{sec:np-hard}.
Sections~\ref{sec:tb} and~\ref{sec:planar} are devoted to the computation of tree-breadth.
In particular, in Section~\ref{sec:planar} we present and we prove correctness of an algorithm to recognize planar graphs with tree-breadth one.
Finally, we discuss about some open questions in the conclusion (Section~\ref{sec:conclusion}). 

\section{Definitions and preliminary results}
\label{sec:def}

We refer to~\cite{Diestel2012} for a survey on graph theory.
Graphs in this study are finite, simple, connected and unweighted.
Let $G=(V,E)$ be a graph. 
For any $X \subseteq V$, let $G[X]$ denote the subgraph of $G$ induced by $X$. 
For any subgraph $H$ of $G$, let $N_H(v)$ denote the set of neighbors of $v\in V$ in $H$, and let $N_H[v]=N_H(v)\cup \{v\}$. 
The distance $dist_H(u,v)$ between two vertices $u,v \in V$ in $H$ is the minimum length (number of edges) of a path between $u$ and $v$ in a subgraph $H$ of $G$. 
In what follows, we will omit the subscript when no ambiguity occurs. 
A set $S \subseteq V$ is a dominating set of $G$ if any vertex of $V\setminus S$ has a neighbor in $S$. 
The {\it dominating number} $\gamma(G)$ of a graph $G$ is the minimum size of a dominating set of $G$.

\paragraph{Tree decompositions and path decompositions of a graph.}
A {\it tree decomposition} $(T,{\cal X})$ of $G$ is a pair consisting of a tree $T$ and of a family ${\cal X}=(X_t)_{t \in V(T)}$ of subsets of $V$ indexed by the nodes of $T$ and satisfying:
\begin{itemize}
\item $\bigcup_{t \in V(T)}X_t=V$;
\item for any edge $e=\{u,v\} \in E$, there exists $t\in V(T)$ such that $u,v \in X_t$;
\item for any $v \in V$, the set of nodes $t \in V(T)$ such that $v \in X_t$ induces a subtree, denoted by $T_v$, of $T$. 
\end{itemize}

The sets $X_t$ are called {\it the bags} of the decomposition. 
If no bag is contained into another one, then the tree decomposition is called \emph{reduced}.
Starting from any tree decomposition, a reduced tree decomposition can be obtained in polynomial-time by contracting any two adjacent bags with one contained into the other until it is no more possible to do that.

In the following we will make use of the \emph{Helly property} in our proofs:

\begin{lemma}~\cite[Helly property]{Balakrishnan2012}
\label{lem:helly}	
 Let $T$ be a tree and let $T_1,T_2,\ldots,T_k$ be a finite family of pairwise intersecting subtrees.
 Then, $\bigcap_{i=1}^k T_i \neq \emptyset$, or equivalently there is a node contained in all the $k$ subtrees.
\end{lemma}

%The proof of the Helly property is part of the folklore about tree decompositions.
Finally, let $(T,{\cal X})$ be a tree decomposition, it is called a \emph{path decomposition} if $T$ induces a path.

\paragraph{Metric tree-likeness and path-likeness.}
All graph invariants that we consider in the paper can be defined in terms of tree decompositions and path decompositions.
Let $(T,{\cal X})$ be any tree decomposition of a graph $G$.
For any $t \in V(T)$, 
\begin{itemize}
 \item the {\it diameter} of bag $X_t$ equals $\max_{v,w \in X_t} dist_G(v,w)$;
 \item the {\it radius} $\rho(t)$ of a bag $X_t$ equals $\min_{v \in V} \max_{w\in X_t} dist_G(v,w)$. 
\end{itemize}
The {\it length} of $(T,{\cal X})$ is the maximum diameter of its bags, while the {\it breadth} of $(T,{\cal X})$ is the maximum radius of its bags.
The {\it tree-length} and the {\it tree-breadth} of $G$, respectively denoted by $tl(G)$ and $tb(G)$, are the minimum length and breadth of its tree decomposition, respectively.

Let $k$ be a positive integer, the tree decomposition $(T,{\cal X})$ is called \emph{$k$-good} when each bag contains a dominating induced path of length at most $k-1$.
%The \emph{good treewidth} of a graph $G$, denoted by $gtw(G)$, is the minimum $k$ such that it admits a $k$-good tree-decomposition.
%Given $G=(V,E)$, let $ch(G)$ be the length of a longest induced cycle of $G$.
It is proved in~\cite{Kosowski2015} every graph $G$ has a $k$-good tree decomposition for $k = ch(G) - 1$, with $ch(G)$ denoting the size of a longest induced cycle of $G$.
Finally, path-length, path-breadth and $k$-good path decompositions are similarly defined and studied for the path decompositions as tree-length, tree-breadth and $k$-good tree decompositions are defined and studied for the tree decompositions.
The path-length and path-breadth of $G$ are respectively denoted by $pl(G)$ and $pb(G)$.

\medskip
It has been observed in~\cite{Dragan2014a,Dragan2014b} that the four parameters tree-length, tree-breadth and path-length, path-breadth are contraction-closed invariants.
We will use the latter property in our proofs.

\begin{lemma}[~\cite{Dragan2014a,Dragan2014b}]
\label{lem:contraction-closed}
For every $G = (V,E)$ and for any edge $e \in E$:
$$tl(G/e) \leq tl(G), \ tb(G /e) \leq tb(G) \text{ and } pl(G/e) \leq pl(G), \ pb(G/e) \leq pb(G).$$
\end{lemma} 

Furthermore, it can be observed that for any graph $G$, $tb(G)  \leq tl(G) \leq 2 \cdot tb(G)$ and similarly $pb(G) \leq pl(G) \leq 2 \cdot pb(G)$.
Moreover, if a graph $G$ admits a $k$-good tree decomposition, then $tb(G) \leq \lfloor k/2 \rfloor + 1$ and $tl(G) \leq k + 1$.
Before we end this section, let us prove the stronger equivalence, $tb(G) = 1$ if and only if $G$ admits a $1$-good tree decomposition. 
This result will be of importance in the following.
Since a tree decomposition is $1$-good if and only if each bag contains a spanning star, we will name the $1$-good tree decompositions \emph{star-decompositions} in the following. 

\begin{definition}\label{def:star-tree-dec}
Let $G=(V,E)$ be a connected graph, a {\it star-decomposition} is a tree decomposition $(T,{\cal X})$ of $G$ whose each bag induces a subgraph of dominating number one, i.e., for any $t \in V(T)$, $\gamma(G[X_t])=1$.
\end{definition}

Clearly, if a graph admits a star-decomposition, then it has tree-breadth at most one.
Let us prove that the converse also holds.

\begin{lemma}\label{lem:dominator-in-bag}
For any graph $G$ with $tb(G) \leq 1$, every \emph{reduced} tree decomposition of $G$ of breadth one is a star-decomposition.
In particular:
\begin{itemize}
	\item any tree decomposition of $G$ of breadth one can be transformed into a star-decomposition in polynomial-time;
	\item similarly, any path decomposition of $G$ of breadth one can be transformed into a $1$-good path decomposition in polynomial-time.
\end{itemize}
\end{lemma}
\begin{proof}
Let $(T,{\cal X})$ be any reduced tree decomposition of $G$ of breadth one.
We will prove it is a star-decomposition.
To prove it, let $X_t \in {\cal X}$ be arbitrary and let $v \in V$ be such that $\max_{w\in X_t} dist_G(v,w) = 1$, which exists because $X_t$ has radius one.
We now show that $v \in X_t$.
Indeed, since the subtree $T_v$ and the subtrees $T_w, w \in X_t$, pairwise intersect, then it comes by the Helly Property (Lemma~\ref{lem:helly}) that $T_v \cap \left(\bigcap_{w \in X_t} T_w\right) \neq \emptyset$ {\it i.e.}, there is some bag containing $\{v\} \cup X_t$.
As a result, we have that $v \in X_t$ because $(T,{\cal X})$  is a reduced tree decomposition, hence $\gamma(G[X_t]) = 1$.
The latter implies that $(T,{\cal X})$ is a star-decomposition because $X_t$ is arbitrary. 

Now let $(T,{\cal X})$ be any tree decomposition of $G$ of breadth one.
It can be transformed in polynomial-time into a reduced tree decomposition $(T',{\cal X}')$ so that ${\cal X}' \subseteq {\cal X}$.
Furthermore, $(T',{\cal X}')$ has breadth one because it is the case for $(T,{\cal X})$, therefore $(T',{\cal X}')$ is a star-decomposition.
In particular, if $(T,{\cal X})$ is a path decomposition then so is $(T',{\cal X}')$.
\end{proof}

\begin{corollary}\label{cor:equiv-tb-star}
For any graph $G$, $tb(G) \leq 1$ if and only if $G$ admits a star-decomposition.
\end{corollary}

\section{Intractability results}
\label{sec:np-hard}

\subsection{Path-length and path-breadth}
\label{sec:plpb}

This section is devoted to the complexity of all path-like invariants that we consider in this paper.

\begin{theorem}\label{thm:pl-npc}
Deciding whether a graph has path-length at most $k$ is NP-complete even if $k=2$.
\end{theorem}

In contrast to Theorem~\ref{thm:pl-npc}, graphs with path-length one are exactly the interval graphs~\cite{Dragan2014b}, {\it i.e.}, they can be recognized in linear-time.

\begin{theorem}\label{thm:pb-npc}
Deciding whether a graph has path-breadth at most $k$ is NP-complete even if $k=1$.
\end{theorem}

From the complexity result of Theorem~\ref{thm:pb-npc}, we will also prove the hardness of deciding on the existence of $k$-good path decompositions.  

\begin{theorem}\label{thm:k-good-path-npc}
Deciding whether a graph admits a $k$-good path decomposition is NP-complete even if $k=1$.
\end{theorem}

\begin{proof}
The problem is in NP.
%(see also Appendix~\ref{sec:k-good}).
By Lemma~\ref{lem:dominator-in-bag}, a graph $G$ admits a $1$-good path decomposition if and only if $pb(G) \leq 1$, therefore it is NP-hard to decide whether a graph admits a $1$-good path decomposition by Theorem~\ref{thm:pb-npc}.
\end{proof}

All of the NP-hardness proofs in this section will rely upon the same reduction from the Betweenness problem, defined below.
The Betweenness problem, sometimes called the Total Ordering problem, is NP-complete~\cite{Opatrny1979}.
In~\cite{Golumbic1994}, it was used to show that the Interval Sandwich problem is NP-complete.
What we here prove is that the Interval Sandwich problem remains NP-complete even if the second graph is a \emph{power} of the first one, where the $k^{\text{th}}$ power $G^k$ of any graph $G$ is obtained from $G$ by adding an edge between every two distinct vertices that are at distance at most $k$ in $G$ for every integer $k \geq 1$.
Indeed, a graph $G$ has path-length at most $k$ if and only if there is an Interval Sandwich between $G$ and $G^k$ (we refer to~\cite{Lokshtanov2010} for the proof of a similar equivalence between tree-length and the Chordal Sandwich problem).

\begin{center}
    \fbox{\begin{minipage}{.95\linewidth}\label{prob:betweenness}
        \begin{problem}[Betweenness]\ 
          \begin{description}
          \item[Input:] a set ${\cal S}$ of $n$ elements, a set ${\cal T}$ of $m$ ordered triples of elements in ${\cal S}$.
          \item[Question:] Is there a total ordering of ${\cal S}$ such that for every triple $t =(s_i,s_j,s_k) \in {\cal T}$, \\either $s_i < s_j < s_k$ or $s_k < s_j < s_i$ ? 
          \end{description}
        \end{problem}     
      \end{minipage}}
  \end{center}

Now, given an instance $({\cal S},{\cal T})$ of the Betweenness problem, we will construct from ${\cal S}$ and ${\cal T}$ a graph $G_{{\cal S},{\cal T}}$ as defined below.
We will then prove that $pl(G_{{\cal S},{\cal T}}) \leq 2$ (resp. $pb(G_{{\cal S},{\cal T}}) \leq 1$) if and only if $({\cal S},{\cal T})$ is a yes-instance of the Betweenness problem.

\begin{definition}\label{def:plpb-reduction}
Let ${\cal S}$ be a set of $n$ elements, let ${\cal T}$ be a set of $m$ ordered triples of elements in ${\cal S}$.
The graph $G_{{\cal S},{\cal T}}$ is constructed as follows:
\begin{itemize}
\item For every element $s_i \in {\cal S}$, $1 \leq i \leq n$, there are two adjacent vertices $u_i,v_i$ in $G_{{\cal S},{\cal T}}$.
The vertices $u_i$ are pairwise adjacent {\it i.e.}, the set $U = \{ u_i \mid 1 \leq i \leq n \}$ is a clique.
\item For every triple $t =(s_i,s_j,s_k) \in {\cal T}$, let us add in $G_{{\cal S},{\cal T}}$ the $v_iv_j$-path $(v_i,a_t,b_t,v_j)$ of length $3$, and the $v_{j}v_{k}$-path $(v_j,c_t,d_t,v_k)$ of length $3$.
\item Finally, for every triple $t =(s_i,s_j,s_k) \in {\cal T}$ let us make adjacent $a_t,b_t$ with every $u_l$ such that $l \neq k$, similarly let us make adjacent $c_t,b_t$ with every $u_l$ such that $l \neq i$. 
\end{itemize}
\end{definition}

It can be noticed from Definition~\ref{def:plpb-reduction} that for any $1 \leq i \leq n$, the vertex $u_i$ is adjacent to any vertex but those $v_j$ such that $j \neq i$, those $a_t,b_t$ such that $s_i$ is the last element of triple $t$ and those $c_t,b_t$ such that $s_i$ is the first element of triple $t$. 
We refer to Figure~\ref{fig:plpb} for an illustration (see also Figure~\ref{fig:adj}).
Observe that $G_{{\cal S},{\cal T}}$ has diameter $3$ because the clique $U$ dominates $G_{{\cal S},{\cal T}}$, therefore $pl(G_{{\cal S},{\cal T}}) \leq 3$ and we will show that it is hard to distinguish graphs with path-length two from graphs with path-length three.
Similarly, the clique $U$ dominates $G_{{\cal S},{\cal T}}$ hence $pb(G_{{\cal S},{\cal T}}) \leq 2$, thus we will show that it is hard to distinguish graphs with path-breadth one from graphs with path-breadth two. 

\begin{figure}[h]
\begin{minipage}[h]{.46\linewidth}
  \centering
  \includegraphics[width=.75\textwidth]{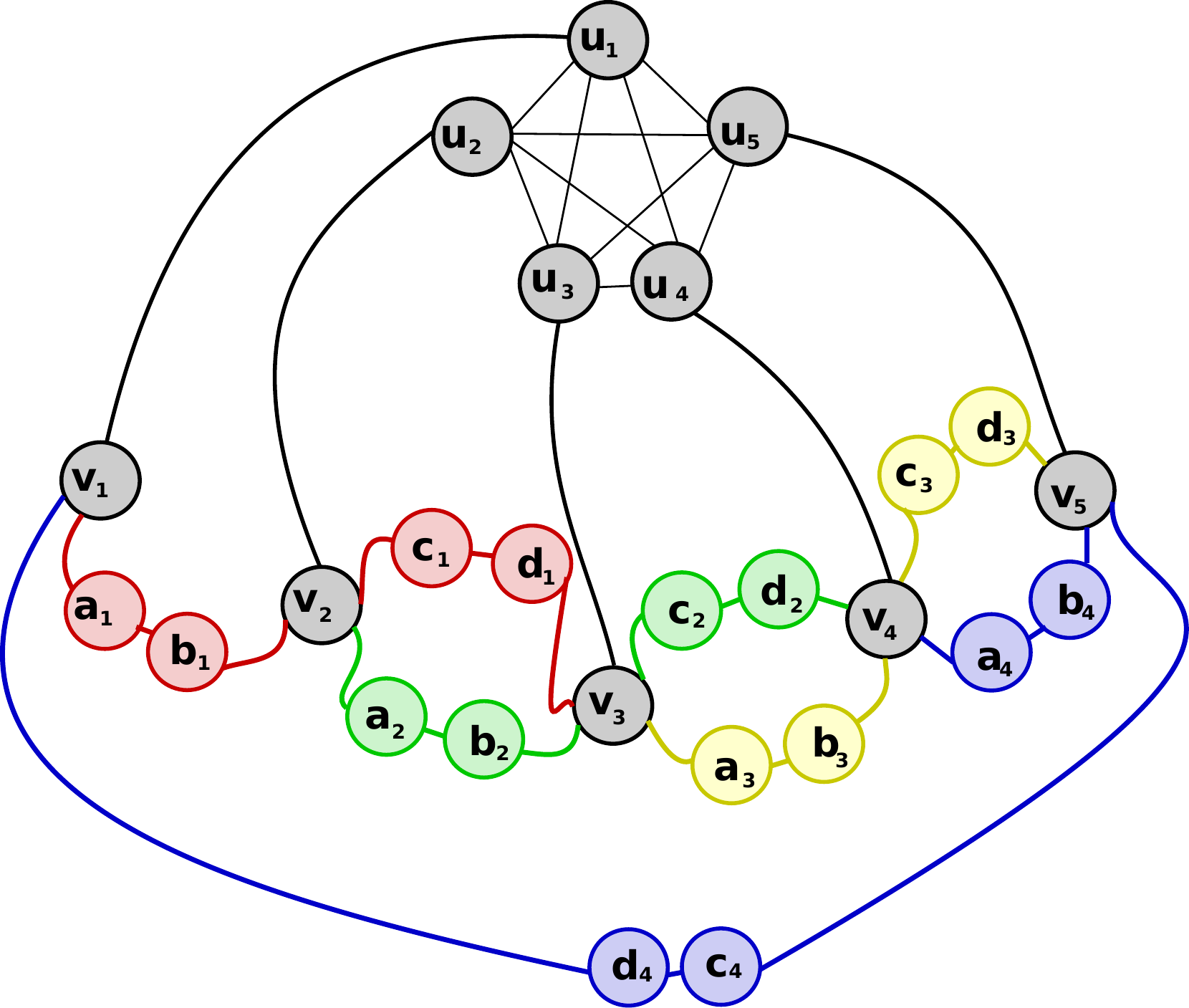}
  \caption{The graph $G_{{\cal S},{\cal T}}$ for ${\cal S}=[|1,5|]$ and ${\cal T}=\{ (i,i+1,i+2) \mid 1 \leq i \leq 4\}$. Each colour corresponds to a given triple of ${\cal T}$. For ease of reading, the adjacency relations between the vertices $u_i$ and the colored vertices $a_t,b_t,c_t,d_t$ are not drawn. 
%In particular, vertex $u_1$ is non-adjacent only to $x_{41},x_{42}$ and $y_{11},y_{12}$.
%Similarly, vertex $u_2$ is non-adjacent only to $y_{21},y_{22}$.
%Vertex $u_3$ is non-adjacent only to $x_{11},x_{12}$ and $y_{31},y_{32}$. 
%Vertex $u_4$ is non-adjacent only to $x_{21},x_{22}$ and $y_{41},y_{42}$.
%Last, vertex $u_5$ is non-adjacent only to $x_{31},x_{32}$.
}
  \label{fig:plpb}
\end{minipage} \hfill
\begin{minipage}[c]{.46\linewidth}
      \centering
      \includegraphics[width=0.55\textwidth]{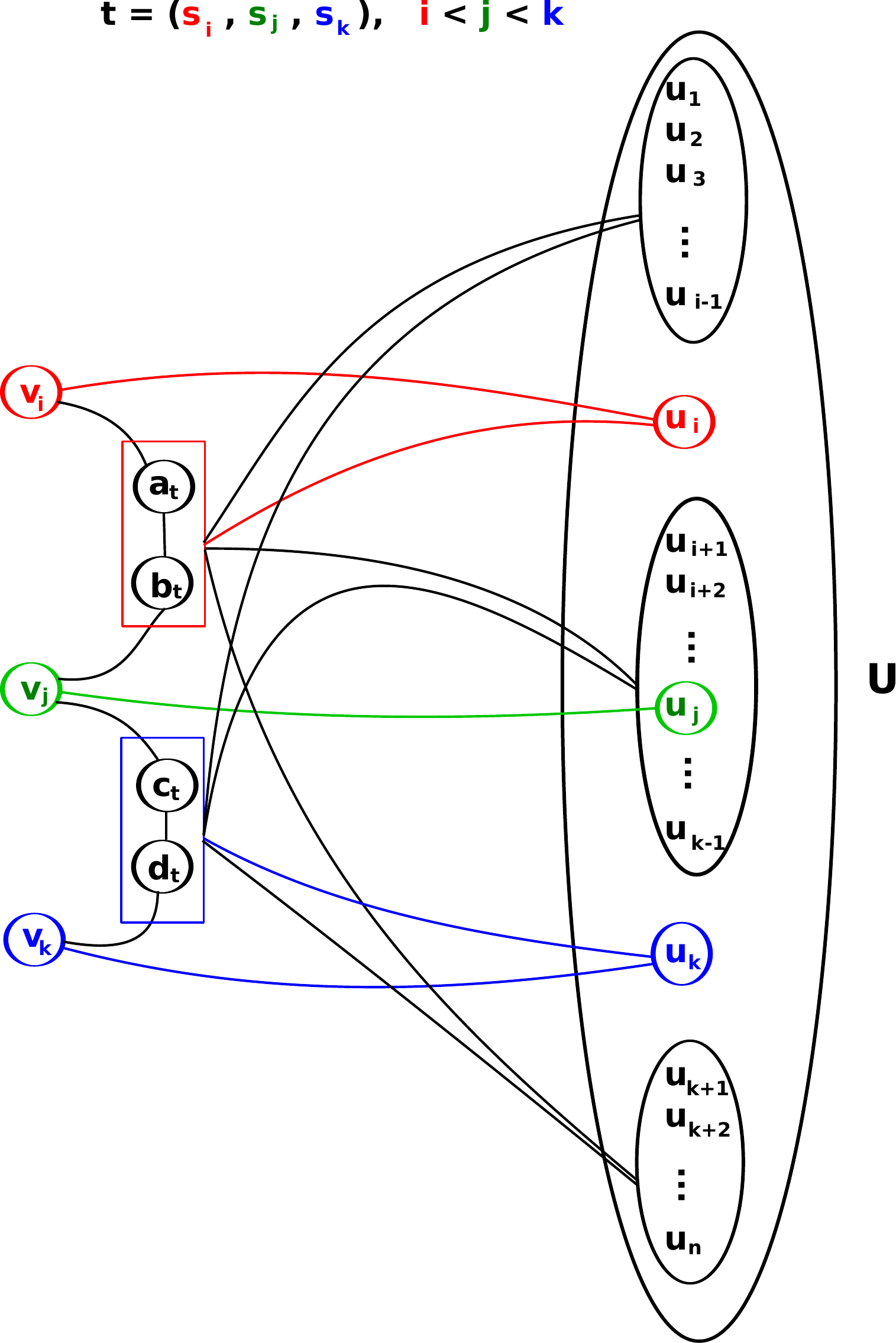}
      \caption{Adjacency relations in $G_{{\cal S},{\cal T}}$ for one given triple $t=(s_i,s_j,s_k)$.}
      \label{fig:adj}
\end{minipage}
\end{figure}

\begin{lemma}
\label{lem:betweenness-implies}
Let ${\cal S}$ be a set of $n$ elements, let ${\cal T}$ be a set of $m$ ordered triples of elements in ${\cal S}$.
If $({\cal S},{\cal T})$ is a yes-instance of the Betweenness problem then $pb(G_{{\cal S},{\cal T}}) \leq 1$ and $pl(G_{{\cal S},{\cal T}}) \leq 2$, where $G_{{\cal S},{\cal T}}$ is the graph that is defined in Definition~\ref{def:plpb-reduction}.
\end{lemma}
\begin{proof}
Since $pl(G_{{\cal S},{\cal T}}) \leq 2 \cdot pb(G_{{\cal S},{\cal T}})$ then we only need to prove that $pb(G_{{\cal S},{\cal T}}) \leq 1$.
For convenience, let us reorder the elements of ${\cal S}$ so that for every triple $(s_i,s_j,s_k) \in {\cal T}$ either $i < j < k$ or $k < j < i$.
It is possible to do that because by the hypothesis $({\cal S},{\cal T})$ is a yes-instance of the Betweenness problem.
If furthermore $k < j < i$, let us also replace $(s_i,s_j,s_k)$ with the inverse triple $(s_k,s_j,s_i)$.
This way, we have a total ordering of ${\cal S}$ such that $s_i < s_j < s_k$ for every triple $(s_i,s_j,s_k) \in {\cal T}$.
Then, let us construct a path decomposition $(P,{\cal X})$ with $n$ bags, denoted $X_1,X_2, \ldots, X_n$, as follows. 
For every $1 \leq i \leq n$, $U \subseteq X_i$ and $v_i \in X_i$.
For every $t = (s_i,s_j,s_k) \in {\cal T}$, we add both $a_t,b_t$ into the bags $X_l$ with $i \leq l \leq j$, similarly we add both $c_t,d_t$ into the bags $X_l$ with $j \leq l \leq k$.
By construction, the clique $U$ is contained in any bag of $P$ and for every triple $t=(s_i,s_j,s_k) \in {\cal T}$ we have $a_t,b_t,v_i \in X_i$ and $a_t,b_t,c_t,d_t,v_j \in X_j$ and $c_t,d_t,v_k \in X_k$, therefore $(P,{\cal X})$ is indeed a path decomposition of $G_{{\cal S},{\cal T}}$.

We claim that for every $i$, $X_i \subseteq N[u_i]$, that will prove the lemma.
Indeed if it were not the case for some $i$ then by Definition~\ref{def:plpb-reduction} there should exist $t \in {\cal T},j,k$ such that: either $t = (s_i,s_j,s_k) \in {\cal T}$ and $c_t,d_t \in X_i$; or $t = (s_k,s_j,s_i) \in {\cal T}$ and $a_t,b_t \in X_i$.
But then by construction either $a_t,b_t$ are only contained in the bags $X_l$ for $k \leq l \leq j$, or $c_t,d_t$ are only contained in the bags $X_l$ for $j \leq l \leq k$, thus contradicting the fact that either $a_t,b_t \in X_i$ or $c_t,d_t \in X_i$.
\end{proof}

\begin{lemma}
\label{lem:betweenness-implied}
Let ${\cal S}$ be a set of $n$ elements, let ${\cal T}$ be a set of $m$ ordered triples of elements in ${\cal S}$.
If $pb(G_{{\cal S},{\cal T}}) \leq 1$ or $pl(G_{{\cal S},{\cal T}}) \leq 2$ then $({\cal S},{\cal T})$ is a yes-instance of the Betweenness problem, where $G_{{\cal S},{\cal T}}$ is the graph that is defined in Definition~\ref{def:plpb-reduction}.
\end{lemma}
\begin{proof}
Since $pl(G_{{\cal S},{\cal T}}) \leq 2 \cdot pb(G_{{\cal S},{\cal T}})$ then we only need to consider the case when $pl(G_{{\cal S},{\cal T}}) \leq 2$.
Let $(P,{\cal X})$ be a path decomposition of length two, that exists by the hypothesis.
Since the vertices $v_i$ are pairwise at distance $3$ then the subpaths $P_{v_{i}}$ that are induced by the bags containing vertex $v_i$ are pairwise disjoint.
Therefore, starting from an arbitrary endpoint of $P$ and considering each vertex $v_i$ in the order that it appears in the path decomposition, this defines a total ordering over ${\cal S}$.
Let us reorder the set ${\cal S}$ so that vertex $v_i$ is the $i^{\text{th}}$ vertex to appear in the path-decomposition.
We claim that for every triple $t = (s_i,s_j,s_k) \in {\cal T}$, either $i < j < k$ or $k < j < i$, that will prove the lemma.

By way of contradiction, let $t = (s_i,s_j,s_k) \in {\cal T}$  such that either $j < \min \{i,k\}$ or $j > \max \{i,k\}$.
By symmetry, we only need to consider the case when $j < i < k$.
In such case by construction the path between $P_{v_j}$ and $P_{v_k}$ in $P$ contains $P_{v_i}$.
Let $B \in P_{v_i}$, by the properties of a tree decomposition it is a $v_jv_k$-separator, so it must contain one of $c_t,d_t$.
However, vertex $v_i \in B$ is at distance $3$ from both vertices $c_t,d_t$, thus contradicting the fact that $(P,{\cal X})$ has length $2$.
\end{proof}

We are now able to prove Theorems~\ref{thm:pl-npc} and~\ref{thm:pb-npc}.

\begin{proofof}{Theorem~\ref{thm:pl-npc}}
To prove that a graph $G$ satisfies $pl(G) \leq k$, it suffices to give as a certificate a tree decomposition of $G$ with length at most $k$.
Indeed, the all-pairs-shortest-paths in $G$ can be computed in polynomial-time.
Therefore, the problem of deciding whether a graph has path-length at most $k$ is in NP.
Given an instance $({\cal S},{\cal T})$ of the Betweenness problem, let $G_{{\cal S},{\cal T}}$ be as defined in Definition~\ref{def:plpb-reduction}.
We claim that $pl(G_{{\cal S},{\cal T}}) \leq 2$ if and only if the pair $({\cal S},{\cal T})$ is a yes-instance of the Betweenness problem.
This will prove the NP-hardness because our reduction is polynomial and the Betweenness problem is NP-complete.
To prove the claim in one direction, if $({\cal S},{\cal T})$ is a yes-instance then by Lemma~\ref{lem:betweenness-implies} $pl(G_{{\cal S},{\cal T}}) \leq 2 $.
Conversely, if $pl(G_{{\cal S},{\cal T}}) \leq 2$ then $({\cal S},{\cal T})$ is a yes-instance by Lemma~\ref{lem:betweenness-implied}, that proves the claim in the other direction.
\end{proofof}

\begin{proofof}{Theorem~\ref{thm:pb-npc}}
To prove that a graph $G$ satisfies $pb(G) \leq k$, it suffices to give as a certificate a tree decomposition of $G$ with breadth at most $k$.
Indeed, the all-pairs-shortest-paths in $G$ can be computed in polynomial-time.
Therefore, the problem of deciding whether a graph has path-breadth at most $k$ is in NP.
Given an instance $({\cal S},{\cal T})$ of the Betweenness problem, let $G_{{\cal S},{\cal T}}$ be as defined in Definition~\ref{def:plpb-reduction}.
We claim that $pb(G_{{\cal S},{\cal T}}) \leq 1$ if and only if the pair $({\cal S},{\cal T})$ is a yes-instance of the Betweenness problem.
This will prove the NP-hardness because our reduction is polynomial and the Betweenness problem is NP-complete.
To prove the claim in one direction, if $({\cal S},{\cal T})$ is a yes-instance then by Lemma~\ref{lem:betweenness-implies} $pb(G_{{\cal S},{\cal T}}) \leq 1$.
Conversely, if $pb(G_{{\cal S},{\cal T}}) \leq 1$ then $({\cal S},{\cal T})$ is a yes-instance by Lemma~\ref{lem:betweenness-implied}, that proves the claim in the other direction.
\end{proofof}

To conclude this section, we strenghten the above hardness results with two inapproximability results.
Indeed, it has to be noticed that for any graph parameter $param$, an $\alpha$-approximation algorithm for $param$ with $\alpha < 1 + \frac 1 k$ is enough to separate the graphs $G$ such that $param(G) \leq k$ from those such that $param(G) \geq k+1$.
Therefore, the two following corollaries follow from our polynomial-time reduction.

\begin{corollary}
\label{cor:tl-inapprox}
For every $\varepsilon > 0$, the path-length of a graph cannot be approximated within a factor $\frac 3 2 - \varepsilon$ unless P=NP.
\end{corollary}

\begin{proof}
Let $G_{{\cal S},{\cal T}}$ be the graph of the reduction in Theorem~\ref{thm:pl-npc}.
By Definition~\ref{def:plpb-reduction}, it has diameter at most $3$ and so $pl(G_{{\cal S},{\cal T}}) \leq 3$.
Since it is NP-hard to decide whether $pl(G_{{\cal S},{\cal T}}) \leq 2$, therefore it does not exist a  $\left(\frac 3 2 - \varepsilon\right)$-approximation algorithm for path-length unless P=NP.
\end{proof}

\begin{corollary}
\label{cor:tb-inapprox}
For every $\varepsilon > 0$, the path-breadth of a graph cannot be approximated within a factor $2 - \varepsilon$ unless P=NP.
\end{corollary}

\begin{proof}
Let $G_{{\cal S},{\cal T}}$ be the graph of the reduction in Theorem~\ref{thm:pb-npc}.
By Definition~\ref{def:plpb-reduction}, the set $U$ is a dominating clique and so $pb(G_{{\cal S},{\cal T}}) \leq 2$.
Since it is NP-hard to decide whether $pb(G_{{\cal S},{\cal T}}) \leq 1$, therefore it does not exist a  $\left(2 - \varepsilon\right)$-approximation algorithm for path-breadth unless P=NP.
\end{proof}

So far, there exists a $2$-approximation algorithm for path-length and a $3$-approximation algorithm for path-breadth~\cite{Dragan2014b}.
Therefore, we let open whether there exist $\frac 3 2$-approximation algorithms for path-length and $2$-approximation algorithms for path-breadth. 

\subsection{Tree-breadth}
\label{sec:tb-npc}

We prove next that computing the tree-breadth is NP-hard.

\begin{theorem}\label{thm:tb-npc}
	Deciding whether a graph has tree-breadth at most $k$ is NP-complete even if $k=1$.
\end{theorem}

\begin{theorem}\label{thm:k-good-tree-npc}
	Deciding whether a graph admits a $k$-good tree decomposition is NP-complete even if $k=1$.
\end{theorem}

\begin{proof}
	The problem is in NP.
	%(see also Appendix~\ref{sec:k-good}).
	By Corollary~\ref{cor:equiv-tb-star}, a graph $G$ admits a star-decomposition if and only if $tb(G) \leq 1$, therefore it is NP-hard to decide whether a graph admits a $1$-good path decomposition by Theorem~\ref{thm:tb-npc}.
\end{proof}

In order to prove Theorem~\ref{thm:tb-npc}, we will reduce from the Chordal Sandwich problem (defined below).
In~\cite{Lokshtanov2010}, the author also proposed a reduction from the Chordal Sandwich problem in order to prove that computing tree-length is NP-hard.
However, we will need different gadgets than in~\cite{Lokshtanov2010}, and we will need different arguments to prove correctness of the reduction.

\begin{center}
	\fbox{\begin{minipage}{.95\linewidth}\label{prob:chordal-sandwich}
			\begin{problem}[Chordal Sandwich]\ 
				\begin{description}
					\item[Input:] graphs $G_1=(V,E_1)$ and $G_2=(V,E_2)$ such that $E_1 \subseteq E_2$.
					\item[Question:] Is there a chordal graph $H=(V,E)$ such that $E_1 \subseteq E \subseteq E_2$ ?
				\end{description}
			\end{problem}     
		\end{minipage}}
	\end{center}
	
The Chordal Sandwich problem is NP-complete even when the $2n = |V|$ vertices induce a perfect matching in $\bar{G_2}$ (the complementary of $G_2$)~\cite{Bodlaender1992,Golumbic1995}.
Perhaps surprisingly, the later constriction on the structure of $\bar{G_2}$ is a key element in our reduction.
Indeed, we will need the following technical lemma.

\begin{lemma}
\label{lem:chordal-sandwich}
Let $G_1 =(V,E_1)$, $G_2=(V,E_2)$ such that $E_1 \subseteq E_2$ and $\bar{G_2}$ (the complementary of $G_2$) is a perfect matching.
Suppose that $\langle G_1, G_2 \rangle$ is a yes-instance of the Chordal Sandwich problem.

Then, there exists a reduced tree decomposition $(T,{\cal X})$ of $G_1$ such that for every forbidden edge $\{u,v\} \notin E_2$: $T_u \cap T_v = \emptyset$, $T_u \cup T_v = T$, furthermore there are two adjacent bags $B_u \in T_u$ and $B_v \in T_v$ such that $B_u \setminus u = B_v \setminus v$.
\end{lemma}

\begin{proof}
Let $H=(V,E)$ be any chordal graph such that $E_1 \subseteq E \subseteq E_2$ (that exists because $\langle G_1, G_2 \rangle$ is a yes-instance of the Chordal Sandwich problem by the hypothesis) and the number $|E|$ of edges is maximized.
We will prove that any clique-tree $(T,{\cal X})$ of $H$ satisfies the above properties (given in the statement of the lemma).
To prove it, let $\{u,v\} \notin E_2$ be arbitrary.
Observe that $T_u \cap T_v = \emptyset$ (else, $\{u,v\} \in E$, that would contradict that $E \subseteq E_2$).

Furthermore, let $B_u \in T_u$ minimize the distance in $T$ to the subtree $T_v$, let $B$ be the unique bag that is adjacent to $B_u$ on a shortest-path between $B_u$ and $T_v$ in $T$.
Note that $B \notin T_u$ by the minimality of $dist_T(B_u,T_v)$, however $B$ may belong to $T_v$.
Removing the edge $\{B_u,B\}$ in $T$ yields two subtrees $T_1,T_2$ with $T_u \subseteq T_1$ and $T_v \subseteq T_2$.
In addition, we have that for every $x \in V \setminus u$ such that $T_x \cap T_1 \neq \emptyset$, $\{u,x\} \in E_2$ since $x \neq v$ and $\bar{G_2}$ is a perfect matching by the hypothesis.
Similarly, we have that for every $y \in V \setminus v$ such that $T_y \cap T_2 \neq \emptyset$, $\{v,y\} \in E_2$.
Therefore, by maximality of the number $|E|$ of edges, it follows that $T_1 = T_u$ and $T_2 = T_v$, and so, $T_u \cup T_v = T$.
In particular, $B=B_v \in T_v$.

Finally, let us prove that $B_u \setminus u = B_v \setminus v$.
Indeed, assume for the sake of contradiction that $B_u \setminus u \neq B_v \setminus v$.
In particular, $(B_u \setminus B_v) \setminus u \neq \emptyset$ or $(B_v \setminus B_u) \setminus v \neq \emptyset$.
Suppose w.l.o.g. that $(B_u \setminus B_v) \setminus u \neq \emptyset$.
Let $H'=(V,E')$ be obtained from $H$ by adding an edge between vertex $v$ and every vertex of $(B_u \setminus B_v) \setminus u$.
By construction $|E'| > |E|$.
Furthermore, $H'$ is chordal since a clique-tree of $H'$ can be obtained from $(T,{\cal X})$ by adding a new bag $(B_u \setminus u) \cup \{v\}$ in-between $B_u$ and $B_v$.
However, for every $x \in (B_u \setminus B_v) \setminus u$ we have that $\{x,v\} \in E_2$ since $x \neq u$ and $\bar{G_2}$ is a perfect matching by the hypothesis.
As a result, $E' \subseteq E_2$, thus contradicting the maximality of the number $|E|$ of edges in $H$.       
\end{proof}

\begin{proofof}{Theorem~\ref{thm:tb-npc}}
	The problem is in NP.
	To prove the NP-hardness,let $\langle G_1, G_2 \rangle$ be any input of the Chordal Sandwich problem such that  $\bar{G_2}$ is a perfect matching.
	The graph $G'$ is constructed from $G_1$ as follows.
	First we add a clique $V'$ of $2n = |V|$ vertices in $G'$.
	Vertices $v \in V$ are in one-to-one correspondance with vertices $v' \in V'$.
	Then, for every forbidden edge $\{u,v\} \notin E_2$, vertices $u,v$ are respectively made adjacent to all vertices in $V' \setminus v'$ and $V' \setminus u'$.
	Finally, we add a distinct copy of the gadget $F_{uv}$ in Figure~\ref{fig:gadget}, and we make adjacent $s_{uv}$ and $t_{uv}$ to the two vertices $u',v'$ (see also Figure~\ref{fig:reduction-tb} for an illustration).
	We will prove $tb(G') = 1$ if and only if $\langle G_1, G_2 \rangle$ is a yes-instance of the Chordal Sandwich problem.
	This will prove the NP-hardness because our reduction is polynomial and the Chordal Sandwich problem is NP-complete even when the $2n = |V|$ vertices induce a perfect matching in $\bar{G_2}$ (the complementary of $G_2$)~\cite{Bodlaender1992,Golumbic1995}.
	
	\begin{figure}[h!]
		\centering
		\includegraphics[width=0.25\textwidth]{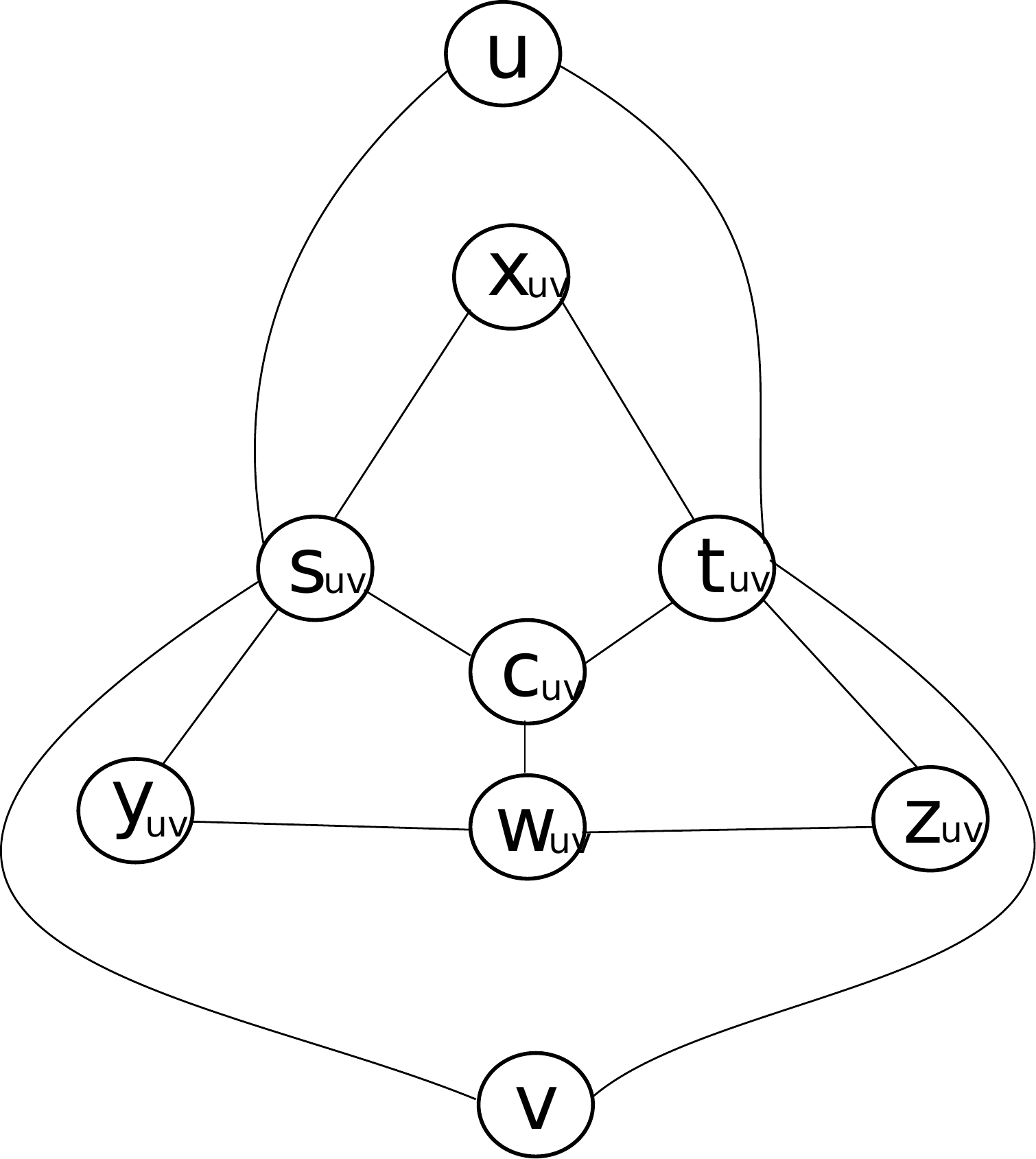}
		\caption{The gadget $F_{uv}$.}
		\label{fig:gadget}
	\end{figure}
	
	In one direction, assume $tb(G') = 1$, let $(T,{\cal X})$ be a star-decomposition of $G'$. 
	Let $H=(V, \{ \{u,v\} \mid T_u \cap T_v \neq \emptyset\})$, that is a chordal graph such that $E_1 \subseteq E(H)$.
	To prove that $\langle G_1, G_2 \rangle$ is a yes-instance of the Chordal Sandwich problem, it suffices to prove that $T_u \cap T_v = \emptyset$ for every forbidden edge $\{u,v\} \notin E_2$.
	More precisely, we will prove that $T_{s_{uv}} \cap T_{t_{uv}} \neq \emptyset$, for we claim that the latter implies $T_u \cap T_v = \emptyset$.
	Indeed, assume $T_{s_{uv}} \cap T_{t_{uv}} \neq \emptyset$ and $T_u \cap T_v \neq \emptyset$.
	Since $s_{uv}$ and $t_{uv}$ are both adjacent to $u$ and $v$, therefore the four subtrees $T_u,T_v,T_{s_{uv}},T_{t_{uv}}$ pairwise intersect.
	By the Helly property (Lemma~\ref{lem:helly}) $T_u \cap T_v \cap T_{s_{uv}} \cap T_{t_{uv}} \neq \emptyset$, hence there is a bag containing $u,v,s_{uv},t_{uv}$ but then it contradicts the fact that $(T,{\cal X})$ is a star-decomposition because no vertex dominates the four vertices.
	Therefore, $T_{s_{uv}} \cap T_{t_{uv}} \neq \emptyset$ implies $T_u \cap T_v = \emptyset$.
	Let us prove that $T_{s_{uv}} \cap T_{t_{uv}} \neq \emptyset$.
	By contradiction, assume $T_{s_{uv}} \cap T_{t_{uv}} = \emptyset$.
	Every bag $B$ onto the path between $T_{s_{uv}}$ and $T_{t_{uv}}$ must contain $c_{uv},x_{uv}$, furthermore $N[c_{uv}] \cap N[x_{uv}] = \{s_{uv}, t_{uv}\}$.
	Since, $(T,{\cal X})$ is a star-decomposition, the latter implies either $s_{uv} \in B$ and $B \subseteq N[s_{uv}]$ or $t_{uv} \in B$ and $B \subseteq N[t_{uv}]$.
	Consequently, there exist two adjacent bags $B_s \in T_{s_{uv}},B_t \in T_{t_{uv}}$ such that $B_s \subseteq N[s_{uv}]$ and $B_t \subseteq N[t_{uv}]$.
	Furthermore, $B_s \cap B_t$ is an $s_{uv}t_{uv}$-separator by the properties of a tree decomposition. 
	In particular, $B_s \cap B_t$ must intersect the path $(y_{uv},w_{uv},z_{uv})$ because $y_{uv} \in N(s_{uv})$ and $z_{uv} \in N(t_{uv})$. 
	However, $B_s \subseteq N[s_{uv}], B_t \subseteq N[t_{uv}]$ but $N[s_{uv}] \cap N[t_{uv}] \cap \{y_{uv},w_{uv},z_{uv}\} = \emptyset$, hence $B_s \cap B_t \cap \{y_{uv},w_{uv},z_{uv}\} = \emptyset$, that is a contradiction.
	As a result, $T_{s_{uv}} \cap T_{t_{uv}} \neq \emptyset$ and so, $T_u \cap T_v = \emptyset$.
	
		\begin{figure}[h!]
			\centering
			\includegraphics[width=0.45\textwidth]{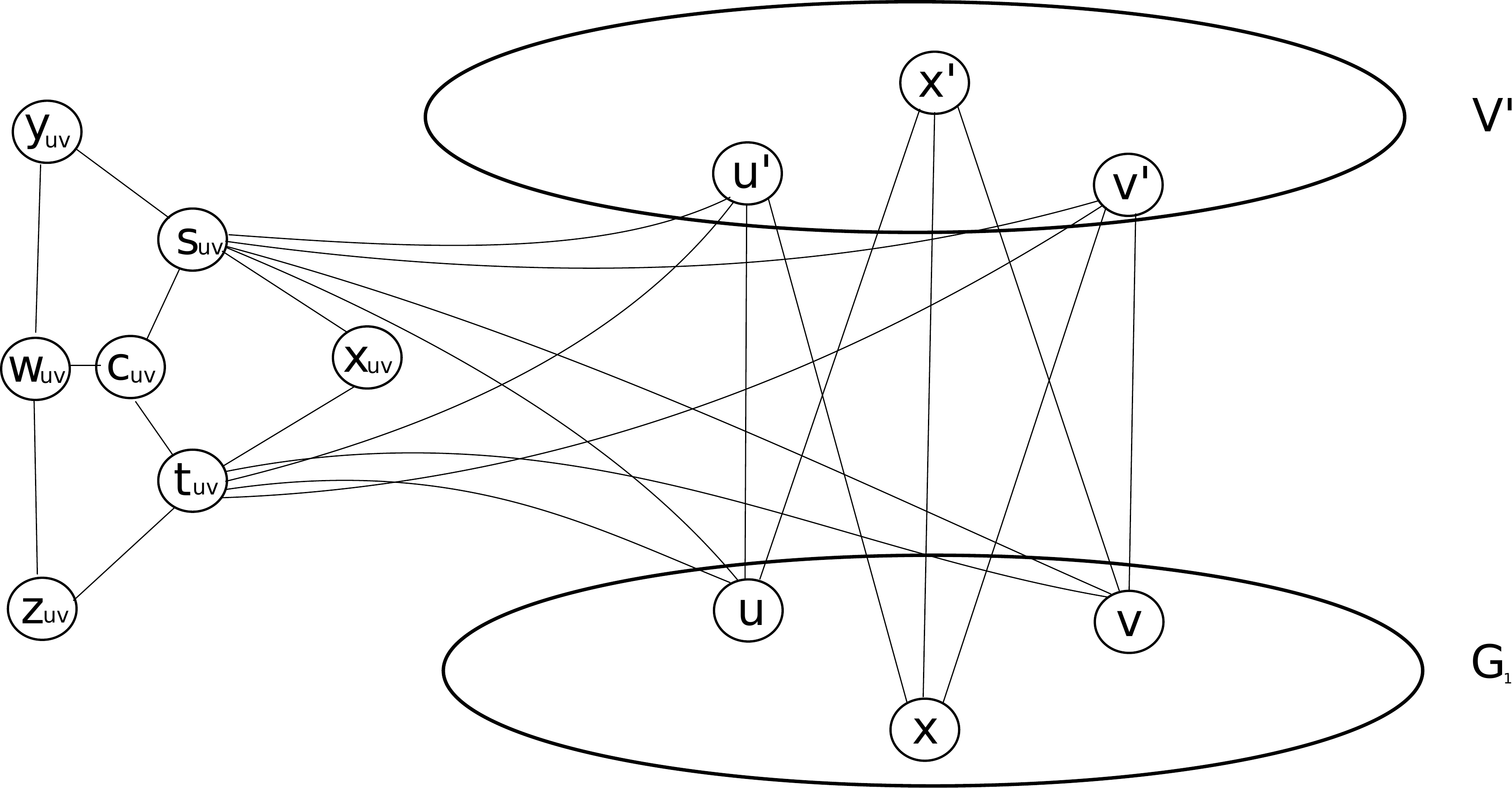}
			\caption{The graph $G'$ (simplified view).}
			\label{fig:reduction-tb}
		\end{figure}
	
	Conversely, assume that $\langle G_1, G_2 \rangle$ is a yes-instance of the Chordal Sandwich problem.
	Since $\bar{G_2}$ is a perfect matching by the hypothesis, by Lemma~\ref{lem:chordal-sandwich} there exists a reduced tree decomposition $(T,{\cal X})$ of $G_1$ such that for every forbidden edge $\{u,v\} \notin E_2$: $T_u \cap T_v = \emptyset, T_u \cup T_v = T$ and there are two adjacent bags $B_u \in T_u, B_v \in T_v$ so that $B_u \setminus u = B_v \setminus v$.
	Let us modify $(T,{\cal X})$ in order to obtain a star-decomposition of $G'$.
	
	In order to achieve the result, we first claim that for every edge $\{t,t'\} \in E(T)$, the bags $X_t,X_{t'}$ differ in exactly one vertex, that is, $|X_t \setminus X_{t'}| = 1$ and similarly $|X_{t'} \setminus X_t| = 1$. 
	Indeed, $X_t \setminus X_{t'} \neq \emptyset$ because $(T,{\cal X})$ is reduced, so, let $u_{tt'} \in X_t \setminus X_{t'}$.
	Let $v_{tt'} \in V$ be the unique vertex satisfying $\{u_{tt'},v_{tt'}\} \notin E_2$, that is well-defined because $\bar{G_2}$ is a perfect matching by the hypothesis.
	Note that $v_{tt'} \in X_{t'}$ because $u_{tt'} \notin X_{t'}$ and $T_{u_{tt'}} \cup T_{v_{tt'}} = T$.
	Furthermore, $v_{tt'} \notin X_{t}$ because $u_{tt'} \in X_t$ and $T_{u_{tt'}} \cap T_{v_{tt'}} = \emptyset$.
	By construction of $(T,{\cal X})$, there are two adjacent bags $B_{u_{tt'}} \in T_{u_{tt'}},B_{v_{tt'}} \in T_{v_{tt'}}$ such that $B_{u_{tt'}} \setminus u_{tt'} = B_{v_{tt'}} \setminus v_{tt'}$.
	Since $u_{tt'} \in X_t \setminus X_{t'}$ and $v_{tt'} \in X_{t'} \setminus X_t$, therefore, $X_t = B_{u_{tt'}}$ and $X_{t'} = B_{v_{tt'}}$, and so, $X_t \setminus X_{t'} = \{u_{tt'}\}$ and $X_{t'} \setminus X_t = \{v_{tt'}\}$. 
	In the following, we will keep the above notations $u_{tt'},v_{tt'}$ for every edge $\{t,t'\} \in E(T)$ (in particular, $u_{tt'} = v_{t't}$ and $v_{tt'} = u_{t't}$).
	
	Let us construct the star-decomposition $(T',{\cal X}')$ of $G'$ as follows.
	\begin{itemize}
		\item For every node $t \in V(T)$, let $S_t = X_t \cup V' \cup (\bigcup_{t' \in N_T(t)} \{s_{u_{tt'}v_{tt'}}, t_{u_{tt'}v_{tt'}} \})$ (in particular, $|S_t| = 2n + |X_t| + 2\cdot deg_T(t)$).
		We will first construct a path decomposition of $G'[S_t]$ whose bags are the sets $Y_{tt'} = X_t \cup V' \cup \{s_{u_{tt'}v_{tt'}}, t_{u_{tt'}v_{tt'}}\}$ for every edge $\{t,t'\} \in E(T)$ (note that the bags can be linearly ordered in an arbitrary way in the path decomposition).
		Furthermore, for every edge $\{t,t'\} \in E(T)$, $Y_{tt'} \subseteq N[u'_{tt'}]$, where $u'_{tt'} \in V'$ is the corresponding vertex to $u_{tt'} \in V$ in the clique $V'$ (see Figure~\ref{fig:reduction-tb-step1} for an illustration).
		Therefore, the above constructed path decomposition is a $1$-good path decomposition.
		
				\begin{figure}[h!]
					\centering
					\includegraphics[width=0.85\textwidth]{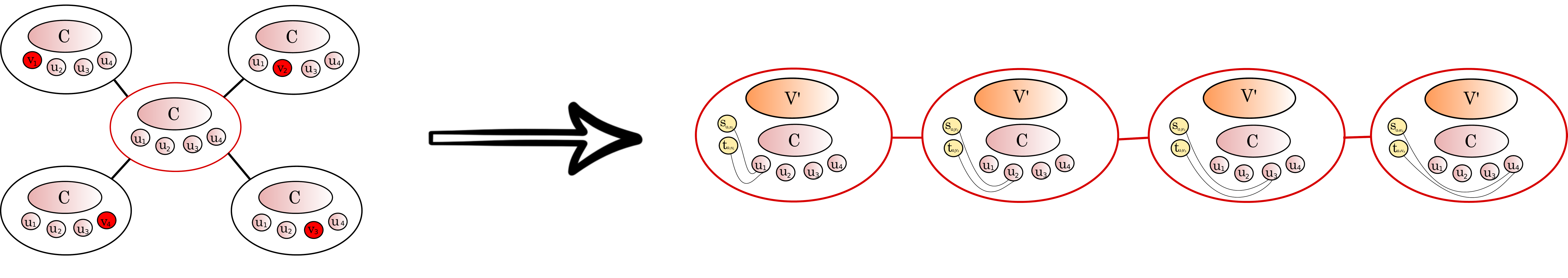}
					\caption{The $1$-good path decomposition (right) obtained from the central bag with degree four (left).}
					\label{fig:reduction-tb-step1}
				\end{figure}
		
		\item Then, we will connect the $1$-good path decompositions together.
		More precisely, let us add an edge between the two bags $Y_{tt'}$ and $Y_{t't}$ for every edge $\{t,t'\} \in E(T)$ (see Figure~\ref{fig:reduction-tb-step2} for an illustration).
		
		In so doing, we claim that one obtains a star-decomposition of $G'[\bigcup_{t \in V(T)} S_t]$.
		Indeed, it is a tree decomposition since: 
		\begin{itemize}
			\item the clique $V'$ is contained in all bags;
			\item for every $\{t,t'\} \in E(T)$ the two vertices $s_{u_{tt'}v_{tt'}}, t_{u_{tt'}v_{tt'}}$ are only contained in the two adjacent bags $Y_{tt'}$ and $Y_{t't}$, furthermore  $u_{tt'},u'_{tt'},v'_{tt'} \in Y_{tt'}$ and $v_{tt'},u'_{tt'},v'_{tt'} \in Y_{t't}$;
			\item last, each vertex $v \in V$ is contained in $\{ Y_{tt'} \mid v \in X_t \text{ and } t' \in N_T(t)\}$ which induces a subtree since $(T,{\cal X})$ is a tree decomposition of $G_1$. 
		\end{itemize}
		Since in addition every bag $Y_{tt'}$, with $\{t,t'\} \in E(T)$, is dominated by $u'_{tt'} \in V'$, this proves the claim that one obtains a star-decomposition.
		
				\begin{figure}[h!]
					\centering
					\includegraphics[width=0.55\textwidth]{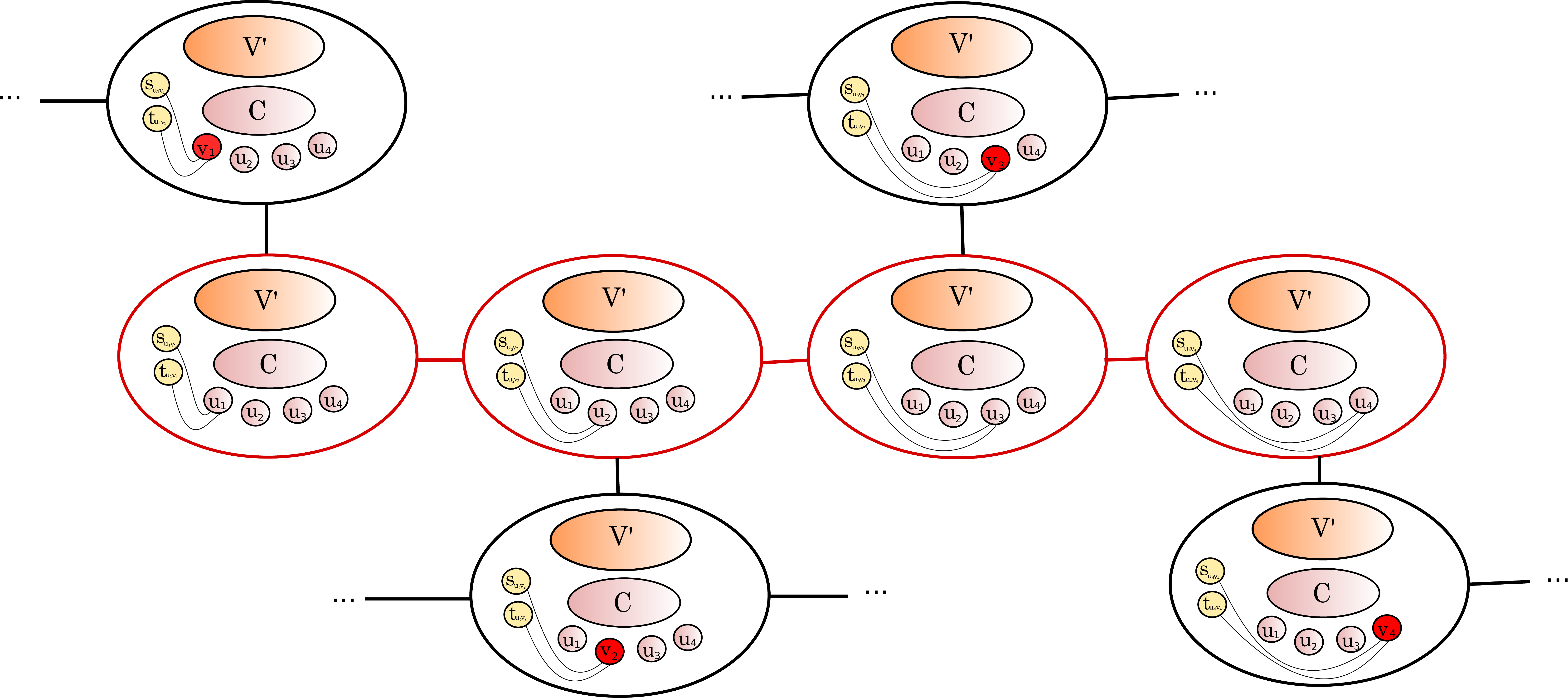}
					\caption{Connection of the $1$-good path decomposition in Figure~\ref{fig:reduction-tb-step1} to the neighbouring $1$-good path decompositions.}
					\label{fig:reduction-tb-step2}
				\end{figure}
		
		\item In order to complete the construction, let us observe that for every forbidden edge $\{u,v\} \notin E_2$, there is a star-decomposition of $F_{uv} \setminus \{u,v\}$ with three leaf-bags $\{x_{uv}, s_{uv}, t_{uv}\}, \ \{y_{uv}, s_{uv}, w_{uv}\}$ and $\{z_{uv}, s_{uv}, w_{uv}\}$ and one internal bag of degree three $B_{uv} = \{c_{uv}, s_{uv}, t_{uv}, w_{uv}\}$.
		For every $\{t,t'\} \in E(T)$, we simply connect the above star-decomposition of $F_{u_{tt'}v_{tt'}} \setminus \{u_{tt'},v_{tt'}\}$ by making the internal bag $B_{u_{tt'}v_{tt'}}$ adjacent to one of $Y_{tt'}$ or $Y_{t't}$ (see Figure~\ref{fig:reduction-tb-step3} for an illustration). 
	\end{itemize}
	
				\begin{figure}[h!]
					\centering
					\includegraphics[width=0.55\textwidth]{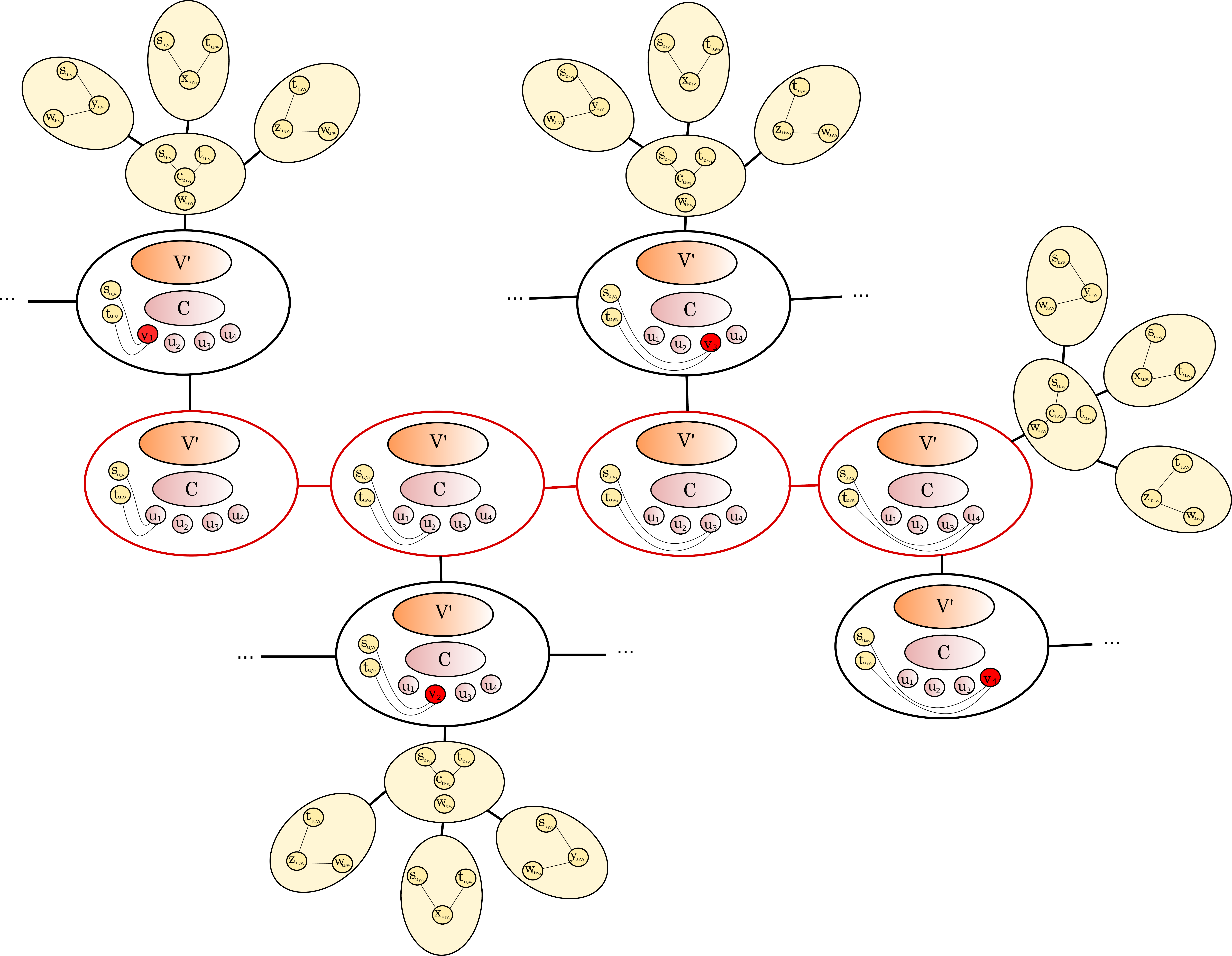}
					\caption{The respective star-decompositions of the gadgets $F_{u_iv_i}$ are connected to the other bags.}
					\label{fig:reduction-tb-step3}
				\end{figure}
				
	By construction, the resulting tree decomposition $(T',{\cal X}')$ of $G'$ is a star-decomposition, hence $tb(G') = 1$.	
				
				\end{proofof}

Recall that we can use our reduction from Definition~\ref{def:plpb-reduction} in order to prove that computing path-length and path-breadth is NP-hard.
By contrast, our reduction from Theorem~\ref{thm:tb-npc} cannot be used to prove that tree-length is NP-hard to compute (in fact, the graph $G'$ resulting from the reduction has tree-length two).

Finally, as in the previous Section~\ref{sec:plpb}, let us strenghten Theorem~\ref{thm:tb-npc} with an inapproximability result.

\begin{corollary}
	\label{cor:pb-inapprox}
	For every $\varepsilon > 0$, the tree-breadth of a graph cannot be approximated within a factor $2 - \varepsilon$ unless P=NP.
\end{corollary}

\section{General properties of graphs with tree-breadth one}
\label{sec:tb}

In Section~\ref{sec:tb-npc}, we prove that computing the tree-breadth is NP-hard.
In particular, the recognition of graphs with tree-breadth one is NP-complete.
In light of this result, we focus on graphs with tree-breadth one (in order to obtain a better understanding of what makes the problem difficult)

\begin{center}
    \fbox{\begin{minipage}{.95\linewidth}\label{prob:tb}
        \begin{problem}[1-tree-breadth]\ 
          \begin{description}
          \item[Input:] a connected graph $G=(V,E)$
          \item[Question:] $tb(G) \leq 1$ ?
          \end{description}
        \end{problem}     
      \end{minipage}}
  \end{center}

In Lemma~\ref{lem:reduction-tb}, we show that the problem of recognizing graphs with tree-breadth at most one is equivalent to the problem of computing tree-breadth.
This further motivates our study of these graphs.
Then, we will prove necessary conditions for a graph to be of tree-breadth one.
\begin{itemize}
\item One is that all graphs with a star-decomposition have a \emph{domination elimination ordering} (see Section~\ref{sec:order}).
We will outline a few implications of this property.
\item Second, we will prove in Lemma~\ref{lem:sep-keep} that if a graph $G$ admits a star-decomposition then so do all the \emph{blocks} of $G$, where the blocks here denote a particular case of induced subgraphs of $G$ ({\it e.g.}, see Definition~\ref{def:min-sep}).
\end{itemize} 
Finally, we will obtain from the latter result a polynomial-time algorithm to decide whether a bipartite graph has tree-breadth at most one ({\it e.g.}, see Section~\ref{sec:bipartite}).

\begin{definition}
\label{def:reduction-tb}
Let $G$ be a graph with $n$ vertices, denoted by $v_1,v_2,\ldots,v_n$, and let $r$ be a positive integer.
The graph $G'_r$ is obtained from $G$ by adding a clique with $n$ vertices, denoted by $U = \{ u_1,u_2,\ldots,u_n \}$, so that for every $1 \leq i \leq n$, vertex $u_i$ is adjacent to $B_G(v_i,r) = \{ x \in V(G) \mid dist_G(v_i,x) \leq r \}$.
\end{definition}

\begin{lemma}
\label{lem:reduction-tb}
For every graph $G$, for every positive integer $r$, let $G'_r$ be as defined in Definition~\ref{def:reduction-tb}, $tb(G) \leq r$ if and only if $tb(G'_r) \leq 1$.
\end{lemma}

\begin{proof}
If $tb(G) \leq r$ then we claim that starting from any tree decomposition $(T,{\cal X})$ of $G$ with breadth at most $r$, one obtains a star-decomposition of $G_r'$ by adding the clique $U$ in every bag $X_t, \ t\in V(T)$.
Indeed, in such case for every bag $X_t, \ t \in V(T)$, by the hypothesis there is $v_i \in V(G)$ such that $\max_{x \in X_t} dist_G(v_i,x) \leq r$, whence $X_t \cup U \subseteq N_{G'_r}[u_i]$.
Conversely, if $tb(G'_r) \leq 1$ then we claim that starting from any tree decomposition $(T',{\cal X}')$ of $G'_r$ with breadth at most one, one obtains a tree decomposition of $G$ with breadth at most $r$ by removing every vertex of the clique $U$ from every bag $X_t', \ t \in V(T')$. 
Indeed, in such case for every bag $X_t', \ t \in V(T')$, by the hypothesis there is $y \in V(G'_r)$ such that $X_t' \subseteq N_{G'_r}[y]$.
Furthermore, $y \in \{ u_i, v_i \}$ for some $1 \leq i \leq n$, and so, since $N_{G'_r}[v_i] \subseteq N_{G'_r}[u_i]$ by construction, therefore $X_t' \setminus U \subseteq N_{G'_r}(u_i) \setminus U = \{ x \in V(G) \mid dist_G(v_i,x) \leq r \}$.
\end{proof}

\subsection{Existence of specific elimination orderings}
\label{sec:order}

Independently from the remaining of the section, let us prove some interesting properties of graphs with tree-breadth one in terms of \emph{elimination orderings}. 
More precisely, a \emph{domination elimination ordering}~\cite{Dahlhaus1994} of a graph $G$ is a total ordering of its vertex-set, denoted by $v_1,v_2,\ldots,v_n$, so that for every $1 \leq i < n$, there is $j > i$ satisfying that $N_G(v_i) \cap \{ v_{i+1}, v_{i+2}, \ldots, v_n \} \subseteq N_G[v_j]$.
The existence of domination elimination orderings in some graph classes and their algorithmic applications has been studied in~\cite{Chepoi1998}.  
Let us prove that graphs with tree-breadth one all admit a domination elimination ordering.

\begin{lemma}
\label{lem:dom-order}
Let $G$ be such that $tb(G) \leq 1$, $G$ admits a domination elimination ordering.
\end{lemma}

\begin{proof}
Assume $G$ has at least two vertices (or else, the lemma is trivial).
To prove the lemma, it suffices to prove the existence of $u,v \in V(G)$ distinct such that $N(v) \subseteq N[u]$ and $tb(G \setminus v) \leq 1$ (then, the lemma follows by induction on the order of the graph).

If $G$ admits a universal vertex $u$, then one can pick $v \in V(G) \setminus u$ arbitrary, $N(v) \subseteq N[u]$ because $u$ is universal in $G$, furthermore $tb(G \setminus v) \leq 1$ because $G \setminus v$ admits a universal vertex $u$.

Else, $G$ does not admit any universal vertex, let $(T,{\cal X})$ be a reduced tree decomposition of $G$ of breadth one, that is a star-decomposition by Lemma~\ref{lem:dominator-in-bag}.
Let $X_t, \ t \in V(T)$ be a leaf.
Since the tree decomposition is reduced, there must be $v \in X_t$ satisfying $T_v = \{ X_t \}$.
Now there are two cases.
\begin{itemize}
\item Suppose there is $u \in X_t \setminus v$ such that $X_t \subseteq N[u]$.
Then, $N(v) \subseteq X_t \subseteq N[u]$, and $tb(G \setminus v) \leq 1$ because $G \setminus v$ can be obtained from $G$ by contracting the edge $\{u,v\}$ and tree-breadth is contraction-closed by Lemma~\ref{lem:contraction-closed}.
\item Else, $X_t \subseteq N[v]$, and for every $x \in X_t \setminus v, \ X_t \not\subseteq N[x]$.
Let $t' \in V(T)$ be the unique node adjacent to node $t$ in $T$, that exists because $G$ does not admit any universal vertex and so, $T$ has at least two bags. 
Let us assume that for every $x \in X_t \setminus v$, $x \in X_t \cap X_{t'}$ (for otherwise, $N(x) \subseteq N[v]$ and $tb(G \setminus x) \leq 1$ because $G \setminus x$ can be obtained from $G$ by contracting the edge $\{v,x\}$ to $v$ and tree-breadth is contraction-closed by Lemma~\ref{lem:contraction-closed}).
In particular, let $u \in X_{t'}$ satisfy $X_{t'} \subseteq N[u]$.
Then, $N(v) = X_t \cap X_{t'} \subseteq N[u]$, furthermore $tb(G \setminus v) \leq 1$ because $(T \setminus t, {\cal X} \setminus X_t)$ is a star-decomposition of $G \setminus v$.
\end{itemize}\vspace{-15pt}\end{proof}

Note that for a graph to have tree-breadth one, it must satisfy the necessary condition of Lemma~\ref{lem:dom-order} and this can be checked in polynomial-time.
However, the existence of some domination elimination ordering is not a sufficient condition for the graph to have tree-breadth one.
Indeed, every grid has a domination elimination ordering but the \emph{tree-length} of the $n \times m$ grid is at least $\min\{n,m\} -1$~\cite{Dourisboure2007b} (recall that $tl(G) \leq 2tb(G)$ for any graph $G$).

%For instance, consider a $5 \times 5$ grid, it has a domination elimination ordering and treewidth $5$.
%However, as proved in Lemma~\ref{lem:tw-planar-tb-1} in the following, planar graphs with tree-breadth one must have treewidth at most $4$, so, the $5 \times 5$-grid does not have tree-breadth one.

The existence of a domination elimination ordering has some interesting consequences about the graph structure.
Let us recall one such a consequence about the \emph{cop-number} of the graph.

\begin{corollary}
\label{cor:cop-number}
For any graph $G$ with $tb(G) \leq 1$, $G$ has cop-number $\leq 2$ and the upper-bound is sharp.
\end{corollary}

\begin{proof}
By Lemma~\ref{lem:dom-order}, $G$ admits a domination elimination ordering.
Therefore, by~\cite[Theorem 4]{Clarke2005} $G$ has cop-number $ \leq 2$. 
One can prove the sharpness of the upper-bound by setting $G := C_4$, the cycle with four vertices. 
\end{proof}

\subsection{Properties of particular decompositions}

In the following, it will be useful not only to constrain the properties of the star-decomposition whose existence we are interested in, but also to further constrain the properties of the graph $G$ that we take as input.
Let us first remind basic terminology about graph separators.

\begin{definition}
\label{def:min-sep}
Let $G=(V,E)$ be connected, a \emph{separator} of $G$ is any subset $S \subseteq V$ such that $G \setminus S$ has at least two connected components.

In particular, a \emph{full component} for $S$ is any connected component $C$ of $G \setminus S$ satisfying $N(C) = S$.
A \emph{block} is any induced subgraph $G[C \cup S]$ with $S$ being a separator and $C$ being a full component for $S$.  

Finally, a \emph{minimal separator} is a separator with at least two full components. 
\end{definition}

Our objective is to prove that if a graph $G$ has tree-breadth one then so do all its blocks.
In fact, we will prove a slightly more general result:

\begin{lemma}\label{lem:sep-keep}
Let $G=(V,E)$ be a graph, $S \subseteq V$ be a separator, and $W \subseteq V \setminus S$ be the union of some connected components of $G \setminus S$.
If $tb(G) = 1$ and $W$ contains a full component for $S$, then $tb(G[W \cup S]) = 1$.
More precisely if $(T,{\cal X})$ is a tree decomposition of $G$ of breadth one, then $(T, \{ X_t \cap (W \cup S) \mid X_t \in {\cal X}\})$ is a tree decomposition of $G[W \cup S]$ of breadth one.
\end{lemma}

\begin{proof}
Let $(T,{\cal X})$ be a tree decomposition of breadth one of $G$.
Let us remove all vertices in $V \setminus (W \cup S)$ from bags in $(T,{\cal X})$, which yields a tree decomposition $(T',{\cal X}')$ of the induced subgraph $G[W \cup S]$.
To prove the lemma, we are left to prove that $(T',{\cal X}')$ has breadth one.
Let $X_t$ be a bag of $(T',{\cal X}')$.
By construction, $X_t$ is fully contained into some bag of $(T,{\cal X})$, so it has radius one in $G$.
Let $v \in V$ be such that $X_t \subseteq N_G[v]$.
If $v \in W \cup S$, then we are done.
Else, since for all $x \notin S \cup W, N(x) \cap (S \cup W) \subseteq S$ (because $S$ is a separator by the hypothesis), we must have that $X_t \subseteq S$.
Let $A \subseteq W$ be a full component for $S$, that exists by the hypothesis, and let $T_A$ be the subtree that is induced by the bags intersecting the component.
Since we have that the subtree $T_A$ and the subtrees $T_x, x \in X_t$ pairwise intersect --- because for all $x \in X_t$, $x \in S$ and so, $x$ has a neighbour in $A$ ---, then by the Helly property (Lemma~\ref{lem:helly}) $T_A \cap \left(\bigcap_{x \in X_t} T_x\right) \neq \emptyset$ {\it i.e.}, there exists a bag in $(T,{\cal X})$ containing $X_t$ and intersecting $A$.
Moreover, any vertex dominating this bag must be either in $S$ or in $A$, so in particular there exists $u \in A \cup S$ dominating $X_t$, which proves the lemma.
\end{proof}

Lemma~\ref{lem:sep-keep} implies that, under simple assumptions, a graph of tree-breadth one can be disconnected using any (minimal) separator, and the components must still induce subgraphs with tree-breadth one.
The converse does not hold in general, yet there are interesting cases when it does.

\begin{lemma}\label{lem:cliqueSep}
Let $G=(V,E)$ be a graph, $S \subseteq V$ be a clique-minimal-separator and $A$ be a full component for $S$.
Then, $tb(G) = 1$ if and only if both $tb(G[A \cup S]) = 1$ and $tb(G[V \setminus A]) = 1$.
\end{lemma}

\begin{proof}
By the hypothesis $V \setminus (A \cup S)$ contains a full component because $S$ is a minimal separator.
Therefore, if $G$ has tree-breadth one, then so do $G[A \cup S]$ and $G[V \setminus A]$ by Lemma~\ref{lem:sep-keep}.
Conversely, suppose that we have both $tb(G[A \cup S]) = 1$ and $tb(G[V \setminus A]) = 1$.
Let $(T^1,{\cal X}^1)$ be a tree decomposition of $G[A \cup S]$ with breadth one, let $(T^2,{\cal X}^2)$ be a tree decomposition of $G[V \setminus A]$ with breadth one.
Then for every $i \in \{1,2\}$ we have that since $S$ is a clique the subtrees $T^i_s, \ s \in S$, pairwise intersect, so by the Helly Property (Lemma~\ref{lem:helly}) $\bigcap_{s \in S}T^i_s \neq \emptyset$ {\it i.e.}, $S$ is fully contained into some bag of $(T^1,{\cal X}^1)$ and it is fully contained into some bag of $(T^2,{\cal X}^2)$.
Moreover, $(A \cup S) \cap (V \setminus A) = S$, therefore a tree decomposition of $G$ with breadth one can be obtained by adding an edge between some bag of $(T^1,{\cal X}^1)$ containing $S$ and some bag of $(T^2,{\cal X}^2)$ containing $S$.
\end{proof}

Recall that computing the clique-minimal-decomposition of a graph $G$ can be done in ${\cal O}(nm)$-time, where $m$ denotes the number of edges~\cite{Berry2010}.
By doing so, one replaces a graph $G$ with the maximal subgraphs of $G$ that have no clique-separator, \textit{a.k.a.} \emph{atoms}.
Therefore, we will assume in the remaining of the proofs that there is no clique-separator in the graphs that we will study, we will call them {\it prime graphs}.
%Note that there are only ${\cal O}(m)$ vertices in total in the disjoint union of all atoms.
%However, we will have to be careful that we do not modify our input in a way that new clique-separators may appear in the prime graph (otherwise, new clique-decompositions are needed, and we are not able to bound the number of subgraphs to be considered).

\subsection{Application to bipartite graphs}
\label{sec:bipartite}

In this section, we describe an ${\cal O}(nm)$-time algorithm so as to decide whether a prime bipartite graph has tree-breadth one.
This combined with Lemma~\ref{lem:cliqueSep} proves that it can be decided in polynomial-time whether a bipartite graph has tree-breadth one.

We will first describe a more general problem and how to solve it in polynomial-time.

\paragraph{Tree decompositions with constrained set of bags.}
Our algorithm for bipartite graphs makes use of the correspondance between tree decompositions and triangulations of a graph.
Indeed, recall that any reduced tree decomposition $(T,{\cal X})$ of a graph $G$ is a clique-tree for some chordal supergraph $H$ of $G$ whose maximal cliques are the bags of ${\cal X}$.
Conversely, for any chordal supergraph $H$ of $G$, every clique-tree of $H$ is a tree decomposition of $G$ whose bags are the maximal cliques of $H$~\cite{Gavril1974}.
Therefore as shown below, the following subproblem can be solved in polynomial-time:

\begin{center}
    \fbox{\begin{minipage}{.95\linewidth}\label{pb:td-bags-given}
        \begin{problem}\ 
          \begin{description}
          \item[Input:] a graph $G$, a family ${\cal X}$ of subsets of $V(G)$.
          \item[Question:] Does there exist a tree $T$ such that $(T,{\cal X})$ is a tree decomposition of $G$ ?
          \end{description}
        \end{problem}     
      \end{minipage}}
  \end{center}

Let us assume w.l.o.g. that no subset of ${\cal X}$ is properly contained into another one.
To solve Problem~\ref{pb:td-bags-given}, it suffices to make every subset $X \in {\cal X}$ a clique in $G$, then to verify whether the resulting supergraph $H$ of $G$ is a chordal graph whose maximal cliques are exactly the sets in ${\cal X}$. 
Since chordal graphs can be recognized in linear-time, and so can be enumerated their maximal cliques~\cite{Gavril1972}, therefore Problem~\ref{pb:td-bags-given} can be solved in polynomial-time.

\paragraph{The algorithm for bipartite graphs.}
Now, given a bipartite graph $G$, we aim to exhibit a family ${\cal X}$ so that $tb(G) = 1$ if and only if there is a star-decomposition of $G$ whose bags are ${\cal X}$.
By doing so, we will reduce the recognition of bipartite graph with tree-breadth at most one to the more general Problem~\ref{pb:td-bags-given}.

\begin{lemma}
\label{lem:bipartite-tb}
Let $G = (V_0 \cup V_1, E)$ be a prime bipartite graph with tree-breadth one.
There is $(T,{\cal X})$ a star-decomposition of $G$ such that either ${\cal X} = \{ N[v_0] \mid v_0 \in V_0 \}$, or ${\cal X} = \{ N[v_1] \mid v_1 \in V_1 \}$.
\end{lemma}

\begin{proof}
Let $(T,{\cal X})$ be a star-decomposition of $G$, that exists by Lemma~\ref{lem:dominator-in-bag}, minimizing the number $|{\cal X}|$ of bags.
Suppose there is some $v_0 \in V_0$, there is $t \in V(T)$ such that $X_t \subseteq N_G[v_0]$ (the case when there is some $v_1 \in V_1$, there is $t \in V(T)$ such that $X_t \subseteq N_G[v_1]$ is symmetrical to this one).
We claim that for every $t' \in V(T)$, there exists $v_0' \in V_0$ satisfying $X_{t'} \subseteq N_G[v_0']$.
By contradiction, let $v_0 \in V_0, v_1 \in V_1$, let $t,t'\in V(T)$ be such that $X_t \subseteq N_G[v_0], X_{t'} \subseteq N_G[v_1]$.
By connectivity of the tree $T$ we may assume w.l.o.g. that $\{t,t'\} \in E(T)$.
Moreover, $N_G(v_0) \cap N_G(v_1) = \emptyset$ because $G$ is bipartite.
Therefore, $X_t \cap X_{t'} \subseteq \{ v_0, v_1 \}$, and in particular if $X_t \cap X_{t'} = \{ v_0, v_1 \}$ then $v_0,v_1$ are adjacent in $G$.
However, by the properties of a tree decomposition this implies that $X_t \cap X_{t'}$ is a clique-separator (either an edge or a single vertex), thus contradicting the fact that $G$ is prime.

Now, let $v_0 \in V_0$ be arbitrary.
We claim that there is a unique bag $X_t, \ t \in V(T)$, containing $v_0$.
Indeed, any such bag $X_t$ must satisfy $X_t \subseteq N_G[v_0]$, whence the subtree $T_{v_0}$ can be contracted into a single bag $\bigcup_{t \in T_{v_0}} X_t$ without violating the property for the tree decomposition to be a star-decomposition.
As a result, the unicity of the bag $X_t$ follows from the minimality of $|{\cal X}|$.
Finally, since $X_t$ is unique and $X_t \subseteq N_G[v_0]$, therefore $X_t = N_G[v_0]$ and so, ${\cal X} = \{ N[v_0] \mid v_0 \in V_0 \}$.
\end{proof}

We can easily deduce from Lemma~\ref{lem:bipartite-tb} the following algorithm for deciding whether a prime bipartite graph $G$ has tree-breadth one.
Let $(V_0,V_1)$ be the (unique) bipartition of the vertex-set of $G$ into two stable sets.
Let ${\cal X}_0 = \{ N[v_0] \mid v_0 \in V_0 \}$, let ${\cal X}_1 = \{ N[v_1] \mid v_1 \in V_1 \}$.
By Lemma~\ref{lem:bipartite-tb}, $tb(G) = 1$ if and only if one of $(G,{\cal X}_0), (G,{\cal X}_1)$ is a yes-instance of Problem~\ref{pb:td-bags-given}.

\section{Algorithm for planar graphs}
\label{sec:planar}

We are now ready to present our main result.
In this section, we describe a quadratic-time algorithm for deciding whether a prime planar graph has tree-breadth one. 
Overall, we claim that it gives us a quadratic-time algorithm for deciding whether a general planar graph has tree-breadth one.
Indeed, the clique-decomposition of a planar graph takes ${\cal O}(n^2)$-time to be computed, furthermore the disjoint union of the atoms has ${\cal O}(n+m)$ vertices~\cite{Berry2010}, that is ${\cal O}(n)$ for planar graphs.

\medskip
Roughly, we will construct a star-decomposition of the graph by increments.
The main principle of the recursive algorithm is to find a particular vertex, called {\it leaf-vertex}.
Informally, it extracts a new bag of the star-decomposition from some ball around the leaf-vertex.
Then, depending on the case, either the leaf-vertex vertex is removed or some edge is added or contracted.
In both cases, the resulting graph remains prime and planar and has tree-breadth one if and only if the initial one has tree-breadth one. 
%However, we will have to be careful that we do not modify our input in a way that new clique-separators may appear in the prime planar graph (otherwise, new clique-decompositions are needed, and we are not able to bound the number of subgraphs to be considered).

We prove that each inductive step takes a linear time. Moreover, we prove that there are at most a linear number of recursive iterations (Lemma~\ref{lem:ub-steps}). 

\medskip
There are three kinds of leaf-vertices ({\it e.g.}, see Figure~\ref{fig:leaf-vertex}). 
\begin{definition}\label{def:leafVertex}
Let $G=(V,E)$ be a graph. A vertex $v$ is a \emph{leaf-vertex} if one of the following conditions hold.
\begin{description}
\item[Type 1.]  $N(v)$ induces an $a_vb_v$-path for some $a_v,b_v \in V \setminus \{v\}$, denoted by $\Pi_v$, of length at least $3$ and there exists $d_v \in V \setminus \{v\}$ such that $N(v) \subseteq N(d_v)$, i.e., $d_v$ dominates $\Pi_v$. 
\item [Type 2.] $N(v)$ induces a path, denoted by $\Pi_v=(a_v,b_v,c_v)$,  of length $2$. 
\item[Type 3.] $N(v)$ consists of two non adjacent vertices $a_v$ and $c_v$, and there exists $b_v \in (N(a_v) \cap N(c_v)) \setminus \{v\}$. 
\end{description} 
\end{definition}

\begin{figure}[h!]
 \centering
 \includegraphics[width=0.75\textwidth]{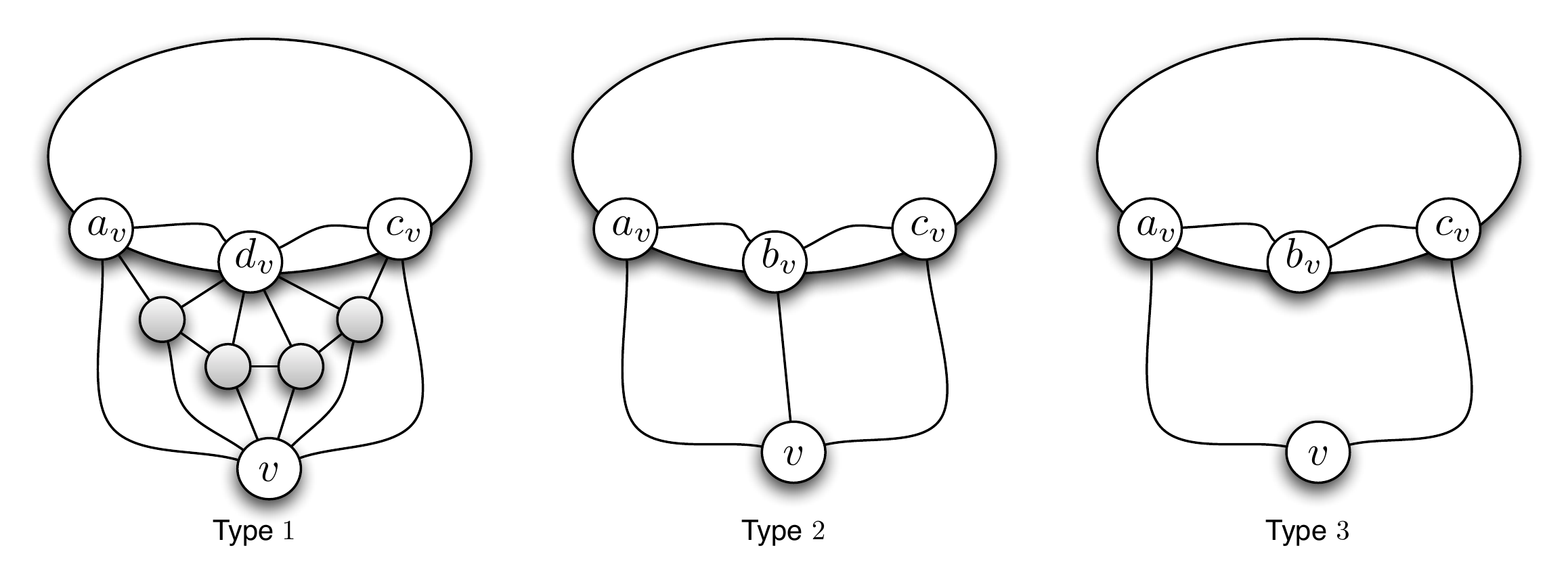}
 \caption{The three kinds of leaf-vertices.}
 \label{fig:leaf-vertex}
\end{figure}

%Note that, while Type $3$ includes Type $2$, we were not able not to distinguish the cases. 

We are now ready to describe the algorithm. 

\subsection{Algorithm \texttt{Leaf-BottomUp}}

Let $G=(V,E)$ be prime planar graph.
Assume $G$ has at least $7$ vertices (else, it is easy to conclude).

\begin{enumerate}[label = Step \arabic*, itemsep=15pt]
	
\item\label{step:find-leaf-vertex}
The first step is to find a leaf-vertex in $G$.
In Section~\ref{sec:complexity-leaf-vertex}, we describe how to decide whether $G$ has a leaf-vertex in linear-time. 
%In the latter case, given a vertex of degree two, Algorithm \texttt{Leaf-BottomUp} also verifies whether it is a leaf-vertex of Type 3 in linear-time.
%If not, then by Lemma~\ref{lem:2-separator} $tb(G) > 1$. 

	\begin{itemize}[itemsep=10pt]
		
		\item if $G$ has no leaf-vertex, then, by Theorem~\ref{th:noP3}, no minimal separator of $G$ induces a path of length $2$. 
		Therefore, by Lemma~\ref{lem:2bags}, $tb(G)=1$ only if $G$ has a star-decomposition with at most $2$ bags.
		In that case, Algorithm \texttt{Leaf-BottomUp} checks whether it exists a star-decomposition with at most $2$ bags, which can be done in quadratic time (see Lemma~\ref{lem:complexity-2bags}).
		If it exists, then $tb(G)=1$.
		Otherwise, $tb(G)>1$. 
		
		\item otherwise, let $v$ be a leaf-vertex of $G$ and go to~\ref{step:type-1} if $v$ is of Type $1$ and go to~\ref{step:type-2-3} otherwise.
		
	\end{itemize}

\item\label{step:type-1}{\bf Case $v$ is of Type $1$.}
Let $\Pi_v$ and $d_v$ be defined as in Definition~\ref{def:leafVertex}.
If $V = N[v] \cup \{d_v\}$ then trivially $tb(G) = 1$. 
Else by Theorem~\ref{th:Type1}, $G'$ is prime and planar, where $G'$ is the graph obtained from $G \setminus v$ by contracting the internal nodes of $\Pi_v$ to a single edge, and $tb(G)=1$ if and only if $tb(G')=1$. In that case, Algorithm \texttt{Leaf-BottomUp} is recursively applied on $G'$. 

\item\label{step:type-2-3}{\bf Case $v$ is of Type $2$ or $3$.}
Let $a_v,b_v,c_v$ be defined as in Definition~\ref{def:leafVertex}.

In that case, Algorithm \texttt{Leaf-BottomUp} checks whether $G \setminus v$ is prime.
By Theorem~\ref{th:GminusVprime?}, for any clique minimal separator $S$ of $G \setminus v$ (if any), there exists $u_v \in V \setminus \{ a_v, b_v, c_v, v \}$ such that $S=\{b_v,u_v\}$.  
Therefore, this can be checked in linear time (by checking with a Depth-First-Search whether there is a cut-vertex of $G \setminus \{ a_v, b_v, c_v, v \}$ in the neighbors of $b_v$).
If $G \setminus v$ is prime then go to~\ref{step:prime-case}, else go to~\ref{step:new-clique-sep}.

	\begin{enumerate}[label=\labelenumi.\arabic*, itemsep=10pt]
		
		\item\label{step:prime-case}{\bf Case $v$ is of Type $2$ or $3$ and $G \setminus v$ is prime.}
		There are $4$ cases that can be determined in linear-time.
		
			\begin{enumerate}[label=(\alph*), itemsep=5pt]
				
				\item\label{step:many-neighbours}{\bf Case $|N(a_v) \cap N(c_v)| \geq 3$ in $G \setminus v$, or there exists a minimal separator $S\subseteq (N(a_v) \cap N(c_v)) \cup \{a_v,c_v\}$ in $G \setminus v$ and $\{a_v,c_v\} \subseteq S$.}
				
				By Theorem~\ref{th:primeEasy}, $tb(G)=1$ if and only if $tb(G \setminus v)=1$.
				Since, moreover, $G \setminus v$ is planar and prime, then Algorithm \texttt{Leaf-BottomUp} is recursively applied on $G \setminus v$.  
				
				\item\label{step:few-neighbours}{\bf Case: $|N(a_v) \cap N(c_v)| < 3$ in $G \setminus v$ and there is no minimal separator $S\subseteq (N(a_v) \cap N(c_v)) \cup \{a_v,c_v\}$ in $G \setminus v$ such that $\{a_v,c_v\} \subseteq S$.}
				
				\smallskip
				
					\begin{enumerate}[label=\roman*]
						
						\item\label{step:one-neighbour}{\bf Subcase: $|N(a_v) \cap N(c_v)| =1$ in $G \setminus v$}.
						In that subcase, $N(a_v) \cap N(c_v)=\{v,b_v\}$ and, by Theorem~\ref{thm:force-clique}, $tb(G)=1$ if and only if $G=C_4$, a cycle with four vertices.
						Note that here it implies that $tb(G) > 1$ because $G$ has at least $7$ vertices.
						
						\item\label{step:two-neighbours}{\bf Subcase: $|N(a_v) \cap N(c_v)| =2$ in $G \setminus v$}.
						In that subcase, let $N(a_v) \cap N(c_v)=\{v,b_v,u_v\}$.
						By Theorem~\ref{th:primeDifficult}, since $G$ has more than $5$ vertices, the graph $G'$ obtained from $G$ by adding edges $\{v,u_v\}$ and $\{b_v,v\}$ is planar and prime, and moreover $tb(G) = 1$ if and only if $tb(G') = 1$. In that case, Algorithm \texttt{Leaf-BottomUp} is recursively applied on $G'$.
						
					\end{enumerate}
				
			\end{enumerate}
		
		\item\label{step:new-clique-sep}{\bf Case $v$ is of Type $2$ or $3$ and $G \setminus v$ has a clique separator.}
		As mentioned in~\ref{step:type-2-3}, in that case, there exists $u_v \in V \setminus \{ a_v, b_v, c_v, v\}$ such that $S=\{b_v,u_v\}$ is a minimal clique separator of $G \setminus v$. 
		Moreover, by Theorem~\ref{th:GminusVprime?}, $G \setminus \{ a_v, b_v, c_v, v\}$ is connected.
		
		\smallskip
		
		By Theorem~\ref{th:addEdgebv}, $tb(G)=1$ if and only if $tb(G')=1$ where $G'$ is obtained from $G$ by adding the edge $\{v,b_v\}$ (if it were not already there). 
		Moreover, $G'$ is prime and planar. 
		Hence, we may assume that $\{v,b_v\} \in E$ (if not Algorithm \texttt{Leaf-BottomUp} adds it).
		
		\smallskip
		
		Furthermore by Theorem~\ref{th:GminusVprime?}, since $G$ has more than $5$ vertices, $u_v \notin N(a_v) \cap N(c_v)$. 
		In the latter case, let us assume w.l.o.g. that $u_v \notin N(a_v)$, that is either $u_v \notin N(a_v) \cup N(c_v)$ or $u_v \in N(c_v) \setminus N(a_v)$. 
		There are several cases to be considered.
		
			\begin{enumerate}[label=(\alph*), itemsep=5pt]
				
				\item\label{step:no-separation}{\bf Case  $u_v \notin N(a_v) \cup N(c_v)$, \\ or $(N(u_v) \cap N(a_v)) \cup \{v,c_v\}$ does not separate $u_v$ and $a_v$ in $G$.}
				
				\smallskip
				
				By Theorem~\ref{th:contractva}, $G/ va_v$ is prime and planar, and $tb(G)=1$ if and only if $tb(G/va_v)=1$. 
				In that case, Algorithm \texttt{Leaf-BottomUp} is recursively applied on $G/va_v$, the graph obtained from $G$ by contracting the edge $\{v,a_v\}$. 
				
				\item\label{step:separation}{\bf Case $u_v \in N(c_v) \setminus N(a_v)$, \\ and $(N(u_v) \cap N(a_v)) \cup \{v,c_v\}$ separates $u_v$ and $a_v$ in $G$.} 
				
				\smallskip
				
				In that case, recall that $G$ has at least $7$ vertices. 
				Again, Algorithm \texttt{Leaf-BottomUp} distinguishes several subcases.
				
					\begin{enumerate}[label=\roman*]
						
						\item\label{step:contraction-case}{\bf Subcase $N(b_v)=\{v,a_v,c_v,u_v\}$.}
						In that subcase, by Theorem~\ref{claim:connect-vertex-b}, we can find in linear-time a vertex $x \in (N(a_v) \cap N(u_v)) \setminus \{b_v\}$ such that $G'$ is planar, where $G'$ is obtained from $G$ by adding the edge $\{b_v,x\}$. 
						Moreover, by Theorem~\ref{lem:clique-case-2}, $G' / b_vx$ (obtained by contracting $\{b_v,x\}$) is prime and $tb(G)=1$ if and only if $tb(G' / b_vx)=1$. 
						In that case, Algorithm \texttt{Leaf-BottomUp} is recursively applied on $G' / b_vx$.
						
						\item\label{step:diamond-case}{\bf Subcase $\{v,a_v,c_v,u_v\} \subset N(b_v)$ and $N(b_v)\cap N(a_v) \cap N(u_v)\neq \emptyset$.}
						In that subcase, $|N(b_v)\cap N(a_v) \cap N(u_v)|=1$ by Lemma~\ref{claim:sep-b} and let $x$ be this common neighbor. 
						By Theorem~\ref{lem:clique-case-2}, $G/ b_vx$ (obtained by contracting $\{b_v,x\}$) is prime and $tb(G)=1$ if and only if $tb(G/ b_vx)=1$. 
						In that case, Algorithm \texttt{Leaf-BottomUp} is recursively applied on $G/ b_vx$.
						
						\item\label{step:no-diamond-case}{\bf Subcase $\{v,a_v,c_v,u_v\} \subset N(b)$ and $N(b_v)\cap N(a_v) \cap N(u_v)= \emptyset$.}
						In that subcase, by Theorem~\ref{lem:final-case-1}, there must be a unique $x \in (N(a_v) \cap N(u_v)) \setminus \{b_v\}$ such that $N(b_v) \cap N(x)$ is a $b_vx$-separator of $G$ and $|N(b_v) \cap N(x)| \geq 3$ (or else, $tb(G) > 1$).

							\begin{itemize}
								
								\item {\bf Suppose there is a leaf-vertex $\ell \in N(b_v) \cap N(x)$.}
								By Lemma~\ref{claim:final-case-2}, $\ell$ is of Type 1 or $G \setminus \ell$ is prime.
								In that case, go to~\ref{step:type-1} if $\ell$ is of Type $1$, and go to~\ref{step:prime-case} if $\ell$ has Type $2$ or $3$ (in both cases, $\ell$ takes the role of $v$).
								Note that we never go back to~\ref{step:new-clique-sep} in such case, so the algorithm cannot loop. 
								
								\item Otherwise, by Theorem~\ref{lem:final-case-3}, there exist $y,z \in N(b_v) \cap N(x)$ two non-adjacent vertices, such that $G'$ is prime and planar, and $tb(G)=1$ if and only if $tb(G')=1$, where $G'$ is obtained from $G$ by adding the edge $\{x,y\}$. 
								In that case, Algorithm \texttt{Leaf-BottomUp} is recursively applied on $G'$.
								
							\end{itemize}
						
					\end{enumerate}
				
			\end{enumerate}
		
	\end{enumerate} 
	
\end{enumerate}

\subsection{Properties of prime planar graphs with tree-breadth one}

\subsubsection{General lemmas}

We will first investigate on general properties of prime planar graphs.
In particular, the following properties do not depend on the existence of a star-decomposition, therefore we do not use tree decompositions in our proofs.
However, note that we refer to Definition~\ref{def:leafVertex} in Theorem~\ref{th:GminusVprime?}.
For clarity, we will separate the properties that hold for every {\it biconnected} planar graph from those that only hold for prime planar graphs.

\paragraph{Properties of biconnected planar graphs.}
In order to obtain these properties, we will mostly rely on the notion of \emph{intermediate graphs}, defined below. 

\begin{definition}~\cite[Definition 6]{Bouchitte2003}
\label{def:intermediate-graph}
Let $G=(V,E)$ be a planar graph.
We fix a plane embedding of $G$.
Let $F$ be the set of faces of this embedding.
The \emph{intermediate graph} $G_I = (V \cup F, E_I)$ has vertex-set $V \cup F$,
furthermore $E \subseteq E_I$ and we add an edge in $G_I$ between an original vertex $v \in V$ and a face-vertex $f \in F$ whenever the corresponding vertex and face are incident in $G$ (see Figure~\ref{fig:intermediate-graph}).
\end{definition}

\begin{figure}[h!]
 \centering
 \includegraphics[width=0.5\textwidth]{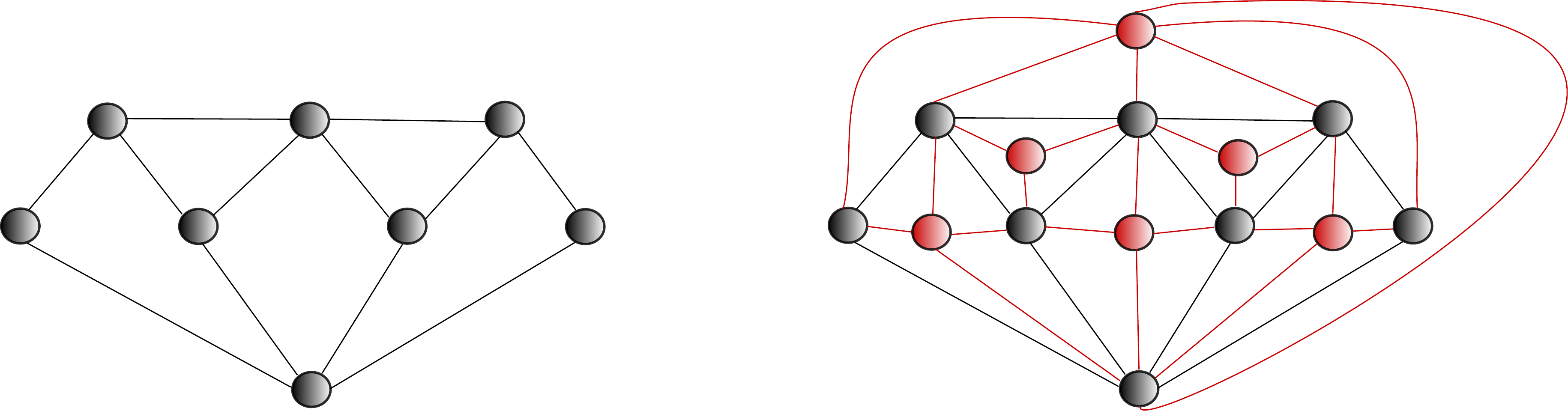}
 \caption{A plane embedding of some planar graph (left) and the corresponding intermediate graph (right). Face-vertices are coloured in red.}
 \label{fig:intermediate-graph}
\end{figure}

Note that an intermediate graph is planar.
Furthermore, since a plane embedding can be constructed in linear-time~\cite{Hopcroft1974}, therefore so can be an intermediate graph.
This is important for the quadratic-time complexity of Algorithm \texttt{Leaf-BottomUp}.
To prove the correctness of Algorithm \texttt{Leaf-BottomUp} in the following, we will rely upon the following property of intermediate graphs.

\begin{lemma}~\cite[Proposition 9]{Bouchitte2003}
\label{lem:prop-todinca}
Let $S$ be a minimal separator of some biconnected planar graph $G=(V,E)$ and let $C$ be a full component of $G \setminus S$.
We fix a plane embedding of $G$.
Then $S$ corresponds to a cycle $v_S(C)$ of $G_I$, of length $2|S|$ and with $V \cap v_S(C) = S$, and such that $G_I \setminus v_S(C)$ has at least two connected components.
Moreover, the original vertices of one of these components are exactly the vertices of $C$. 
\end{lemma}

In the following, we will rely upon two properties which both follow from Lemma~\ref{lem:prop-todinca}.
The first one is the following structural property of minimal separators of planar graphs.

\begin{corollary}
\label{cor:sep-planar}
Let $S$ be a minimal separator of a biconnected planar graph $G=(V,E)$.
Then, $S$ either induces a cycle or a forest of paths.
\end{corollary}

\begin{proof}
Let us fix a plane embedding of $G$, let $G_I$ be the corresponding intermediate graph.	
Then, let $C_S$ be a smallest cycle of $G_I$ such that $V \cap C_S = S$, that exists by Lemma~\ref{lem:prop-todinca}.
To prove the corollary, it suffices to prove that $C_S$ is an induced cycle of $G_I$.
By contradiction, assume the existence of a chord $xy$ of $C_S$.
Note that $x \in S$ or $y \in S$ because face-vertices are pairwise non-adjacent in $G_I$.
Therefore assume w.l.o.g. that $x \in S$. 
Let us divide $C_S$ in two cycles $C_1,C_2$ such that $C_1 \cap C_2 = \{x,y\}$.
By the minimality of $C_S$, $S$ intersects both $C_1 \setminus C_2$ and $C_2 \setminus C_1$.
Therefore, let $z_1,z_2 \in S$ such that $z_1 \in C_1 \setminus C_2$ and $z_2 \in C_2 \setminus C_1$.
Finally, let $A,B$ be two full components of $G \setminus S$.
Observe that $(A \cup B) \cap (C_1 \cup C_2) = \emptyset$ because $V \cap C_S = S$.
Let us contract $C_1,C_2$ in order to obtain the two triangles $(z_1,x,y)$ and $(z_2,x,y)$.
In such case, there is a $K_{3,3}$-minor of $G_I$ with $\{A,B,y\}$ and $\{x,z_1,z_2\}$ being the respective sides of the bipartition, thus contradicting the fact that $G_I$ is planar.
Therefore, $C_S$ is an induced cycle of $G_I$ and so, $S$ induces a subgraph of a cycle in $G$, that is either a cycle or a forest of path.
\end{proof}

%The proof of Lemma~\ref{lem:prop-todinca} in~\cite{} is constructive.
%More precisely, let $F$ be the faces of the fixed embedding and let $F_S(C) = \{ f \in F \mid f \in N_{G_I}(C) \text{ and } N_{G_I}(f) \not\subset S \cup C \}$.
%Then, one can choose $v_S(C)$ as the cycle induced by $S \cup F_S(C)$ in $G_I$.
On the algorithmic side, one can also deduce from Lemma~\ref{lem:prop-todinca} the following corollary.

\begin{corollary}
\label{cor:make-sep-cyc}
Let $G$ be a biconnected planar graph, let $S$ be a minimal separator of $G$.
There is a planar supergraph $G_S$ of $G$ with same vertex-set so that $S$ either induces an edge (if $|S|=2$) or a cycle of $G_S$, and it can be constructed in linear-time.
\end{corollary}

\begin{figure}[h!]
 \centering
 \includegraphics[width=0.75\textwidth]{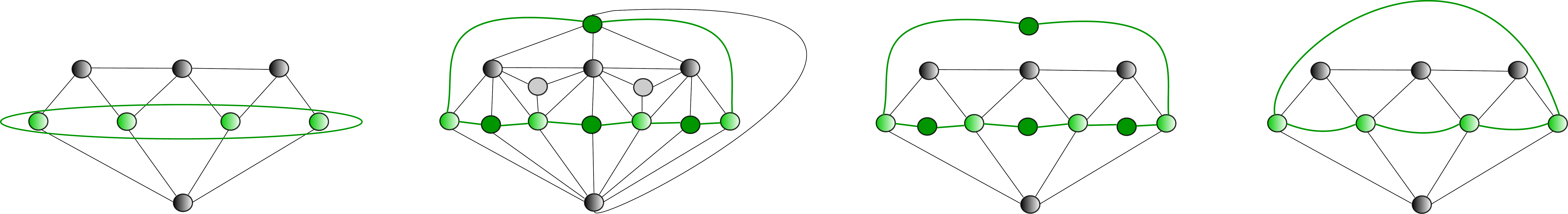}
 \caption{Addition of edges in a planar graph $G$ so as to make a minimal separator of $G$ induce a cycle.}
 \label{fig:make-sep-cyc}
\end{figure}

\begin{proof}
Let us fix a plane embedding of $G$, let $G_I$ be the corresponding intermediate graph.
For every face-vertex $f$ of $G_I$, let us consider $S_f = S \cap N_{G_I}(f)$.
We first claim that $|S_f| \leq 2$.
Indeed, let $A,B$ be two full components of $G \setminus S$, let us contract them to any two vertices $a \in A,b \in B$.
Then, there is a $K_{3,|S_f|}$-minor of $G_I$ with $\{a,b,f\}$ and $S_f$ being the respective parts of the bipartition.
Since $G_I$ is planar by construction, therefore, $|S_f| \leq 2$.

Now, the graph $G_S$ is constructed from $G$ as follows (we refer to Figure~\ref{fig:make-sep-cyc} for an illustration of the proof).
For every face-vertex $f$ of $G_I$, if $S_f = (x,y)$ then we add the edge $\{x,y\}$ in $G_S$.
Note that $G_S$ is a minor of $G_I$ and so, it is a planar graph.
Moreover, by Lemma~\ref{lem:prop-todinca} there is a cycle of $G_I$ whose original vertices are exactly $S$ and so, $S$ induces a connected subgraph of $G_S$.
In particular if $|S|=2$, then it must be an edge.
Else, $|S| > 2$ and the connected subgraph $G_S[S]$ contains a cycle by construction.
Since $S$ is a minimal separator of $G_S$ by construction and $G_S[S]$ is not acyclic, it follows from Corollary~\ref{cor:sep-planar} that $S$ induces a cycle of $G_S$.
\end{proof}

We will often make use of the routine of Corollary~\ref{cor:make-sep-cyc} in order to prove the quadratic-time complexity of Algorithm \texttt{Leaf-BottomUp}.

\paragraph{Properties of prime planar graphs.}
Unlike the above Corollaries~\ref{cor:sep-planar} and~\ref{cor:make-sep-cyc} (which hold for every biconnected planar graph), the following results only hold for prime planar graphs.
We will make use of the following structural properties of prime planar graphs in order to prove the correctness of Algorithm \texttt{Leaf-BottomUp}.

\begin{lemma}\label{lem:P3exists}
Let $G=(V,E)$ be a prime graph that is $K_{3,3}$-minor-free. 
Let $v \in V$, for every minimal separator $S \subseteq N_G(v)$ of the subgraph $G \setminus v$, $S$ consists of two non-adjacent vertices.
\end{lemma}
\begin{proof}
Let $S \subseteq N_G(v)$ be a minimal separator of $G \setminus v$.
There must exist two full components $A$ and $B$ of $S$ in $G \setminus (S \cup \{v\})$. 
Let us remove all nodes of the components of $G \setminus (S \cup \{v\})$ but the ones in $A$ or $B$. 
Then, let us contract $A$ (resp., $B$) in a single vertex $a$ (resp., $b$). We get a $K_{3,|S|}$ as a minor of $G$ where $\{a,b,v\}$ is one part of the bipartition, and so $|S|\leq 2$. 
Finally, since $S \cup \{v\}$ is also a separator of $G$, then $|S| \geq 2$ because otherwise $S \cup \{v\}$ would be an edge-separator.
Therefore, $|S| = 2$ and it is a stable set because otherwise there would be a clique-separator of size $3$ in $G$.
\end{proof}

\begin{lemma}\label{lem:unique-neighbour}
Let $G$ be a prime planar graph, let the path $\Pi=(a,b,c)$ be a separator of $G$, and let $C$ be a component of $G \setminus \Pi$.
Then, there is at most one common neighbour of $a,b$ in $C$.
\end{lemma}

\begin{proof}
First note that $\Pi$ is induced or else it would be a clique-separator of $G$.	
Furthermore, $a,c \in N(C')$ for every component $C'$ of $G \setminus (\Pi \cup C)$ or else $N(C')$ would be a clique-separator of $G$ (either a vertex-separator or an edge-separator).
In particular, it is always possible to make vertices $a,c$ adjacent by contracting an arbitrary component of $G \setminus (\Pi \cup C)$.
	
\begin{figure}[h!]
	\centering
	\includegraphics[width=0.15\textwidth]{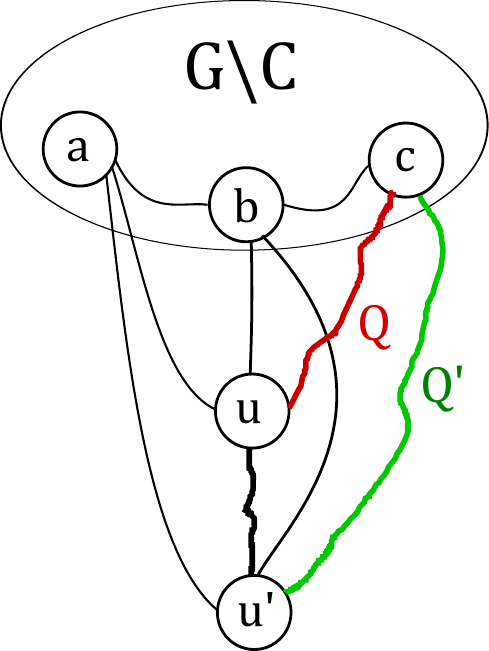}
	\caption{Case where the paths $Q$ and $Q'$ are internally vertex-disjoint paths.}
	\label{fig:disjoint-path-k5}
\end{figure}	
	
By contradiction, let $u,u' \in N(a) \cap N(b) \cap C$ be distinct.
We claim that there exists a $uc$-path $Q$ in $C \cup \{c\}$ that does not contain $u'$, because else the triangle $a,b,u'$ would separate $u$ from $\Pi$, that contradicts the fact that $G$ is prime.
By symmetry, there also exists a $u'c$-path $Q'$ in $C \cup \{c\}$ that does not contain $u$.
There are two cases.

\begin{itemize}
\item $Q$ and $Q'$ are internally vertex-disjoint paths (see Figure~\ref{fig:disjoint-path-k5} for an illustration).
Let us contract $Q \setminus c, Q' \setminus c$ to the vertices $u,u'$, let us contract an arbitrary component of $G \setminus (\Pi \cup C)$ in order to make vertices $a,c$ adjacent, then let us contract a path from $Q$ to $Q'$ in $C$ (that exists, because $C$ is connected by the hypothesis) in order to make vertices $u,u'$ adjacent.
Then one obtains from $a,b,c,u,u'$ a $K_5$-minor, which contradicts the fact that $G$ is planar.

\begin{figure}[h!]
	\centering
	\includegraphics[width=0.175\textwidth]{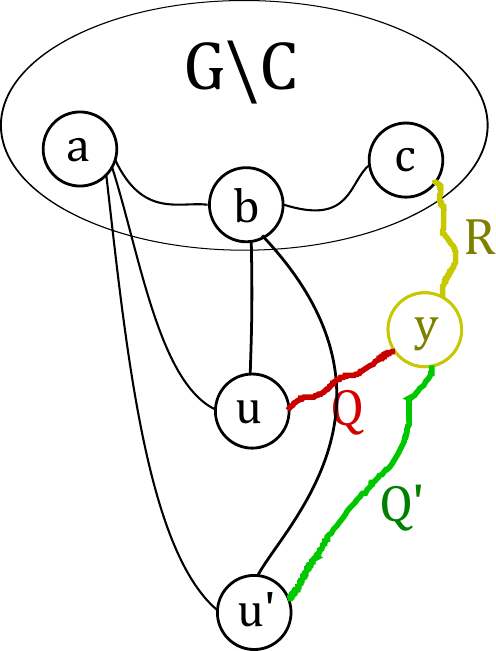}
	\caption{Case where the paths $Q$ and $Q'$ intersect.}
	\label{fig:notdisjoint-path-k33}
\end{figure}

\item $Q$ and $Q'$ intersect (see Figure~\ref{fig:notdisjoint-path-k33} for an illustration).
Let $y \in (Q \cap Q') \setminus c$ be such that the $uy$-subpath of $Q$ does not intersect $Q'$. 
Let $R$ be the $yc$-subpath of $Q'$.
We may assume w.l.o.g. that $R \subseteq Q \cap Q'$ for the remaining of the proof, whence $Q \cap Q' = R$.
Let us contract $Q \setminus R, Q' \setminus R, R \setminus c$ in order to make vertices $u,u',c$ adjacent to vertex $y$, then let us contract an arbitrary component of $G \setminus (P \cup C)$ in order to make vertices $a,c$ adjacent.
One obtains from $a,b,c,u,u',y$ a $K_{3,3}$-minor with $\{a,b,y\}$ being one side of the bipartition, that contradicts the fact that $G$ is planar.
\end{itemize}
\end{proof}

\begin{lemma}\label{lem:common-neighbour}
Let $G$ be a prime planar graph, let the path $\Pi=(a,b,c)$ be a separator of $G$, and let $C$ be a component of $G \setminus \Pi$.
Suppose there is some vertex $v \in C$ that is a common neighbour of $a,b,c$.
Then, either $C$ is reduced to $v$, or $(a,v,c)$ is a separator of $G$.
Furthermore, in the latter case, the path $(a,v,c)$ separates vertex $b$ from $C \setminus v$.
\end{lemma}

\begin{proof}
Let us assume that $C \setminus v \neq \emptyset$.
Let $D$ be a connected component of $G[C\setminus v]$.
Note that $v \in N(D)$ because $C$ is a connected component of $G \setminus \Pi$ by the hypothesis.
To prove the lemma, it suffices to prove that $b \notin N(D)$.
By contradiction, suppose that $b \in N(D)$ (see Figure~\ref{fig:common-neighbour-to-path} for an illustration).
Since $v,b,a$ and $v,b,c$ are pairwise adjacent and $G$ has no clique-separator by the hypothesis, then necessarily $N(D) = \{a,b,c,v\}$.

\begin{figure}[h!]
	\centering
	\includegraphics[width=0.2\textwidth]{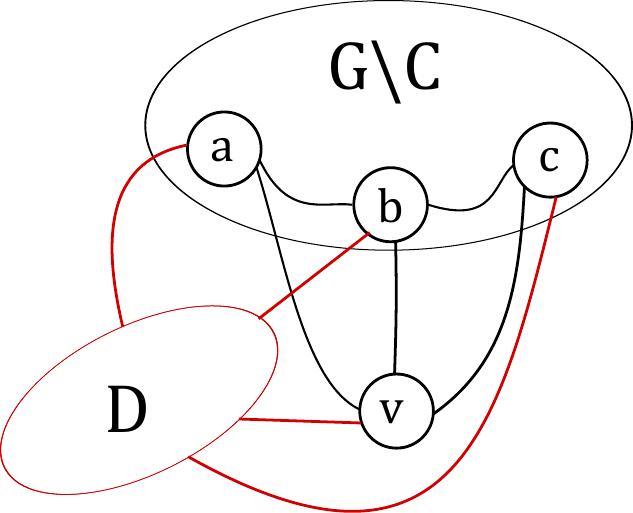}
	\caption{Existence of a component $D \subseteq C \setminus v$ that is adjacent to $b$.}
	\label{fig:common-neighbour-to-path}
\end{figure}

Let us contract the component $D$ to a single vertex $x$.
Then, let $C'$ be any component of $G \setminus (\Pi \cup C)$.
We have that $a,c \in N(C')$ or else $N(C')$ would be a clique-separator of $G$ (either a vertex-separator or an edge-separator).
So, let us contract the component $C'$ onto vertex $a$ in order to make $a$ and $c$ adjacent.
One obtains from $a,b,c,v,x$ a $K_5$-minor, which contradicts the fact that $G$ is planar.
\end{proof}

We recall that the gist of Algorithm \texttt{Leaf-BottomUp} is (informally) to try to remove a leaf-vertex $v$ from $G$ then to apply recursively the algorithm on $G \setminus v$.
Because the algorithm is strongly dependent on the fact that $G$ is prime, it is important to characterize the cases when $G \setminus v$ is also prime. 
Indeed, new clique-decompositions are needed when $G \setminus v$ is not prime, which may provoke a combinatorial explosion of the number of subgraphs to be considered.
Therefore, before we conclude this section, let us characterize whenever there may be clique-separators in $G \setminus v$ with $v$ being a leaf-vertex.
This will first require the following lemma.

\begin{lemma}\label{lem:separatorInCompo}
Let $G=(V,E)$ be a graph and let the path $\Pi=(a,b,c)$ be a separator of $G$. 
Let $C$ be the union of some components of $G \setminus \Pi$ and let $S$ be a separator of $G[C \cup \Pi]$. 
Then, $S$ is a separator in $G$ or $S$ separates $a$ and $c$ in $G[C \cup \Pi]$.

Moreover, in the latter case, $G[C \cup \Pi] \setminus S$ has exactly two components $C_a$ and $C_c$ containing $a$ and $c$ respectively.
\end{lemma}

\begin{proof}
There are two cases.
\begin{itemize}
\item Suppose there exists a component $D$ of $G[C\cup \Pi]	\setminus S$ such that $N_G(D) \subseteq C \cup \Pi$.
Since $D \cap S = \emptyset$, $S \neq V \setminus D$ and $N_G(D) \subseteq S$ therefore $S$ is a separator of $G$ with $D$ being a component of $G \setminus S$.
\item Else, every component $D$ of $G[C\cup \Pi]	\setminus S$ has a neighbour in $V \setminus (C \cup \Pi)$.
The latter implies that $D \cap \Pi \neq \emptyset$ for every component $D$ of $G[C\cup \Pi]	\setminus S$ because $C$ is a union of components of $G \setminus \Pi$ by the hypothesis.
In particular, since there exist at least two components of $G[C\cup \Pi]	\setminus S$ then there must be one containing an endpoint of $\Pi$.
W.l.o.g. assume there is a component $C_a$ of $G[C\cup \Pi]	\setminus S$ such that $a \in C_a$.
Let $C_c$ be any component of $G[C\cup \Pi]	\setminus (S \cup C_a)$.
We have that $C_c \cap N[C_a] = \emptyset$ because $S$ separates $C_a$ and $C_c$ in $G[C \cup \Pi]$.
Therefore, $a,b \notin C_c$ and so, $c \in C_c$ because $C_c \cap \Pi \neq \emptyset$.
This finally proves that $G[C \cup \Pi] \setminus S$ has exactly two components $C_a$ and $C_c$ containing $a$ and $c$ respectively. 
\end{itemize}		
\end{proof}

\begin{theorem}\label{th:GminusVprime?}
Let $G = (V,E)$ be a prime planar graph, let $v$ be a leaf-vertex of Type either 2 or 3 and let $\Pi_v = (a_v,b_v,c_v)$ be as defined in Definition~\ref{def:leafVertex}.
Suppose that there exists a minimal separator $S$ in $G \setminus v$ that is a clique.
Then, $S = \{u_v,b_v\}, u_v \notin \Pi_v$ and the following hold:
\begin{itemize}
\item $V \setminus (\Pi_v \cup \{v\})$ is a full component of $G \setminus \Pi_v$.
\item If $u_v \in N(a_v)$ (resp. $u_v \in N(c_v)$), then $a_v$ (resp. $c_v$) is simplicial in $G \setminus v$ with neighbours $\{u_v,b_v\}$;
\item Furthermore $u_v \notin N(a_v) \cap N(c_v)$ unless $V = \Pi_v \cup \{u_v,v\}$;
\end{itemize}
\end{theorem}

\begin{figure}[h!]
 \centering
 \includegraphics[width=0.75\textwidth]{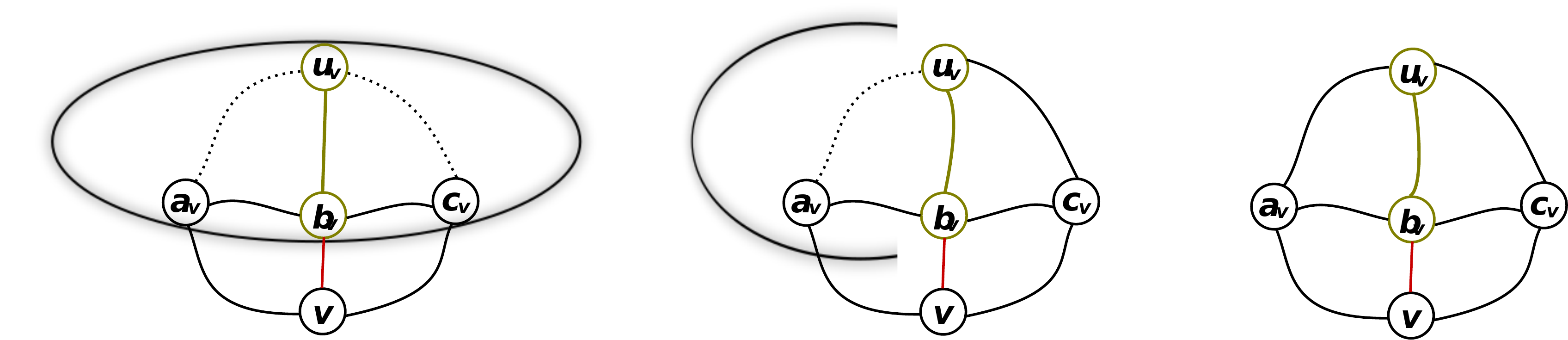}
 \caption{Existence of a clique-separator in $G \setminus v$.}
 \label{fig:clique-reappear}
\end{figure}

\begin{proof}
Note that the subgraph $G \setminus v$ is planar and $S$ is a minimal separator of $G \setminus v$ by the hypothesis, therefore by Corollary~\ref{cor:sep-planar} either $S$ induces a cycle or a forest of path.
Since in addition $S$ is a clique by the hypothesis, it follows that $S$ either induces a singleton, an edge or a triangle.
Since $S$ is a clique and $G$ is prime, $S$ is not a separator of $G$, so by Lemma~\ref{lem:separatorInCompo} with $C = V \setminus (\Pi_v \cup \{v\})$, $S$ is an $a_vc_v$-separator of $G \setminus v$.
This both implies that $b_v \in S$ and $S' := S \cup \{v\}$ is a minimal $a_vc_v$-separator of $G$.
In particular, $S$ being a strict subset of some minimal separator of $G$ it cannot induce a cycle (by Corollary~\ref{cor:sep-planar}), hence it must induce either a singleton or an edge.
Furthermore, still by Lemma~\ref{lem:separatorInCompo} with $C = V \setminus (\Pi_v \cup \{v\})$ there exist exactly two components $C_a,C_c$ of $G \setminus (S \cup \{v\})$, with $a_v \in C_a, c_v \in C_c$.
As a result, $S \setminus b_v \neq \emptyset$, or else $\{a_v,b_v\},\{b_v,c_v\}$ would be edge-separators of $G$, thus contradicting the hypothesis.
Let $S = \{u_v,b_v\}, u_v \notin \Pi_v \cup \{v\}$.
If $u_v \in N(a_v)$, then $C_a \setminus a_v = \emptyset$ (and so, $a_v$ is simplicial in $G \setminus v$), for otherwise $(a_v,u_v,b_v)$ would be a clique-separator of $G$.
Similarly, if $u_v \in N(c_v)$ then $C_c \setminus c_v = \emptyset$ (and so, $c_v$ is simplicial in $G \setminus v$).
In particular if $u_v \in N(a_v) \cap N(c_v)$ then $\Pi_v \cup \{u_v,v\} = V$.
Last, as there exists an $a_vc_v$-path in every component $C'$ of $G \setminus (\Pi_v \cup \{v\})$ because $G$ has no clique-separator by the hypothesis, therefore $u_v \in C'$.
This implies that $V \setminus (\Pi_v \cup \{v\})$ is a full component of $G \setminus \Pi_v$.
\end{proof}

\subsubsection{Constrained star-decompositions}

In the following, it will be useful to impose additional structure on the star-decompositions.
In order to do that, we will prove properties on some pairs of vertices in the graph.
Namely, we will prove that when $x,y \in V$ satisfy a few technical conditions, then it can be assumed that $T_x \cup T_y$ is a subtree of the star-decomposition $(T,{\cal X})$.

\begin{lemma}\label{lem:strong-pair}
Let $G$ be a connected graph with $tb(G) = 1$, let $x,y \in V(G)$ be non-adjacent (and $x \neq y$).

Suppose the pair $(x,y)$ satisfies that for every $xy$-separator $S$ of $G$, if there is $z \notin \{x,y\}$ that dominates $S$ then $z \in N_G(x) \cap N_G(y)$.

Then, there is a star-decomposition $(T,{\cal X})$ of $G$ with $B_x,B_y \in {\cal X}$, $x \in B_x, y \in B_y$ and either $B_x = B_y$ or $B_x,B_y$ are adjacent in $T$.
Moreover, in the latter case, $B_x \subseteq N[x], B_y \subseteq N[y]$.
\end{lemma}

\begin{proof}
Consider a star-decomposition $(T,{\cal X})$ of $G$, that exists by Lemma~\ref{lem:dominator-in-bag}. 
If $x$ and $y$ are not in a same bag, let $B_x$ and $B_y$ be the bags containing respectively $x$ and $y$ and as close as possible in $T$. 
By the properties of a tree decomposition, $N(x) \cap N(y) \subseteq B_x \cap B_y$. 
Hence, for any bag $B$ between $B_x$ and $B_y$ in $T$, $N(x) \cap N(y) \subseteq B$.

\begin{itemize}
\item
 {\bf Case 1:} If $B_x$ and $B_y$ are not adjacent in $T$, let $B$ be any bag in the path between $B_x$ and $B_y$ in $T$.
 By the properties of a tree decomposition, $B$ is an $xy$-separator.
 Moreover, let $z \in B$ dominate the bag, by the hypothesis $z \in N(x) \cap N(y)$ because $x,y \notin B$.
 As a result, adding $x$ and $y$ in each bag $B$ between $B_x$ and $B_y$ achieves a star-decomposition of $G$ that has a bag containing both $x,y$.

\medskip
\noindent 
 \item 
{\bf Case 2:} Now, let us assume that $B_x$ and $B_y$ are adjacent in $T$. 
Note that, if $B_x \subseteq N[z]$ for some $z \in N(x) \cap N(y)$ (resp., if $B_y \subseteq N[z]$ for some $z \in N(x) \cap N(y)$) the result holds. 
Indeed, adding $y$ in $B_x$ (resp., $x$ in $B_y$) achieves a star-decomposition of $G$ that has a bag containing both $x,y$.

So, let us consider the case when none of the two bags $B_x, B_y$ is dominated by a vertex of $N(x) \cap N(y)$.
Then, $B_x \setminus x$ and $B_y \setminus y$ are $xy$-separators by the properties of a tree decomposition.
Let $z_x \in B_x, z_y \in B_y$ satisfy $B_x \subseteq N[z_x]$ and $B_y \subseteq N[z_y]$.
By the hypothesis, $z_x \in \{x\} \cup (N(x) \cap N(y))$ and $z_y \in \{y\} \cup (N(x) \cap N(y))$.
Thus it follows that $z_x = x$ and $z_y = y$ (or else, we are back to Case 1).
Note that $B_x \cap B_y = N(x) \cap N(y)$ in such a case.
\end{itemize}
\end{proof}

We will mostly use the following two weaker versions of Lemma~\ref{lem:strong-pair} in our proofs.

\begin{corollary}
\label{cor:strong-sep}
Let $G$ be a connected graph with $tb(G) = 1$, let $x,y \in V(G)$ be non-adjacent (and $x \neq y$).

Suppose there exists a minimal separator $S\subseteq (N(x) \cap N(y)) \cup \{x,y\}$ in $G$ and $\{x,y\} \subseteq S$.

Then, there is a star-decomposition $(T,{\cal X})$ of $G$ with $B_x,B_y \in {\cal X}$, $x \in B_x, y \in B_y$ and either $B_x = B_y$ or $B_x,B_y$ are adjacent in $T$.
Moreover, in the latter case, $B_x \subseteq N[x], B_y \subseteq N[y]$.
\end{corollary}

\begin{proof}
We claim that for every $xy$-separator $S'$ of $G$, if there is $z \notin \{x,y\}$ such that $S' \subseteq N[z]$ then $z \in N(x) \cap N(y)$.
Observe that if the claim holds, then the corollary follows from Lemma~\ref{lem:strong-pair}.
To prove the claim, let $S \subseteq (N(x) \cap N(y)) \cup \{x,y\}$ be a separator of $G$ and $\{x,y\} \subseteq S$, that exists by the hypothesis.
Note that for any full component $C$ of $G \setminus S$, the $xy$-separator $S'$ must contain some vertex in $C$.
Since there are at least two full components of $G \setminus S$, then $z \in S \setminus (x,y) \subseteq N(x) \cap N(y)$, that finally proves the claim.
\end{proof}

So far, the two above results in this section (Lemma~\ref{lem:strong-pair} and Corollary~\ref{cor:strong-sep}) apply to \emph{general} graphs with tree-breadth one.
However, we will need the fact that the graph is planar for the following corollary. 

\begin{corollary}
\label{cor:strong-three}
Let $G$ be a connected graph with $tb(G) = 1$, let $x,y \in V(G)$ be non-adjacent (and $x \neq y$).

Suppose $G$ is $K_{3,3}$-minor-free and $|N_G(x) \cap N_G(y)| \geq 3$.

Then, there is a star-decomposition $(T,{\cal X})$ of $G$ with $B_x,B_y \in {\cal X}$, $x \in B_x, y \in B_y$ and either $B_x = B_y$ or $B_x,B_y$ are adjacent in $T$.
Moreover, in the latter case, $B_x \subseteq N[x], B_y \subseteq N[y]$.
\end{corollary}

\begin{proof}
We claim that for every $xy$-separator $S$ of $G$, if there is $z \notin \{x,y\}$ such that $S \subseteq N[z]$ then $z \in N(x) \cap N(y)$.
Observe that if the claim holds, then the corollary follows from Lemma~\ref{lem:strong-pair}.
To prove the claim, first recall that $|N(x) \cap N(y)| \geq 3$.
Since vertex $z$ dominates $S$ and $S$ is an $xy$-separator, therefore, $z$ dominates $N(x) \cap N(y)$ because $N(x) \cap N(y) \subseteq S$. 
In such case, $z \in N(x) \cap N(y)$, or else, $G$ admits a $K_{3,|N(x) \cap N(y)|}$-minor with $\{x,y,z\}$ and $N(x) \cap N(y)$ being the respective of the bipartition, which contradicts the hypothesis.
\end{proof}

Before we conclude this section, let us emphasize a useful consequence of Corollary~\ref{cor:strong-sep} regarding minimal $2$-separators.

\begin{lemma}\label{lem:2-separator}
Let $G=(V,E)$ with $tb(G) = 1$, let $x,y \in V$ be non-adjacent such that $S = \{x,y\}$ is a minimal separator of $G$ ($x \neq y$).
For every full component $C$ of $G \setminus S$, we have that $N(x) \cap N(y) \cap C \neq \emptyset$.
\end{lemma}

\begin{proof}
Let $(T,{\cal X})$ be a star-decomposition of $G$, that exists by Lemma~\ref{lem:dominator-in-bag}, minimizing the distance in $T$ between the subtrees $T_x$ and $T_y$ (respectively induced by the bags containing $x$ and $y$).
There are two cases.
\begin{itemize}
\item First, suppose that $T_x \cap T_y \neq \emptyset$.
For any full component $C$ of $G \setminus S$, let $T_C$ be the subtree that is induced by all bags intersecting $C$.
Because $C$ is a full component, there must be an edge between $x$ and a vertex of $C$, and this edge is in a bag of $T_x \cap T_C$.
Similarly, there must be an edge between $y$ and a vertex of $C$, and this edge is in a bag of $T_y \cap T_C$.
As a result, the subtrees $T_x,T_y,T_C$ are pairwise intersecting, and so by the Helly property (Lemma~\ref{lem:helly}) $T_x \cap T_y \cap T_C \neq \emptyset$ {\it i.e.}, there exists a bag $X_t$ which contains $S$ and it intersects $C$.
Let $z \in X_t$ dominate the bag.
Note that $z \in C \cup S$ because it has to dominate some vertices in $C$ and so, it cannot be in $V \setminus (C \cup S)$.
Furthermore, recall that $x,y$ are non-adjacent by the hypothesis.
Therefore, $z \in C \cap N(x) \cap N(y)$, and the result holds for any full component $C$ of $G \setminus S$.
\item Else, since $S = \{x,y\}$ is a minimal separator and we assume $(T,{\cal X})$ to minimize the distance in $T$ between $T_x$ and $T_y$, by Corollary~\ref{cor:strong-sep} there are two adjacent bags $B_x,B_y$ such that $x \in B_x \setminus B_y$ dominates $B_x$, $y \in B_y \setminus B_x$ dominates $B_y$.
Since $B_x \cap B_y$ is an $xy$-separator by the properties of a tree-decomposition, then $B_x \cap B_y \cap C \neq \emptyset$ for every full component $C$ of $G  \setminus S$, that is $N(x) \cap N(y) \cap C \neq \emptyset$.
\end{itemize}\end{proof}

\subsubsection{Bounded Treewidth}

Independently from Algorithm \texttt{Leaf-BottomUp}, let us introduce in this section another property of (not necessarily prime) planar graphs with tree-breadth one.
More precisely, we prove these graphs have bounded treewidth.
To prove this property, we will use the same terminology as for the previous subsections. 

\begin{lemma}
\label{lem:tw-planar-tb-1}
Let $G$ be planar with $tb(G) \leq 1$.
Then, $tw(G) \leq 4$ and the upper-bound is sharp.
\end{lemma}

\begin{figure}[h!]
 \centering
 \includegraphics[width=0.25\textwidth]{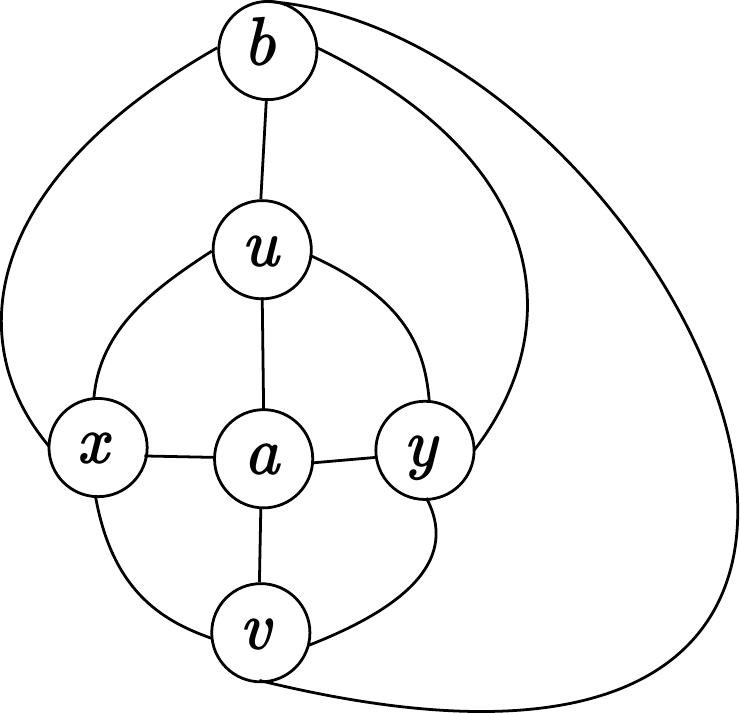}
 \caption{A planar graph $G$ with $tb(G) = 1$ and $tw(G) = 4$.}
 \label{fig:tb-1-tw-4}
\end{figure}

\begin{proof}
The treewidth of $G$ is the maximum treewidth of its atoms~\cite{Bodlaender2005}, so, let us assume $G$ to be a prime planar graph.
Let $(T,{\cal X})$ be any star-decomposition of $G$, the graph $H = (V,\{ \{u,v\} \mid T_u \cap T_v \neq \emptyset  \})$ is chordal.
Furthermore, if $H'$ is a chordal graph with same vertex-set $V$ and such that $E(G) \subseteq E(H') \subseteq E(H)$ then any clique-tree of $H'$ is still a star-decomposition of $G$.
Therefore, we will assume w.l.o.g. that $H$ is a minimal triangulation of $G$ and $(T,{\cal X})$ is a clique-tree of $H$ (in particular, $(T,{\cal X})$ is reduced). 
Additional properties of $(T,{\cal X})$ will be deduced from the latter assumption about $H$ by using the results from~\cite{Bouchitte2002}.
Let us now prove the lemma by induction on $|V(G)|$ (the base-case of the graph with a single vertex is trivial).
\begin{itemize}
\item If $|{\cal X}| = 1$, then $G$ has some universal vertex $u$.
Furthermore, since $G$ is planar therefore, $G \setminus u$ is outerplanar~\cite{Syslo1979}.
Consequently, $tw(G \setminus u) \leq 2$~\cite{Bodlaender1998}, so, $tw(G) \leq 3$.
\item Suppose $|{\cal X}| = 2$. 
Let ${\cal X} = \{B,B'\}$.
Since $(T,{\cal X})$ is assumed to be a clique-tree of some minimal triangulation $H$ of $G$, therefore, $B \cap B'$ is a minimal separator~\cite{Bouchitte2002}.
Let us remind that by Corollary~\ref{cor:make-sep-cyc} there is a planar supergraph $G'$ of $G$ with same vertex-set so that $B \cap B'$ induces either an edge or a cycle of $G'$.
Furthermore $(T,{\cal X})$ is also a star-decomposition of $G'$, so, $tb(G') \leq 1$.
In addition, $tw(G) \leq tw(G')$.
Recall that we can further assume $G'$ to be prime (or else, we apply the induction hypothesis on the atoms of $G'$), hence $B  \cap B'$ induces a cycle of $G'$ of length at least four.
Let $B \Delta B' = (B \setminus B') \cup (B' \setminus B)$.
%Observe that $B \setminus B' \neq \emptyset$ and $B' \setminus B \neq \emptyset$ because $(T,{\cal X})$ is assumed to be reduced, thus $|B \Delta B'| \geq 2$.
Since $H$ is assumed to be a minimal triangulation and $B,B'$ are leaves of a clique-tree of $H$, therefore, $B \setminus B'$ is a (nonempty) dominating clique of the subgraph $G[B]$, and similarly $B' \setminus B$ is a (nonempty) dominating clique of the subgraph $G[B']$~\cite{Bouchitte2002}.
Thus, every vertex $u \in B \Delta B'$ satisfies $B \cap B' \in N(u)$.
Since $|B \cap B'| \geq 4$ because $B \cap B'$ induces a cycle of $G'$ of length at least four, therefore, $|B \Delta B'| \leq 2$ or else there would be a $K_{3,3}$-minor of $G'$ with any three vertices of $B \cap B'$ being one part of the bipartition.
%This finally proves that $|B \Delta B'| = 2$.
As a result, since $B \cap B'$ induces a cycle, $tw(G) \leq tw(G') \leq 2 + |B \Delta B'| = 4$.
\item Finally, suppose $|{\cal X}| \geq 3$. 
Let $t \in V(T)$ be an internal node, by the properties of a tree decomposition the bag $X_t$ is a separator of $G$.
Let $b \in X_t$ satisfy $X_t \subseteq N_G[b]$.
Since $G$ is prime and so, biconnected, therefore $X_t \setminus b$ is a separator of $G \setminus b$.
In such case, let us remind by Lemma~\ref{lem:P3exists} that there exist $a,c \in X_t \setminus b$ non-adjacent such that $\{a,c\}$ is a minimal separator of $G \setminus b$.
In particular, the path $\Pi = (a,b,c)$ is a separator of $G$.
Let $C_1, C_2, \ldots, C_l$ be the components of $G \setminus \Pi$.
For every $1 \leq i \leq l$, let $G_i$ be obtained from $G[C_i \cup \Pi]$ by making the two endpoints $a,c$ of $\Pi$ adjacent.
Note that $G_i$ can be obtained from $G$ by edge-contractions (because $G$ is prime and so, $a,c \in N(C_j)$ for every $1 \leq j \leq l$), therefore, $tb(G_i) \leq 1$ because tree-breadth is stable under edge-contractions (Lemma~\ref{lem:contraction-closed}).
In addition, $tw(G) \leq \max_i tw(G_i)$ because $\Pi$ induces a triangle in every graph $G_i$ by construction.
As a result, for every $1 \leq i \leq l$, since $|V(G_i)| < |V(G)|$ by construction, therefore $tw(G_i) \leq 4$ by the induction hypothesis, whence $tw(G) \leq 4$.
\end{itemize}
Let $G$ be constructed from the cycle $(u,v,x,y)$ of length four by adding two vertices $a,b$ such that $N_G(a) = N_G(b) = \{u,x,v,y\}$ (see Figure~\ref{fig:tb-1-tw-4} for an illustration).
Since there exists a star-decomposition of $G$ with two bags (respectively dominated by $a,b$), $tb(G) \leq 1$.
Moreover, $G$ is $4$-regular by construction, therefore $tw(G) \geq 4$~\cite{Bodlaender2005}.
This proves the sharpness of the upper-bound. 
\end{proof}

Note that since it is well-known that many difficult problems can be solved on bounded-treewidth graphs in linear-time, therefore, it may be the case that the recognition of planar graphs with tree-breadth at most one can be simplified by using Lemma~\ref{lem:tw-planar-tb-1}.
However, we were unable to find a way to use it in our proofs (actually, the star-decomposition that can be computed using our algorithm may have unbounded width --- because of leaf-vertices of Type 1).

\subsection{Correctness of Algorithm \texttt{Leaf-BottomUp} }

\subsubsection{Existence of a $P_3$-separator}

As a first step to prove correctness of Algorithm \texttt{Leaf-BottomUp}, let us prove correctness of~\ref{step:find-leaf-vertex}.
That is, we will prove that for every planar graph $G$ with $tb(G) = 1$, $G$ contains a leaf-vertex or $G$ admits a star-decomposition with at most two bags. 

To prove this step, we will prove additional properties of the minimal separators of prime planar graphs with tree-breadth one.
In the following, let ${\cal P}_3(G)$ be the set of (not necessarily minimal) separators of $G$ that induce paths of length $2$ (we will call them $P_3$-separators since they have three vertices).
We will distinguish the case when ${\cal P}_3(G) \neq \emptyset$ from the case when ${\cal P}_3(G) = \emptyset$.

\begin{theorem}\label{th:noP3}
Let $G$ be a prime planar graph with $tb(G) = 1$. If ${\cal P}_3(G) \neq \emptyset$, then $G$ has a leaf-vertex.
\end{theorem}

\begin{proof}
Let $\Pi =(a,b,c) \in {\cal P}_3(G)$ minimize the size of a smallest component of $G \setminus \Pi$.
We recall that $\{a,c\} \notin E(G)$ because $G$ is assumed to be prime by the hypothesis (the latter fact will be used in the following).
Let $C$ be any component of $G \setminus \Pi$ of minimum size.
Our aim is to prove the existence of some leaf-vertex $v \in C$ (the latter dominating the component $C$), that will prove Theorem~\ref{th:noP3}.

\begin{claim}
\label{claim:path-min-pty}
There do not exist $\Pi' \subseteq \Pi \cup C, \ C' \subset C$ such that $\Pi' \in {\cal P}_3(G)$ and $C'$ is a component of $G \setminus \Pi'$.	
\end{claim}

\begin{proofclaim}
The claim follows from the minimality of $C$.	
\end{proofclaim}

We will often use Claim~\ref{claim:path-min-pty} in the remaining of the proof.

Let $(T,{\cal X})$ be a star-decomposition of $G$, that exists by Lemma~\ref{lem:dominator-in-bag}.
In particular, let $T_a,T_c$ be the subtrees that are respectively induced by the bags containing $a$ or $c$.
Assume w.l.o.g. that $(T,{\cal X})$ minimizes the distance in $T$ between the subtrees $T_a$ and $T_c$.
We will distinguish the case $T_a \cap T_c \neq \emptyset$ from the case $T_a \cap T_c = \emptyset$.

\paragraph{Case $T_a \cap T_c \neq \emptyset$.}
In such case, the subtrees $T_a,T_b,T_c$ are pairwise intersecting and so, by the Helly property (Lemma~\ref{lem:helly}) $T_a \cap T_b \cap T_c \neq \emptyset$.
Let us remove all vertices in $V \setminus (\Pi \cup C)$ from bags in $(T,{\cal X})$.
Let us call $(T,{\cal X}^C)$ the resulting tree decomposition of $G[\Pi \cup C]$.

\begin{claim}
\label{claim:induced-star-dec}
$(T,{\cal X}^C)$ has breadth one.
\end{claim}

\begin{proofclaim}
There are two cases to be considered.
\begin{itemize}
	\item If $b$ has some neighbour in $C$, then $C$ must be a full component of $G \setminus \Pi$, or else one of $\{a,b\},\{b,c\}$ should be a clique-separator thus contradicting the fact that $G$ is prime by the hypothesis.
	In such case, the claim follows from Lemma~\ref{lem:sep-keep}.
	\item Else, $b$ has no neighbour in $C$, and let $D$ be the connected component of $b$ in $G \setminus (a,c)$.
	Let $H$ be obtained from $G$ by contracting $D$ to $b$.
	By Lemma~\ref{lem:contraction-closed}, $tb(H) = 1$.
	Let $(T,{\cal X}^H)$ be the tree decomposition of breadth one of $H$ where for every $t \in V(T), \ X^H_t = X_t$ if $X_t \cap D = \emptyset, \ X^H_t = (X_t \setminus D) \cup \{b\}$ else.
	Moreover, since $b$ has no neighbour in $C$, $D \cap N_G[C] = \emptyset$ and so, $H[C \cup \Pi] = G[C \cup \Pi]$ by construction.
	Finally, since $\{b\}$ is a full component of $H \setminus (a,c)$, therefore, by Lemma~\ref{lem:sep-keep} applied to $H$, the tree decomposition $(T,{\cal X}^C)$ is indeed a tree decomposition of breadth one of $G[C \cup \Pi]$.
\end{itemize}
\end{proofclaim}

Let $(T',{\cal X}')$ be any reduced tree decomposition obtained from  $(T,{\cal X}^C)$.
We point out that $T'_a \cap T'_b \cap T'_c \neq \emptyset$ by construction (because $T_a \cap T_b \cap T_c \neq \emptyset$). 
Furthermore, since by Claim~\ref{claim:induced-star-dec}  $(T,{\cal X}^C)$ has breadth one, therefore $(T',{\cal X}')$ is a star-decomposition of $G[C \cup \Pi]$ by Lemma~\ref{lem:dominator-in-bag}.

We will prove that $C$ contains a leaf-vertex by contradiction. 
Informally, we will show, using the properties of the star-decomposition $(T',{\cal X}')$, that if it is not the case that $C$ contains a leaf-vertex, then ${\cal P}_3(G[C \cup \Pi]) \cap {\cal P}_3(G) \neq \emptyset$ and the latter contradicts Claim~\ref{claim:path-min-pty}.

\medskip
\noindent
In order to prove this, first note that $a$ has at least one neighbour in $C$ because $G$ is prime by the hypothesis (indeed, $(b,c)$ cannot be an edge-separator of $G$).
We now distinguish between several subcases.
\begin{itemize}
\item {\bf Case 1.} There is $u \in C$ such that $u \in N(a) \cap N(b) \cap N(c)$ ({\it e.g.}, see Figure~\ref{fig:fig:leaf-1}).
By Lemma~\ref{lem:common-neighbour}, either $C$ is reduced to $u$ or there exist $\Pi' = (a,u,c) \in {\cal P}_3(G), \ C' \subseteq C \setminus u$ and $C'$ is a component of $G \setminus \Pi'$.
The latter case contradicts Claim~\ref{claim:path-min-pty}, therefore, $C$ is reduced to $u$ and so $u$ is a leaf-vertex of Type 2.

\begin{figure}[h!]
 \centering
 \includegraphics[width=0.125\textwidth]{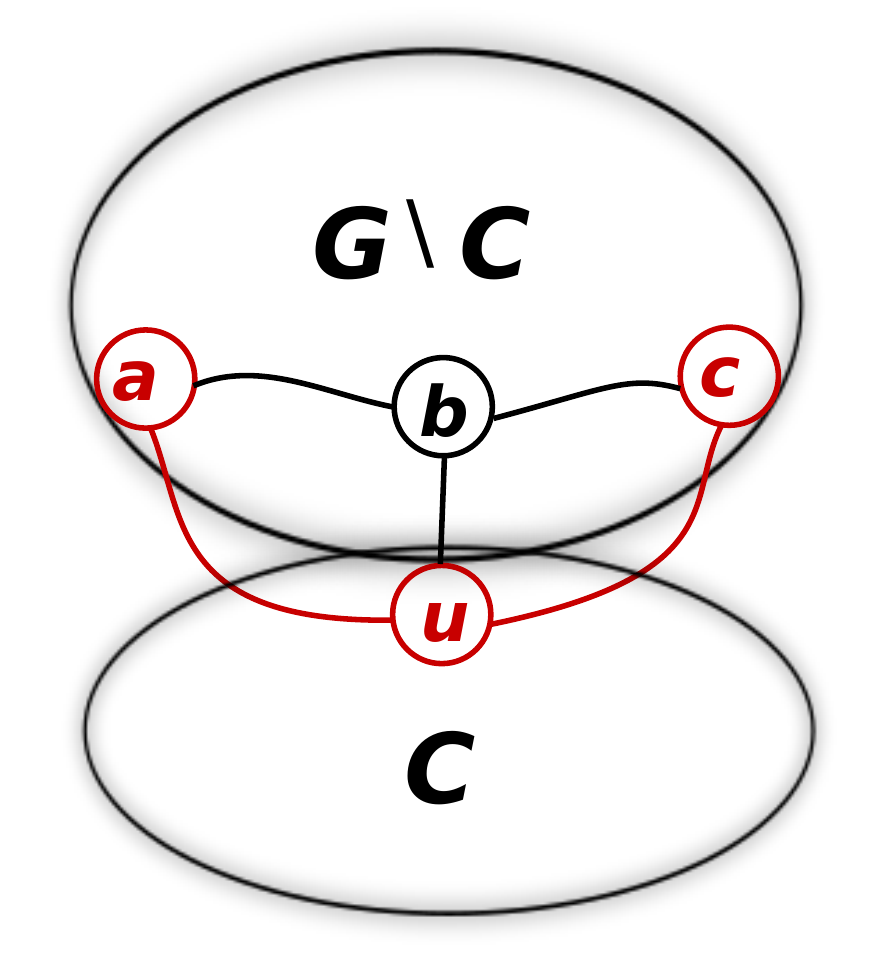}
 \caption{Case 1}
 \label{fig:fig:leaf-1}
\end{figure}

Thus, from now on let us assume that no such vertex $u$ exists.

\item {\bf Case 2.} By contradiction, assume $N(a) \cap C \subseteq N(b) \cap C$.
By Lemma~\ref{lem:unique-neighbour}, $| N(a) \cap N(b) \cap C | \leq 1$, so, $|N(a) \cap C| = 1$. 
Let $u \in N(a) \cap N(b) \cap C$ be the unique neighbour of vertex $a$ in $C$ (see Figure~\ref{fig:fig:leaf-2}).
Since in such case we can assume that $u \notin N(c)$ (for otherwise, we are back to Case 1), and vertex $c$ has some neighbour in $C$ because $G$ is prime (and so, $(a,b)$ cannot be an edge-separator of $G$), therefore, $C$ is not reduced to vertex $u$.
Then, $\Pi' = (u,b,c) \in \mathcal{P}_3(G)$ because it separates $a$ from $C \setminus u$, and so there is at least one component of $G \setminus \Pi'$ that is strictly contained into $C$ by construction.
This contradicts Claim~\ref{claim:path-min-pty}, so, Case 2 cannot occur.

\begin{figure}[h!]
 \centering
 \includegraphics[width=0.15\textwidth]{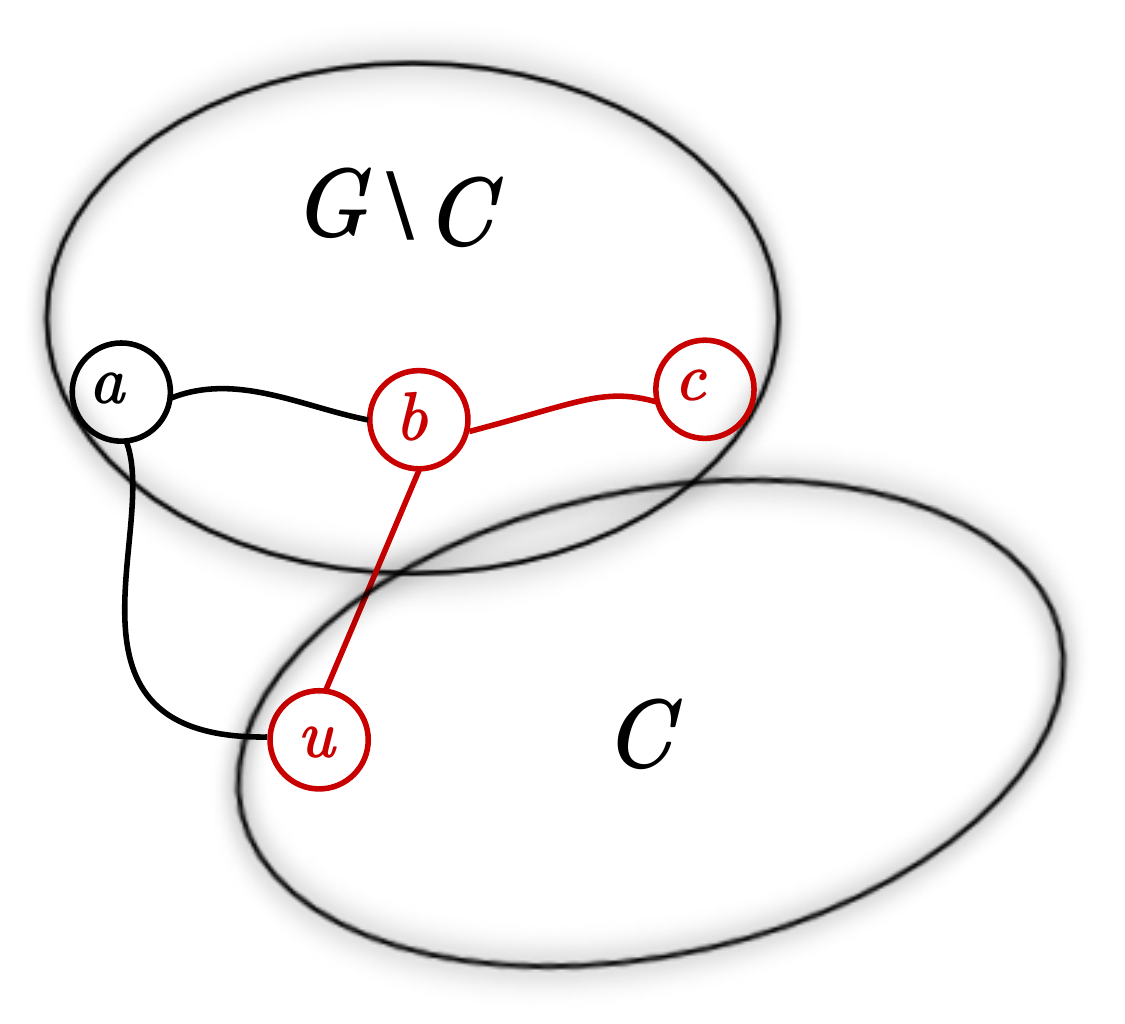}
 \caption{Case 2}
 \label{fig:fig:leaf-2}
\end{figure}

\item {\bf Case 3.} There is $u \in C$ satisfying $u \in N(a) \setminus N(b)$.
By the properties of a tree decomposition, there is some bag $B' \in T'_a \cap T'_u$.
Let $v \in B'$ dominate the bag $B'$.
By construction, $v \neq b$ because $u \in B' \setminus N(b)$, similarly $v \neq c$ because $a \in B' \setminus N(c)$.
We will also prove later that $v \neq a$.
Moreover, $\Pi \setminus N[v] \neq \emptyset$ (or else, we are back to Case 1), hence $\Pi \setminus B' \neq \emptyset$.
So let $B$ be the bag adjacent to $B'$ onto the unique path in $T'$ from $B'$ to $T'_a \cap T'_b \cap T'_c$ (we remind that the latter subtree is nonempty by construction).
By the properties of the tree decomposition $(T',{\cal X}')$, $B \cap B'$ is a separator of $G[C \cup \Pi]$.
Furthermore, $a \in B \cap B'$.
More generally $\Pi \cap B' \subseteq B \cap B'$ by construction, therefore $B \cap B'$ is also a separator of $G$ by Lemma~\ref{lem:separatorInCompo}.
Let $w \in B$ dominate this bag. 
Observe that $w \neq c$ because $a \in B \cap B'$.

We will prove that $v \in C$  and $v$ is a leaf-vertex.
In order to prove these two results, we will need to prove that $C \cup \Pi$ is fully contained into the two adjacent bags $B,B'$ (Claim~\ref{claim:subgraph-two-bags}).
The latter will require intermediate claims.

\begin{claim}
\label{claim:c-is-in-intersection}
$c \in B \cap B'$.
\end{claim}

\begin{proofclaim}
Assume for the sake of contradiction that $c \notin B \cap B'$ (see Figure~\ref{fig:fig:leaf-3}).
Then, $c \notin B'$ because $\Pi \cap B' \subseteq B \cap B'$ by construction.
We will prove that the latter contradicts Claim~\ref{claim:path-min-pty}.

Indeed, first observe that $G \setminus w$ is connected because $G$ is prime and so, biconnected, by the hypothesis.
In addition $(B \cap B') \setminus w$ is a (not necessarily minimal) separator of $G \setminus w$ because it separates $B' \setminus B$ from $c$.
Let $S \subseteq (B \cap B') \setminus w$ be a minimal separator of $G \setminus w$.
By Lemma~\ref{lem:P3exists}, there exist $x,y \in (B \cap B') \setminus w$ non-adjacent such that $S = \{x,y\}$, and so, $\Pi'=(x,w,y) \in \mathcal{P}_3(G)$.
Note that $\Pi' \neq \Pi$, because we assume that $c \notin B \cap B'$ and so $c \notin \{x,y\}$. 
Moreover, since $(T',{\cal X}')$ is a star-decomposition of $G[C\cup \Pi]$ by construction we have that $\Pi' \subseteq \Pi \cup C$, therefore $x \in C$ or $y \in C$, because $c \notin \{x,y\}$ and $a,b$ are adjacent whereas $x,y$ are non-adjacent.
W.l.o.g. let $x \in C$.

\begin{figure}[h!]
	\centering
	\includegraphics[width=0.25\textwidth]{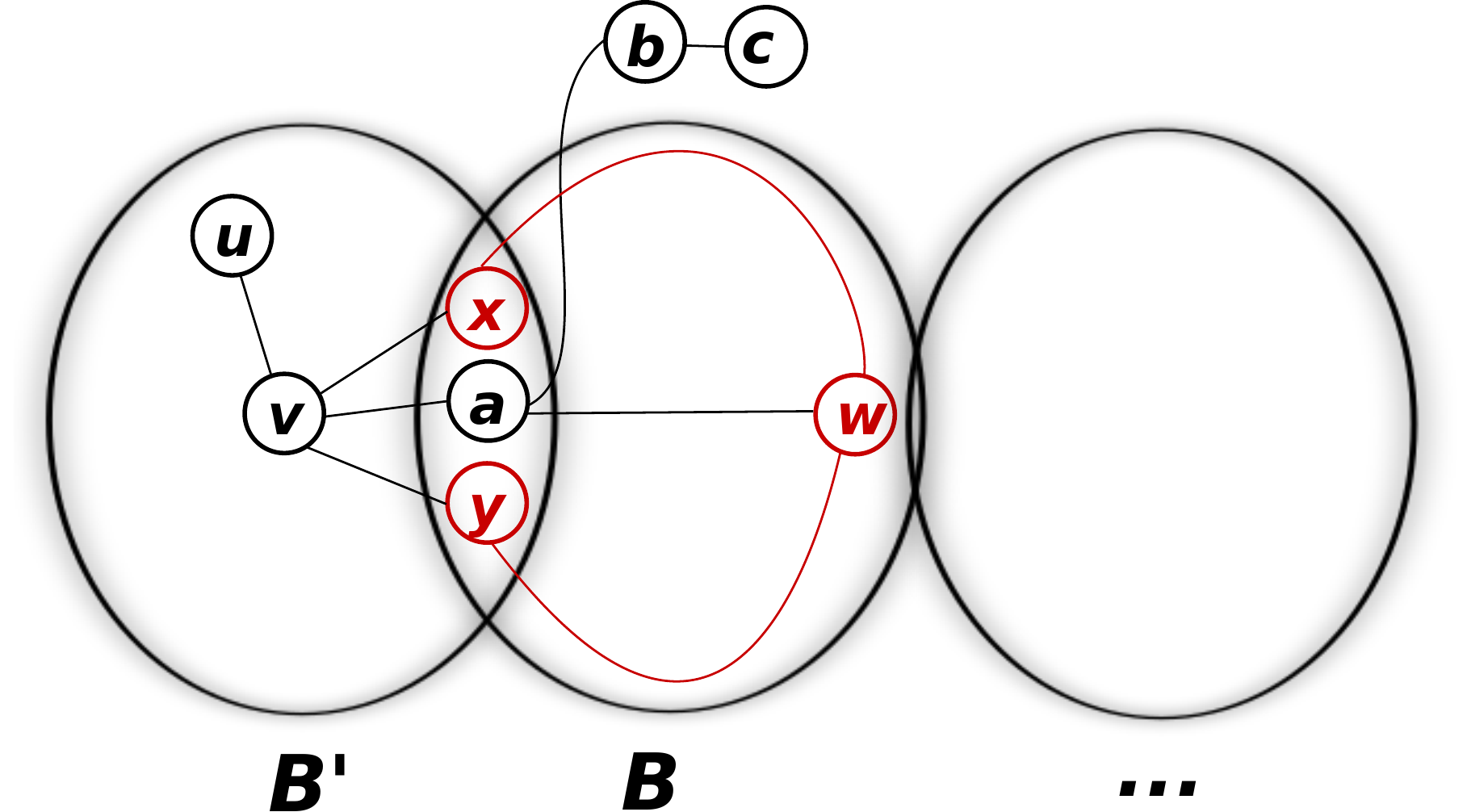}
	\caption{Case $c \notin B \cap B'$}
	\label{fig:fig:leaf-3}
\end{figure}

\begin{subclaim}
	\label{claim:new-path-separator}
	$\Pi'$ is not an $ac$-separator.
\end{subclaim}

\begin{proofsubclaim}	
We refer to Figure~\ref{fig:fig:leaf-3prime} for an illustration of the proof.	
Let $C'$ be any component of $G \setminus (\Pi \cup C)$.
Observe that $C'$ is fully contained into some component $D$ of $G \setminus \Pi'$, because $\Pi' \subseteq \Pi \cup C$.
In addition, $a,c \in N(C')$ because $G$ is prime by the hypothesis (and so, neither $a$ nor $b$ nor $c$ nor $(a,b)$ nor $(b,c)$ can be a separator of $G$).
In particular, since we assume $c \notin B \cap B'$ and so, $c \notin \Pi'$, therefore, $c \in D$.
As a result, either $a \in \Pi'$ or $a,c \in D$, that finally proves the subclaim. 
\end{proofsubclaim}

\begin{figure}[h!]
	\centering
	\includegraphics[width=0.2\textwidth]{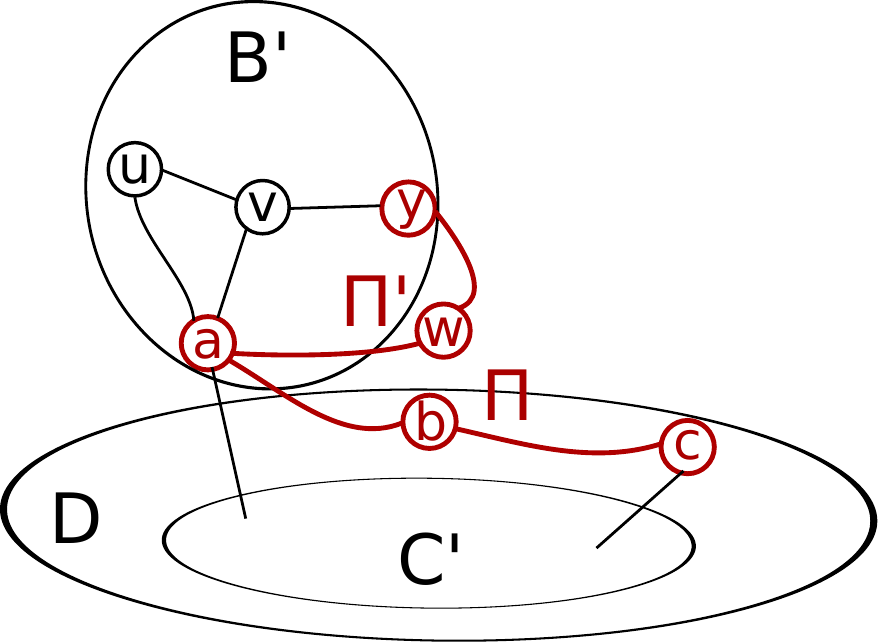}
	\caption{$\Pi'$ is not an $ac$-separator.}
	\label{fig:fig:leaf-3prime}
\end{figure}

Let $D$ be the component of $G \setminus \Pi'$ such that $c \in D$, that exists because we assume $c \notin B \cap B'$ and so, $c \notin \Pi'$.
Since $b,c$ are adjacent and $\Pi'$ is not an $ac$-separator by Claim~\ref{claim:new-path-separator}, therefore, $\Pi \subseteq \Pi' \cup D$.

Moreover, let us show that Claim~\ref{claim:new-path-separator} implies the existence of some $D' \subset C$ being a component of $G \setminus \Pi'$, thus contradicting Claim~\ref{claim:path-min-pty}.
Indeed, let $D'$ be any component of $G \setminus (\Pi' \cup D)$.
Since $G$ is prime by the hypothesis, $x$ has some neighbour in $D'$ and so, $D' \cap C \neq \emptyset$ because $x \in C$ and $\Pi \cap D' = \emptyset$ by construction.
But then, $D' \subseteq C \setminus x$, for the existence of some $z \in D' \setminus C$ would imply that $D' \cap \Pi \neq \emptyset$.

To sum up, we conclude that it must be the case that $c \in B \cap B'$.
\end{proofclaim}

We will use Claim~\ref{claim:c-is-in-intersection} to prove that $v \in C$, as follows:

\begin{claim}
\label{claim:v-in-c}
$v \in C$.
Furthermore, the two vertices $b,v$ are non-adjacent.
\end{claim}

\begin{proofclaim}
Recall that $v \in C \cup \Pi$ because $(T',{\cal X}')$ is a star-decomposition of $G[C \cup \Pi]$ by construction.	
So we will only need to prove that $v \notin \Pi$.
First, since $a \in B \cap B'$ by construction and  $c \in B \cap B'$ by Claim~\ref{claim:c-is-in-intersection}, therefore, $a,c \in B' \subseteq N[v]$.
The latter implies that $v \notin \{a,c\}$ because $a,c \in N[v]$ whereas $a,c$ are non-adjacent.
Furthermore, this implies $b \notin N[v]$ because we assume that $\Pi \not\subseteq N[v]$ (for otherwise, we are back to Case 1).
As a result, $v \notin \Pi$, whence $v \in C$. 
\end{proofclaim}

%In particular, $b \notin B'$.
Then, we will need the following technical claim in order to prove that $w=b$ (Claim~\ref{claim:w-equals-b}).

\begin{claim}
\label{claim:subgraph-prime}
$G[C \cup \Pi]$ is prime.
\end{claim}

\begin{proofclaim}
Suppose by contradiction there exists a clique-separator $S$ of $G[C \cup \Pi]$.
Then, $S$ could not be a separator of $G$ because $G$ is prime by the hypothesis.
By Lemma~\ref{lem:separatorInCompo}, the latter implies that $S$ is an $ac$-separator of $G[C \cup \Pi]$.
Therefore, the two vertices $b,v \in N(a) \cap N(c)$ must be in $S$, and so, since $b,v$ are non-adjacent by Claim~\ref{claim:v-in-c}, the latter contradicts the fact that $S$ is a clique.
\end{proofclaim} 

\begin{claim}
\label{claim:w-equals-b}
$w = b$.
\end{claim}

\begin{proofclaim}
Assume for the sake of contradiction $w \neq b$ (see Figure~\ref{fig:fig:leaf-4}).
We will prove that it contradicts Claim~\ref{claim:path-min-pty}.

Indeed, the graph $G[C \cup \Pi] \setminus w$ is connected because $G[C \cup \Pi]$ is prime by Claim~\ref{claim:subgraph-prime} and so, biconnected.
In addition, $(B \cap B') \setminus w$ is a (not necessarily minimal) separator of $G[C \cup \Pi] \setminus w$ because it separates $b$ from $B' \setminus B$ (recall that $b,v$ are non-adjacent by Claim~\ref{claim:c-is-in-intersection}, and so, $b \notin B' \subseteq N[v]$). 
Let $S \subseteq (B \cap B') \setminus w$ be a minimal separator of $G[C \cup \Pi] \setminus w$. 
By Lemma~\ref{lem:P3exists}, there exist $x,y \in (B \cap B') \setminus w$ non-adjacent such that $S = \{x,y\}$, and so, $\Pi'=(x,w,y) \in \mathcal{P}_3(G[C \cup \Pi])$.
Furthermore, $b \notin \Pi' \subseteq (B \cap B') \cup \{w\} \subseteq N[v] \cup \{w\}$ and so, $\Pi'$ cannot be an $ac$-separator of $G[C \cup \Pi]$, whence by Lemma~\ref{lem:separatorInCompo} $\Pi'$ is a separator of $G$, and so, $\Pi' \in \mathcal{P}_3(G)$.

\begin{figure}[h!]
	\centering
	\includegraphics[width=0.25\textwidth]{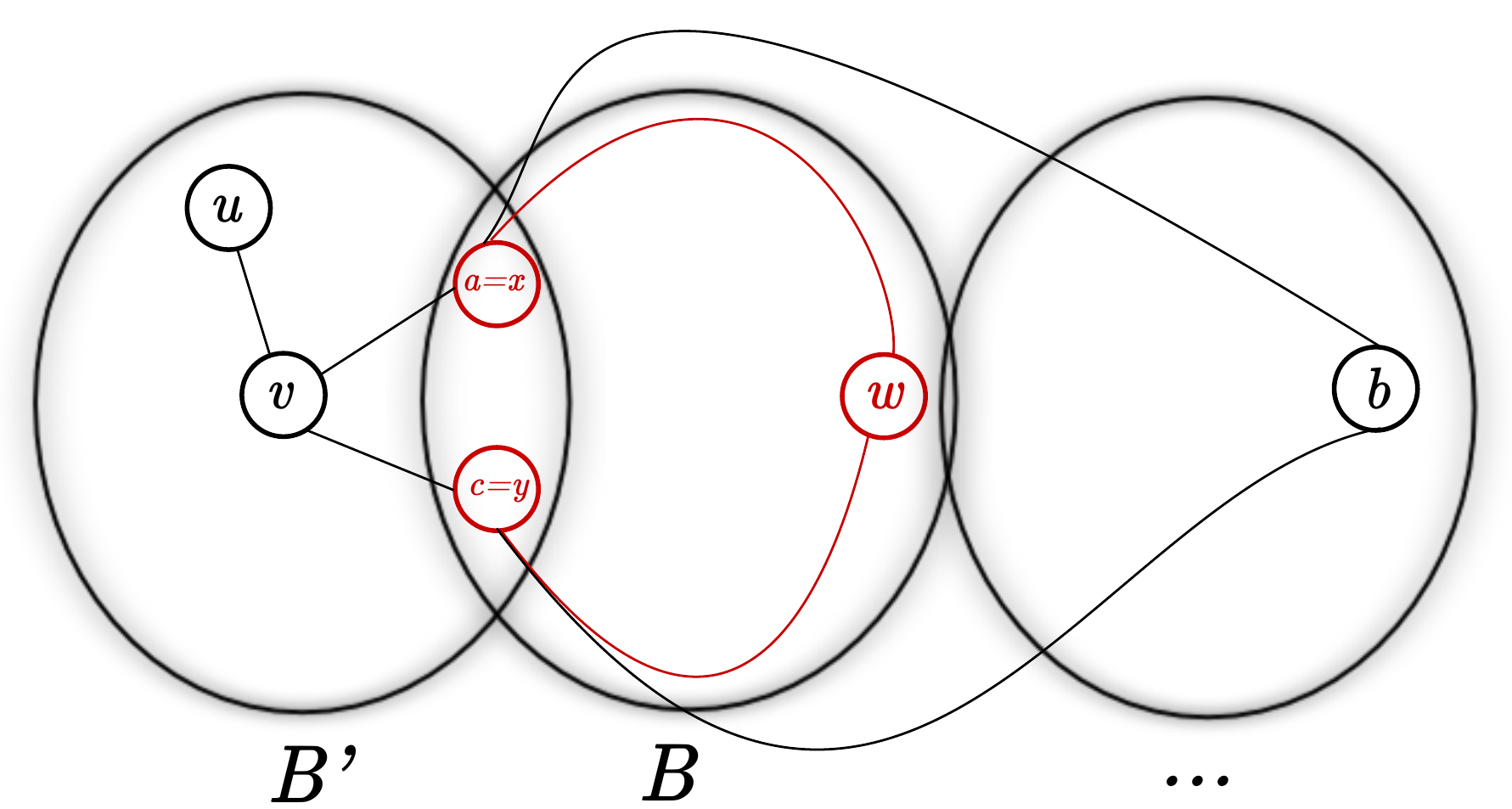}
	\caption{Case $c \in B \cap B'$ and $w \neq b$.}
	\label{fig:fig:leaf-4}
\end{figure}

Let $D \subseteq C \cup \Pi$ be the component of $G[C \cup \Pi] \setminus \Pi'$ containing vertex $b$.
Note that $\Pi \subseteq D \cup \Pi'$ because $\Pi'$ is not an $ac$-separator of $G[C \cup \Pi]$.
Let $D'$ be any component of $G[C \cup \Pi] \setminus (\Pi' \cup D)$, that exists because $\Pi' \in \mathcal{P}_3(G[C \cup \Pi])$.
Since $N(D') \subseteq C \cup \Pi$ by construction, $D' \cap \Pi = \emptyset$ by construction and $D'$ is a component of $G[C \cup \Pi] \setminus \Pi'$, therefore, $D'$ is also a component of $G \setminus \Pi'$.
The latter contradicts Claim~\ref{claim:path-min-pty} because $D' \subset C$.
\end{proofclaim}

Let $S = \{v,b\} \cup (B \cap B')$.
We are now able to prove that $S = C \cup \Pi$ (Claim~\ref{claim:subgraph-two-bags}).
That is, $C \cup \Pi$ is fully contained in the two adjacent bags $B,B'$ (respectively dominated by $b,v$).

\begin{claim}
\label{claim:subgraph-two-bags}
$S = C \cup \Pi$.
\end{claim}

\begin{proofclaim}
Assume by contradiction $S \neq C \cup \Pi$, let $D$ be a component of $G[C \cup \Pi] \setminus S$ (see Figure~\ref{fig:fig:leaf-5}).
Note that $D \subset C$ because $\Pi \subset S$ by construction.
Furthermore, $v,b \notin B \cap B'$ because $w = b$ by Claim~\ref{claim:w-equals-b} and $b \notin N(v)$ by Claim~\ref{claim:c-is-in-intersection}, so, $B \cap B'$ is a (minimal) $vb$-separator of $G[C \cup \Pi]$.
The latter implies $v \notin N(D)$ or $b \notin N(D)$ because $D$ induces a connected subgraph, $D \cap B \cap B' = \emptyset$ by construction, and $B \cap B'$ is a $bv$-separator of $G[C \cup \Pi]$. 
As a result, there exists $z \in \{v,b\}$ such that $N(D) \setminus z \subseteq B \cap B'$.

\begin{figure}[h!]
	\centering
	\includegraphics[width=0.25\textwidth]{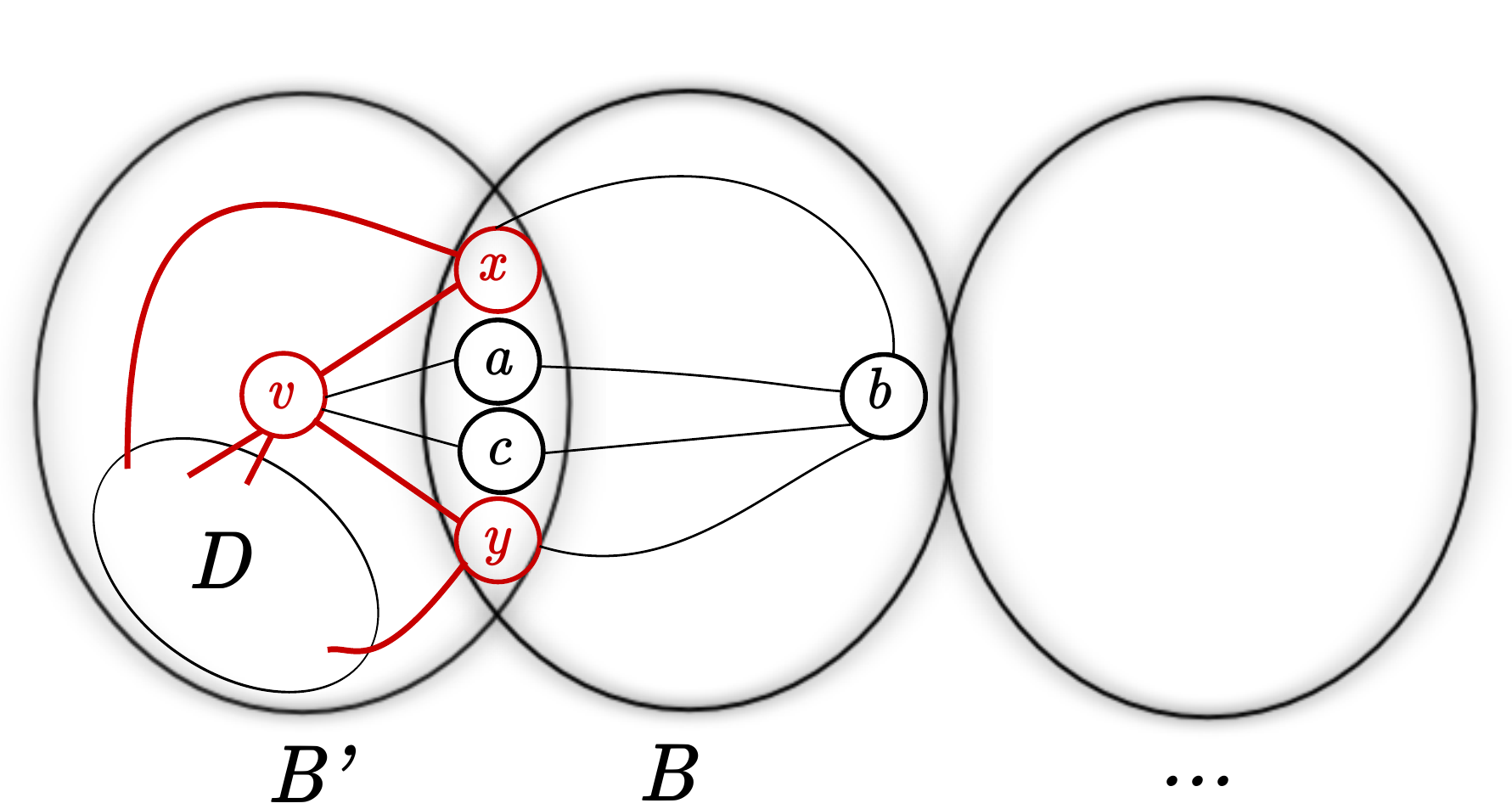}
	\caption{Case $z=v$.}
	\label{fig:fig:leaf-5}
\end{figure}

Moreover let $\{z,z'\} = \{v,b\}$.
$G[C \cup \Pi] \setminus z$ is connected because $G|C\cup\Pi]$ is prime by Claim~\ref{claim:subgraph-prime}, and so, biconnected.
In addition, $N(D) \setminus z$ is a minimal separator of $G[C \cup \Pi] \setminus z$ because it separates $D$ from $z'$ and $N(D) \setminus z \subseteq B \cap B' \subseteq N(z')$ by construction.
By Lemma~\ref{lem:P3exists}, one obtains the existence of two non-adjacent vertices $x,y \in B \cap B'$ such that $N(D) \setminus z = \{x,y\}$, whence $N(D) \subseteq \{x,y,z\}$.
Then, by construction $\Pi' = (x,z,y) \in \mathcal{P}_3(G)$ with $D \subset C$ being a component of $G \setminus \Pi'$, that contradicts Claim~\ref{claim:path-min-pty}.
\end{proofclaim}

By Claim~\ref{claim:subgraph-two-bags}, $C \cup \Pi = S$ (see Figure~\ref{fig:fig:leaf-6}).
Note that it implies that $C \subseteq N[v]$ because $C \setminus v = (B \cap B') \setminus (a,c)$.
In order to conclude that $v$ is a leaf-vertex, we will finally prove in Claim~\ref{claim:intersection-is-a-path} that either $B \cap B' = \{a,c\}$ or $B \cap B'$ induces a path.

\begin{figure}[h!]
	\centering
	\includegraphics[width=0.15\textwidth]{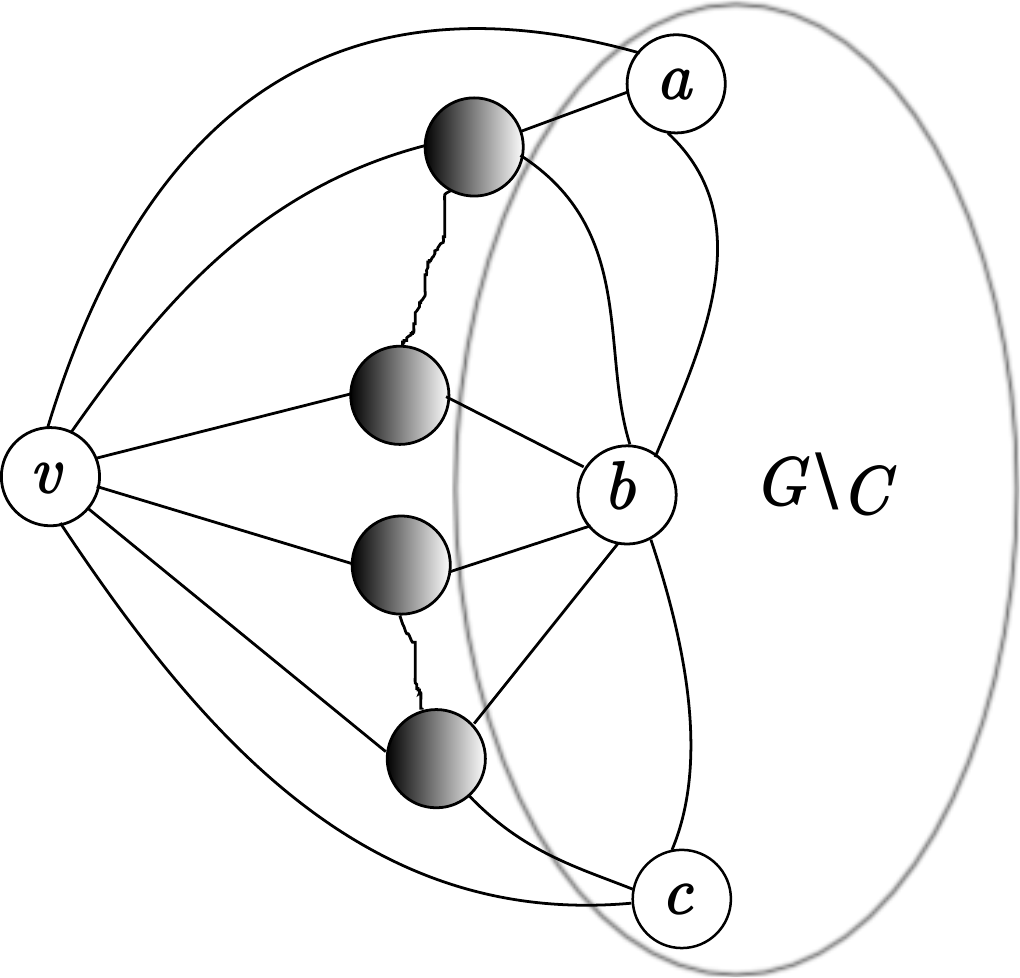}
	\caption{Case $c \in B \cap B'$, $w = b$ and $C \cup \Pi = \{v,b\} \cup (B \cap B')$.}
	\label{fig:fig:leaf-6}
\end{figure} 
 
\begin{claim}
\label{claim:intersection-is-a-path}
If $B \cap B' \neq \{a,c\}$, then $G[B \cap B']$ is a path.
\end{claim} 
 
\begin{proofclaim}
Recall that $b,v \notin B \cap B'$ because $w = b$ by Claim~\ref{claim:w-equals-b} and $b \notin N(v)$ by Claim~\ref{claim:c-is-in-intersection}.	
Hence by the properties of a tree decomposition, $B \cap B'$ is a $bv$-separator of $G[C \cup \Pi]$.	
Since  $v \in C$ by Claim~\ref{claim:v-in-c}, $a \in B \cap B'$ by construction and $c \in B \cap B'$ by Claim~\ref{claim:c-is-in-intersection}, therefore $B \cap B'$ is also a $vb$-separator of $G$.
In particular, $B \cap B'$ is a minimal $bv$-separator of $G$ because $B \cap B' \subseteq N(v) \cap N(w) = N(v) \cap N(b)$ (indeed, recall that $w=b$ by Claim~\ref{claim:w-equals-b}).
By Corollary~\ref{cor:sep-planar}, $B \cap B'$ either induces a cycle or it induces a forest of paths. 

\begin{subclaim}
\label{subclaim:not-a-cycle}
$B \cap B'$ does not induce a cycle.
\end{subclaim}

\begin{proofsubclaim}
By contradiction, let $B \cap B'$ induce a cycle.
Recall that $B \cap B'$ contains the pair of non-adjacent vertices $a,c$ (because $a \in B \cap B'$ by construction and $c \in B \cap B'$ by Claim~\ref{claim:c-is-in-intersection}).
Therefore, one can contract $B \cap B'$ until one obtains an induced quadrangle $(a,x,c,y)$.
Let us contract an arbitrary component of $G \setminus (\Pi \cup C)$ so as to obtain a vertex $z$.
Note that $a,c \in N(z)$ because $G$ is prime by the hypothesis (indeed, neither $a$ nor $b$ nor $c$ nor $(a,b)$ nor $(a,c)$ can be a separator of $G$).
Then, let us contract the edge $\{a,z\}$ to $a$.
By doing so, one obtains a $K_{3,3}$-minor with $\{ a, b, v \}$ being one part of the bipartition and $\{x,y,c\}$ being the other part.
This contradicts the fact that $G$ is planar by the hypothesis, therefore $B \cap B'$ does not induce a cycle.	
\end{proofsubclaim}

It follows from Claim~\ref{subclaim:not-a-cycle} that $B \cap B'$ induces a forest of paths.
Suppose for the sake of contradiction that $B \cap B'$ induces a forest of at least two paths.
Let $x \notin \{a,c\}$ be the endpoint of some path in the forest, that exists because we assume that $B \neq \{a,c\}$.
Observe that $|N(x)| \geq 2$ because $b,v \in N(x)$, and $|N(x)| = |N(x) \cap (C \cup \Pi)| \leq 3$ because $x$ is the endpoint of some path of $B \cap B'$ and $x \in C$.
Furthermore, $N(x) \setminus (b,v) \subseteq B \cap B' \subseteq N(b) \cap N(v)$, and so, if $|N(x)|=3$ then $N(x)$ induces a path. 
Let $\Pi' = N(x)$ if $|N(x)| = 3$, else $\Pi' = (b,a,v)$.
By construction, $\Pi' \subseteq \Pi \cup C$ is a separator of $G$ with $\{x\} \subset C$ being a component of $G \setminus \Pi'$, thus contradicting Claim~\ref{claim:path-min-pty}.
Consequently, $B \cap B'$ induces a path.
\end{proofclaim} 
 
By Claim~\ref{claim:intersection-is-a-path}, either $B\cap B' = \{a,c\}$ or $B \cap B'$ induces a path.
Furthermore, $B \cap B' = N(v)$ because $v \in C$ (Claim~\ref{claim:v-in-c}) and $C \cup \Pi = \{v,b\} \cup (B \cap B')$ (Claim~\ref{claim:subgraph-two-bags}). 
In particular, if $B \cap B' = \{a,c\}$ then $v$ is a leaf-vertex of Type 3.
Else, $B \cap B'$ induces a path and the latter implies that $|B \cap B'| \geq 4$ or else the path $B \cap B'$ would be a separator of $G$ with $\{v\}$ being a component of $G \setminus (B \cap B')$, thus contradicting Claim~\ref{claim:path-min-pty}. 
As a result, since we also have that $B \cap B' \subseteq N(b)$ and $b,v$ are non-adjacent by Claim~\ref{claim:v-in-c}, therefore, $v$ is a leaf-vertex of Type 1.
\end{itemize}

\paragraph{Case $T_a \cap T_c = \emptyset$.}
Since $\Pi$ is a separator of $G$ and $G$ is prime by the hypothesis, one of $\Pi$ or $\Pi \setminus b$ must be a minimal separator of $G$.
Therefore, since $(T,{\cal X})$ is assumed to minimize the distance in $T$ between $T_a$ and $T_c$, by Corollary~\ref{cor:strong-sep} there exist two bags $B_a,B_c$ that are adjacent in $T$ and such that $a \in B_a \setminus B_c$ and $c \in B_c \setminus B_a$.
Furthermore, $a$ dominates $B_a$ while $c$ dominates $B_c$.
Note that $B_a \cap B_c = N(a) \cap N(c)$, so, $b \in B_a \cap B_c$.
In particular, by the properties of a tree decomposition this implies that $S = N(a) \cap N(c)$ is a minimal $ac$-separator of $G$.

We will prove that $C$ is reduced to a vertex (Claim~\ref{claim:component-is-singleton}), the latter being a leaf-vertex.

\begin{claim}
\label{claim:c-is-in-sep}
$C \subseteq S$.
\end{claim}
 
\begin{proofclaim}
Assume for the sake of contradiction that $C \not\subseteq B \cap B'$.
By the properties of a tree decomposition it comes that some vertex $y \in C$ is separated from $a$ or $c$ by the set $S = B \cap B' = N(a) \cap N(c)$.
Say w.l.o.g. that $S$ is an $yc$-separator.
%In particular, $S \cap (C \cup \{b\})$ is an $yc$-separator in the induced subgraph $G[C \cup P]$.
Let $C' \subset C$ be the connected component containing $y$ in $G \setminus (S \cup \{a\})$.
Since we have that $G \setminus a$ is connected because $G$ is prime by the hypothesis (and so, biconnected), that $c \notin C'$ and $N(C') \setminus a \subseteq S \cap (C \cup \Pi) \subseteq N(c) \cap (C \cup \Pi)$, then it comes that $N(C') \setminus a$ is a minimal $yc$-separator of $G \setminus a$.
So, by Lemma~\ref{lem:P3exists} there exist $x',y' \in S$ such that $N(C') \setminus a = \{x',y'\}$.
Therefore, $\Pi' = (x',a,y') \in {\cal P}_3(G)$ and $C' \subset C$ is a component of $G \setminus \Pi'$, that contradicts Claim~\ref{claim:path-min-pty}.
\end{proofclaim} 

By Claim~\ref{claim:c-is-in-sep}, $C \subseteq S$ (see Figure~\ref{fig:fig:leaf-7} for an illustration).
Since $S$ is an $ac$-separator and for any component $C'$ of $G \setminus (\Pi \cup C)$, $a,c \in N(C')$ because $G$ is prime, therefore $S \cap C' \neq \emptyset$.
One thus obtains the following chain of strict subset containment relations $C \subset C \cup \{b\} \subset S$.
Furthermore, by Corollary~\ref{cor:sep-planar}, $S$ either induces a cycle or a forest of paths, so, $C$ being a strict connected subset of $S$, it must induce a path.
In particular, $C \cup \{b\}$ also being a strict subset of $S$, either it induces a path or it is the union of the path induced by $C$ with the isolated vertex $b$. 

\begin{figure}[h!]
	\centering
	\includegraphics[width=0.35\textwidth]{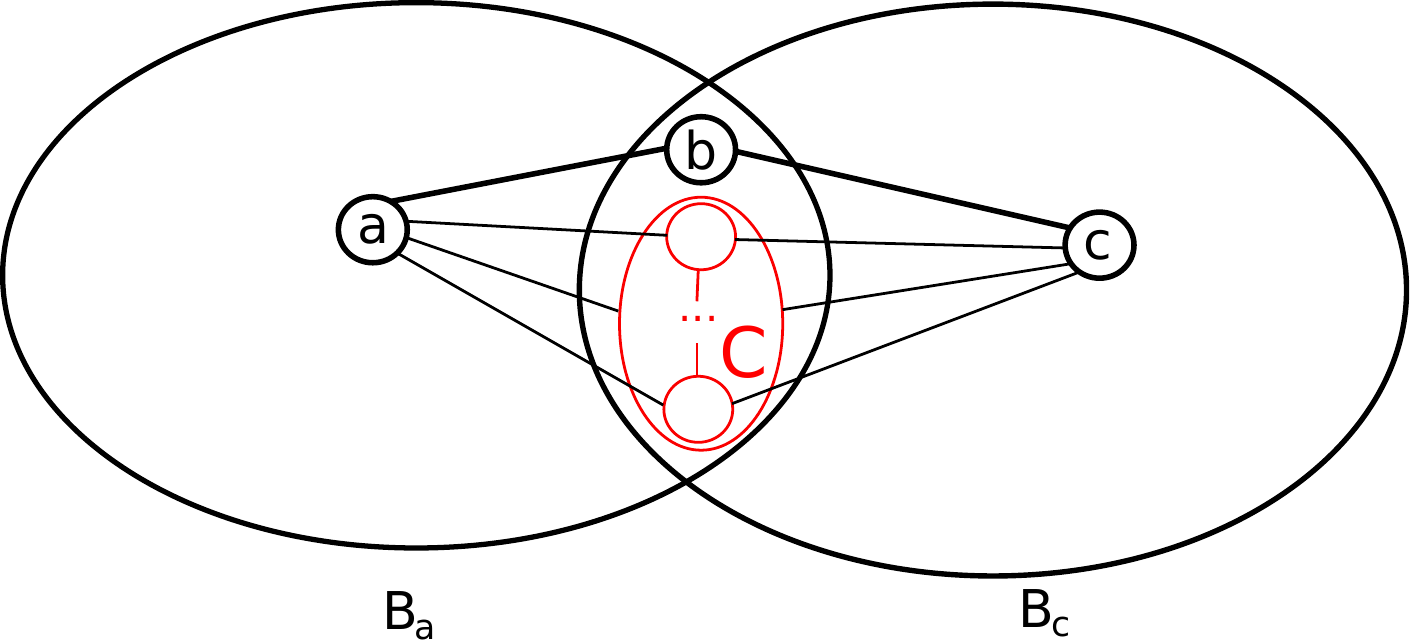}
	\caption{Case $T_a \cap T_c = \emptyset$.}
	\label{fig:fig:leaf-7}
\end{figure} 

\begin{claim}
	\label{claim:component-is-singleton}
	$|C|=1$.
\end{claim}

\begin{proofclaim}
	Assume for the sake of contradiction that $|C| \geq 2$.
	Since $C$ induces a path, let us pick an endpoint $v \in C$ that is not adjacent to vertex $b$ (recall that $C \cup \{b\}$ being a strict subset of $S$, it does not induce a cycle).
	In such a case, $N(v)$ induces a path $\Pi' \in {\cal P}_3(G)$, with $a,c \in \Pi'$ and $\{v\} \subset C$ is a component of $G \setminus \Pi'$, thus contradicting Claim~\ref{claim:path-min-pty}.
\end{proofclaim}

By Claim~\ref{claim:component-is-singleton}, $C$ is reduced to a vertex $v$, that is either a leaf-vertex of Type 2 (if $v \in N(b)$) or of Type 3 (if $v \notin N(b)$).
\end{proof}

Note that in some cases, there may only exist leaf-vertices of only one Type ({\it i.e.}, see respectively Figure~\ref{fig:type-1-only},~\ref{fig:type-2-only} and~\ref{fig:type-3-only} for Types 1,2 and 3).
Therefore, there is none of the three Types of leaf-vertices that can be avoided in our algorithm.

\begin{figure}[h]
	\begin{minipage}[h]{.30\linewidth}
		\centering
		\includegraphics[width=0.75\textwidth]{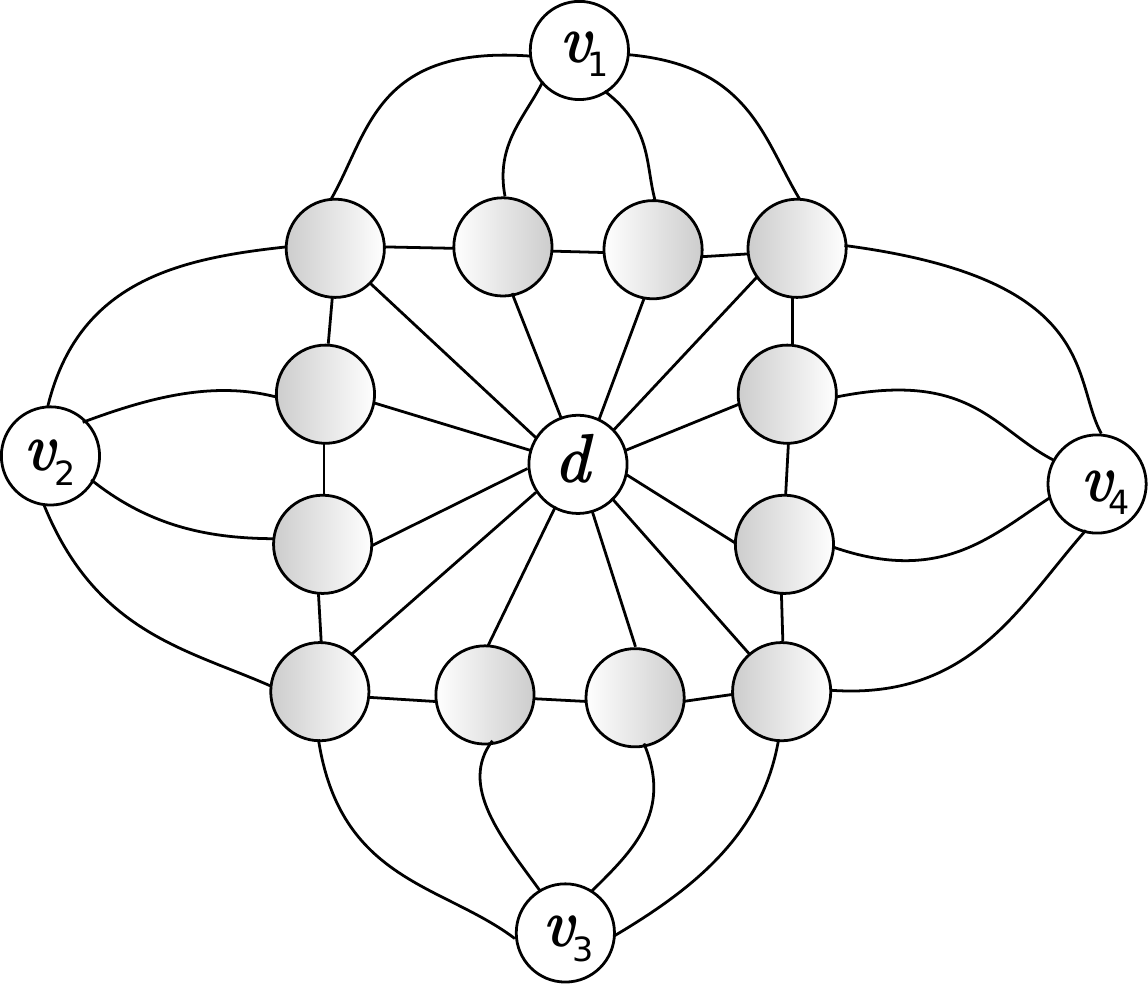}
		\caption{A planar graph $G$ with $tb(G) = 1$ and all of its four leaf-vertices $v_1,v_2,v_3,v_4$ of Type 1.}
		\label{fig:type-1-only}
	\end{minipage} \hfill
	\begin{minipage}[c]{.30\linewidth}
		\centering
		\includegraphics[width=0.75\textwidth]{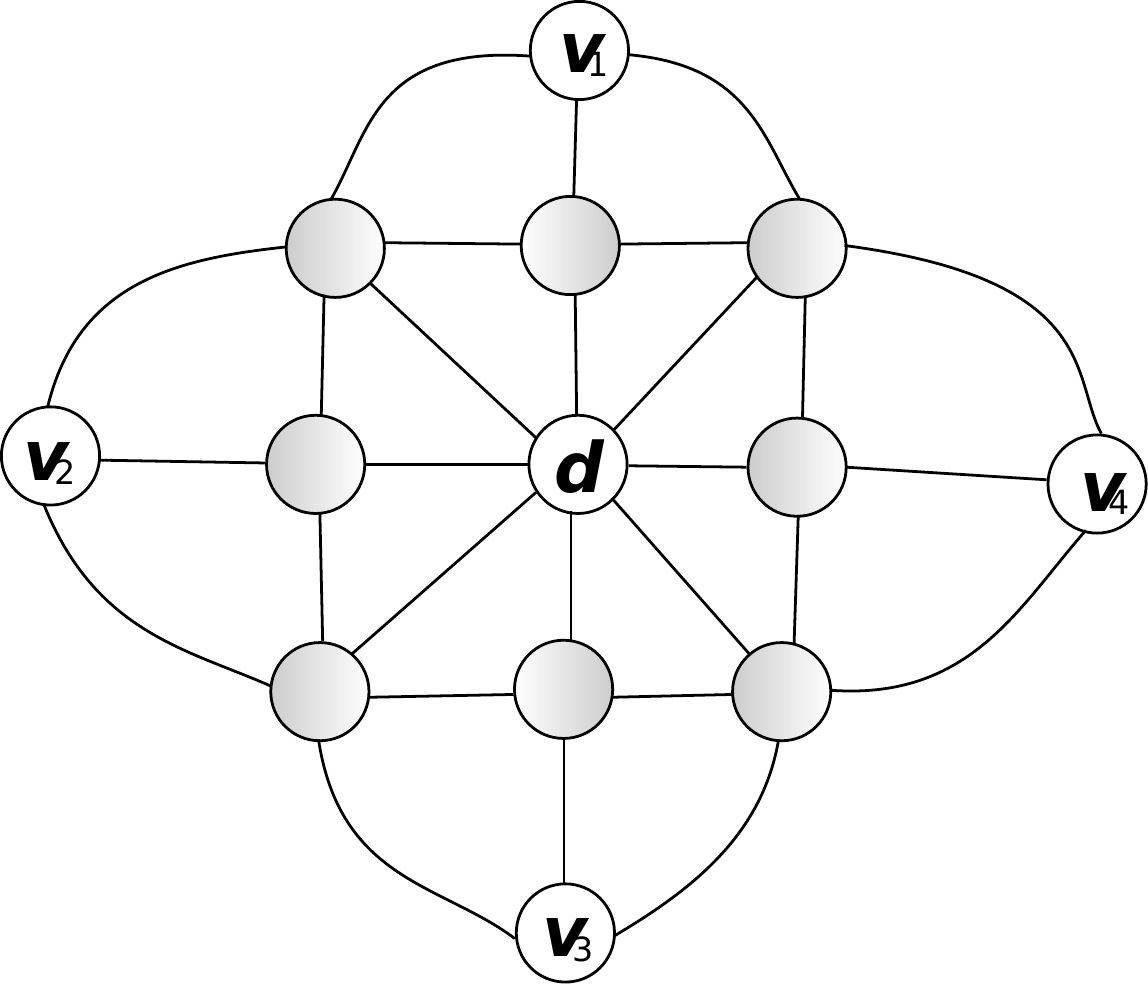}
		\caption{A planar graph $G$ with $tb(G) = 1$ and all of its four leaf-vertices $v_1,v_2,v_3,v_4$ of Type 2.}
		\label{fig:type-2-only}
	\end{minipage} \hfill
	\begin{minipage}[h]{.30\linewidth}
		\centering
		\includegraphics[width=0.75\textwidth]{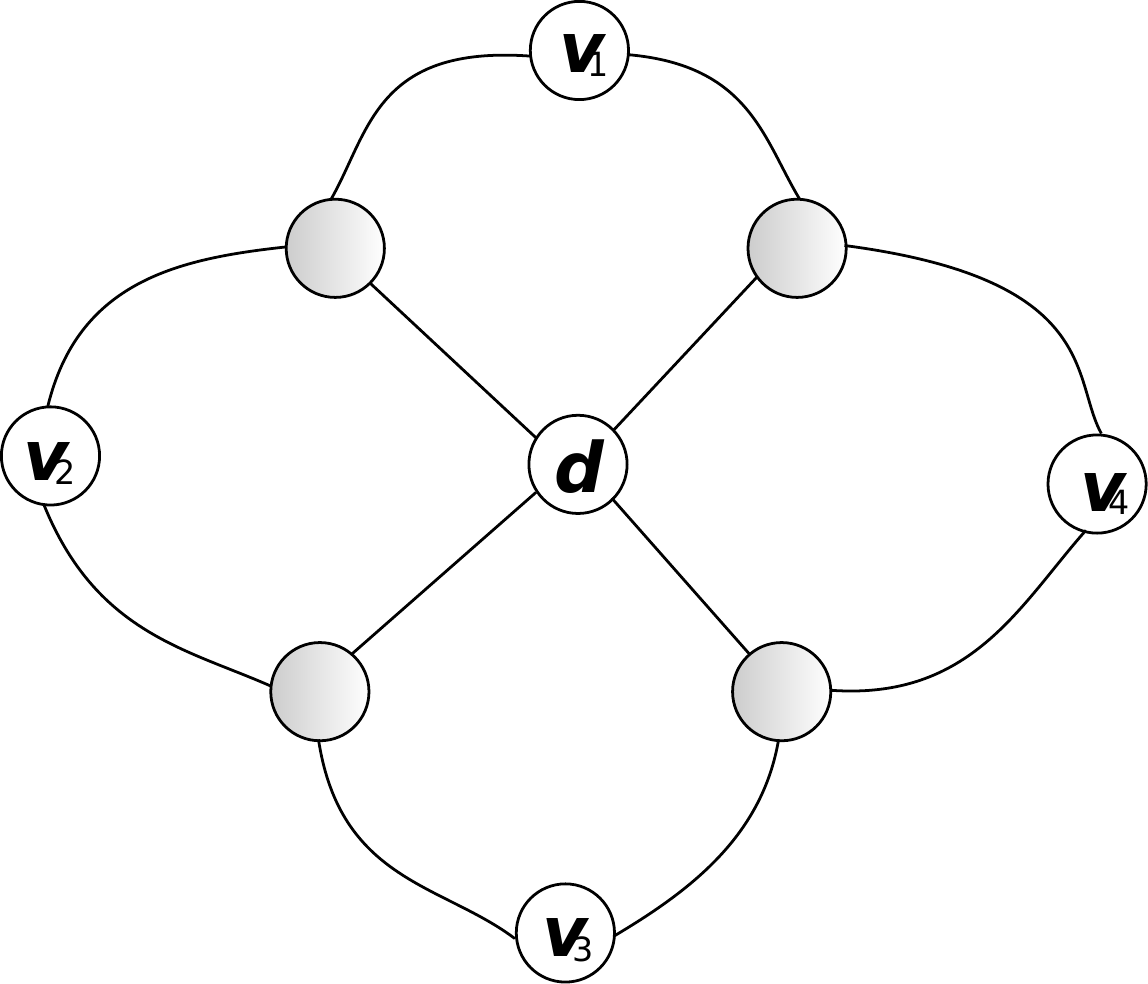}
		\caption{A planar graph $G$ with $tb(G) = 1$ and all of its four leaf-vertices $v_1,v_2,v_3,v_4$ of Type 3.}
		\label{fig:type-3-only}	
	\end{minipage}
\end{figure}

Examples of planar graphs $G$ with $tb(G) = 1$ and ${\cal P}_3(G) = \emptyset$ include $C_4$, the cycle with four vertices.
To prove correctness of~\ref{step:find-leaf-vertex}, it now suffices to prove that all these graphs (with ${\cal P}_3(G) = \emptyset$) admit a star-decomposition with at most two bags.

\begin{lemma}\label{lem:2bags}
For any prime planar graph $G$, if $tb(G)=1$ and ${\cal P}_3(G) = \emptyset$, then $G$ admits a star-decomposition with at most $2$ bags.
\end{lemma}

\begin{proof}
By contradiction, let $(T,{\cal X})$ be a star-decomposition of $G$ with at least three bags.
Let $t \in V(T)$ be an internal node, by the properties of a tree decomposition the bag $X_t$ is a separator of $G$.
Let $u \in X_t$ satisfy $X_t \subseteq N_G[u]$.
Since $G$ is biconnected, therefore $X_t \setminus u$ is a separator of $G \setminus u$.
By Lemma~\ref{lem:P3exists}, there exist $x,y \in X_t \setminus u$ non-adjacent such that $\{x,y\}$ is a minimal separator of $G \setminus u$.
In such case, $(x,u,y) \in {\cal P}_3(G)$, which contradicts the fact that ${\cal P}_3(G) = \emptyset$. 
\end{proof}

\subsubsection{Case of leaf-vertex $v$ of Type $1$}

\begin{lemma}\label{lem:ends-separate}
Let $G$ be a prime planar graph and $v$ be a leaf-vertex of Type $1$. Let $\Pi_v$ be the path induced by $N(v)$ and let $a_v,c_v$ be the ends of $\Pi_v$. 
Suppose $V(G) \neq N[v] \cup \{d_v\}$.

Then $\Pi'=(a_v,d_v,c_v) \in {\cal P}_3(G)$ and $N[v] \setminus \{a_v,c_v\}$ is a component of $G\setminus \Pi'$.
\end{lemma}

\begin{proof}
Let $C$ be a component of $G \setminus (N[v] \cup \{d_v\})$, that exists by the hypothesis.
By construction, $v \notin N[C]$, so, $N(C) \subseteq N(v) \cup \{d_v\}$ separates $v$ from $C$.
Furthermore, since $G$ is prime by the hypothesis, there exist $x,y \in N(C)$ non-adjacent.
Note that $d_v \notin \{x,y\}$ because $N(v) \subseteq N(d_v)$ by the hypothesis, hence $x,y \in N(v)$.

We claim that $\{x,y\} = \{a_v,c_v\}$.
By contradiction, suppose $x \notin \{a_v,c_v\}$.
Let us write $\Pi_v = (P,x,Q,y,R)$ with $P,Q$ non-empty subpaths of $\Pi_v$ and $R$ a (possibly empty) subpath of $\Pi_v$.
In such a case, the connected subsets $S_1 := \{v\} \cup P, \ S_2 := \{d_v\}, \ S_3 := \{x\}, \ S_4 := Q$ and $S_5 := \{y\} \cup C$ induce a $K_5$-minor of $G$, that contradicts the hypothesis that $G$ is planar. 
Therefore, the claim is proved, that is, $\{x,y\} = \{a_v,c_v\}$.

To prove the lemma, it now suffices to prove that $N(C) \cap N(v) = \{a_v,c_v\}$ for in such a case the result will hold for any component $C'$ of $G \setminus (N[v] \cup \{d_v\})$.
By contradiction, let $x' \in (N(C) \cap N(v)) \setminus (a_v,c_v)$.
Since $|N(v)| \geq 4$ because $v$ is a leaf-vertex of Type 1 by the hypothesis, therefore, $x'$ and $a_v$ are non-adjacent or $x'$ and $c_v$ are non-adjacent.
Let $y' \in \{a_v,c_v\}$ be non-adjacent to $x'$.
Since $x',y' \in N(C) \cap N(v)$ are non-adjacent, therefore, by the same proof as for the above claim $\{x',y'\} = \{a_v,c_v\}$, that would contradict the assumption that $x' \notin \{a_v,c_v\}$.
As a result, $N(C) \subseteq (a_v,d_v,c_v)$ and so, since the result holds for any component $C'$ of $G \setminus (N[v] \cup \{d_v\})$, $\Pi' = (a_v,d_v,c_v) \in {\cal P}_3(G)$ with $N[v] \setminus \{a_v,c_v\}$ being a full component of $G\setminus \Pi'$.
\end{proof}

\begin{theorem}\label{th:Type1}
Let $G$ be a prime planar graph and $v$ be a leaf-vertex of Type $1$. Let $\Pi_v$ be the path induced by $N(v)$ and let $a_v,c_v$ be the ends of $\Pi_v$.
Suppose $V(G) \neq N[v] \cup \{d_v\}$.
 
Then, the graph $G'$, obtained from $G \setminus v$ by contracting the internal vertices of $\Pi_v$ to a single edge, is prime and planar, and $tb(G)=1$ if and only if $tb(G')=1$.
\end{theorem}

\begin{figure}[h!]
 \centering
 \includegraphics[width=0.5\textwidth]{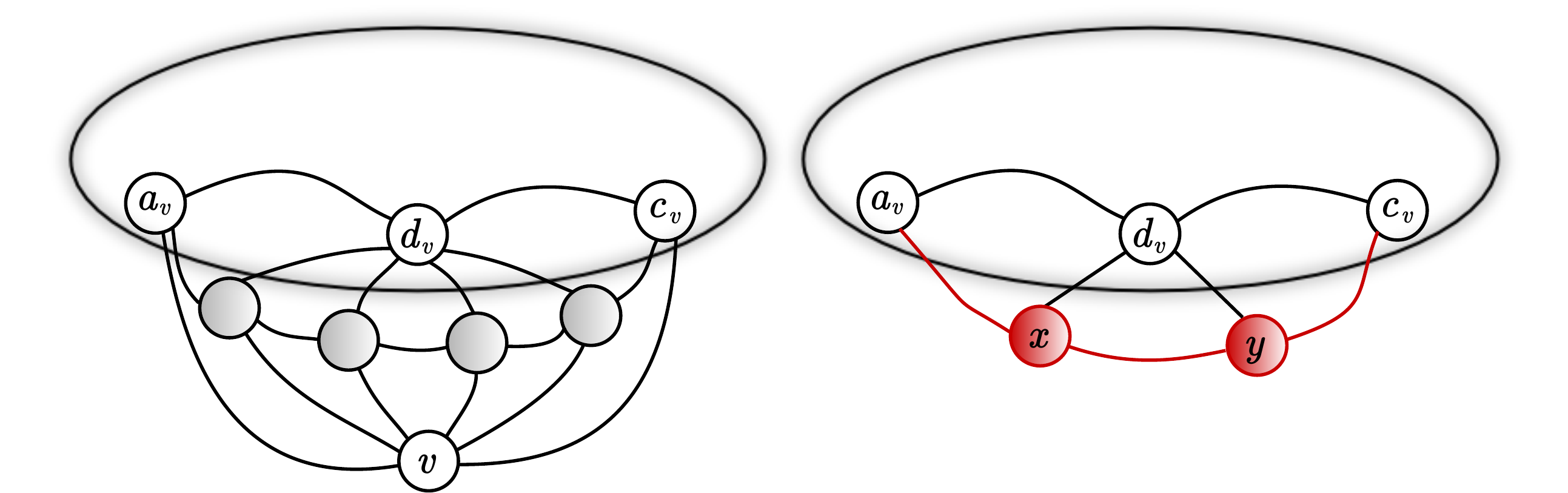}
 \caption{Contraction of the internal vertices of $\Pi_v$ to a single edge and removal of $v$.}
 \label{fig:planar-type1}
\end{figure}

\begin{proof}
For the remaining of the proof, let $\Pi'_v = (a_v,x,y,c_v)$ be the path resulting from the contraction of the internal vertices of $\Pi_v$ to the edge $\{x,y\}$ in $G'$.
By Lemma~\ref{lem:ends-separate} $(a_v,d_v,c_v) \in {\cal P}_3(G)$ with $(N[v] \setminus (a_v,c_v))$ being a full component of $G \setminus (a_v,d_v,c_v)$.
Consequently, $N_{G'}(x) = \{a_v,d_v,y\}$ and $N_{G'}(y) = \{c_v,d_v,x\}$.

%Let us prove that $G'$ is a prime planar graph.
The graph $G'$ is a minor of $G$, that is a planar graph by the hypothesis, so, $G'$ is also planar.
In order to prove that $G'$ is prime, by contradiction, let $S$ be a minimal clique-separator of $G'$.
There are two cases to be considered.
\begin{itemize}
\item Suppose $x \in S$ or $y \in S$.
In such case, $S \subseteq (a_v,x,d_v)$, or $S \subseteq (x,d_v,y)$, or $S \subseteq (y,d_v,c_v)$.
By Lemma~\ref{lem:ends-separate} $(a_v,d_v,c_v) \in {\cal P}_3(G)$ with$(N[v] \setminus (a_v,c_v))$ being a full component of $G \setminus (a_v,d_v,c_v)$, and so, for every component $C$ of $G' \setminus (\Pi'_v \cup \{d_v\}) = G \setminus (N[v] \cup \{d_v\})$ $a_v,c_v \in N(C)$ because $G$ is prime by the hypothesis.
In such case, since $a_v \notin S$ or $c_v \notin S$, therefore, $G' \setminus S$ is connected, that contradicts the assumption that $S$ is a clique-separator of $G'$.
\item Else, $x,y \notin S$.
Since $a_v \notin S$ or $c_v \notin S$ because $a_v$ and $c_v$ are non-adjacent in $G'$, therefore, $S$ must be a separator of $G' \setminus (x,y)$ or else $G' \setminus S$ would be connected because $\Pi'_v$ induces a path of $G'$ (thus contradicting the assumption that $S$ is a separator of $G'$).
In such a case, since by Lemma~\ref{lem:ends-separate} $(a_v,d_v,c_v) \in {\cal P}_3(G)$ with$(N[v] \setminus (a_v,c_v))$ being a full component of $G \setminus (a_v,d_v,c_v)$, since $S$ is a separator of $G' \setminus (x,y) = G  \setminus (N[v] \setminus (a_v,c_v))$ and since $S$ is not a separator of $G$ because $G$ is prime by the hypothesis, therefore, by Lemma~\ref{lem:separatorInCompo} there are exactly two components $C_a,C_c$ in $G' \setminus (S \cup \{x,y\})$ with $a_v \in C_a$ and $c_v \in C_c$.
However, $\Pi_v'$ is an $a_vc_v$-path of $G' \setminus S$, thus contradicting the assumption that $S$ is a separator of $G'$. 
\end{itemize}
As a result, $G'$ is a prime planar graph.

\medskip
\noindent
Finally, let us prove $tb(G) = 1$ if and only if $tb(G') = 1$.

\begin{itemize}

\item If $tb(G) = 1$ then $tb(G \setminus v) = 1$ because $N(v) \subseteq N(d_v)$ by the hypothesis, and so, $tb(G') = 1$ because $G'$ is obtained from $G \setminus v$ by edge-contractions and tree-breadth is contraction-closed (Lemma~\ref{lem:contraction-closed}).

\item Conversely, let us prove that $tb(G) = 1$ if $tb(G') = 1$.
To prove this, let $(T',{\cal X}')$ be a reduced star-decomposition of $G'$, that exists by Lemma~\ref{lem:dominator-in-bag}, minimizing the distance in $T'$ between the two subtrees $T'_{a_v}$ and $T'_{c_v}$.
In order to prove $tb(G) = 1$, it suffices to show how to construct a star-decomposition of $G$ from $(T',{\cal X}')$.

We will prove as an intermediate claim that $T'_{a_v} \cap T'_{c_v} \neq \emptyset$.
By contradiction, suppose $T'_{a_v} \cap T'_{c_v} = \emptyset$.
Since by Lemma~\ref{lem:ends-separate} $(a_v,d_v,c_v) \in {\cal P}_3(G)$ with $(N[v] \setminus (a_v,c_v))$ being a full component of $G \setminus (a_v,d_v,c_v)$, therefore, $(a_v,d_v,c_v) \in {\cal P}_3(G')$ with $\{x,y\}$ being a full component of $G' \setminus (a_v,d_v,c_v)$.
Since we proved that $G'$ is prime, it follows that one of $(a_v,c_v)$ or $(a_v,d_v,c_v)$ is a minimal separator of $G'$.
In such a situation, since $(T',{\cal X}')$ is assumed to minimize the distance in $T'$ between $T'_{a_v}$ and $T'_{c_v}$, therefore, by Corollary~\ref{cor:strong-sep} there are two adjacent bags $B'_{a_v}, \ B'_{c_v}$ such that $a_v \in B'_{a_v} \setminus B'_{c_v}$ and $c_v \in B'_{c_v} \setminus B'_{a_v}$ respectively dominate $B'_{a_v}$ and $B'_{c_v}$ in $G'$.
However by the properties of a tree decomposition this implies that $B'_{a_v} \cap B'_{c_v} = N(a_v) \cap N(c_v)$ is an $a_vc_v$-separator of $G'$, thus contradicting the existence of the $a_vc_v$-path $\Pi'_v$.
Therefore, the claim is proved and $T'_{a_v} \cap T'_{c_v} \neq \emptyset$. 

Recall that $T'_{a_v} \cap T'_{d_v} \neq \emptyset$ and similarly $T'_{c_v} \cap T'_{d_v} \neq \emptyset$ by the properties of a tree decomposition.
Hence, the subtrees $T'_{a_v},T'_{c_v}, T'_{d_v}$ are pairwise intersecting, and so, by the Helly property (Lemma~\ref{lem:helly}), $T'_{a_v} \cap T'_{d_v} \cap T'_{c_v} \neq \emptyset$.
Let us now proceed as follows so as to obtain a star-decomposition of $G$.
Let us remove $x,y$ from all bags in ${\cal X}'$, that keeps the property for $(T',{\cal X}')$ to be a star-decomposition because $x$ and $y$ are dominated by $d_v$ in $G'$.
Then, let us add two new bags $B_1 = N[v], \ B_2 = N(v) \cup \{d_v\}$, and finally let us make $B_1,B_2$ pairwise adjacent and let us make $B_2$ adjacent to some bag of $T'_{a_v} \cap T'_{d_v} \cap T'_{c_v}$.
By construction, the resulting tree decomposition is indeed a star-decomposition of $G$, whence $tb(G) = 1$. 
\end{itemize}\end{proof}

\subsubsection{Proof of~\ref{step:prime-case}~\ref{step:many-neighbours}}
\label{sec-step4-a}

In the following three subsections (~\ref{sec-step4-a},~\ref{sec:step-4-b-i} and~\ref{sec:step-4-b-ii}) we will prove correctness of the algorithm for the case of a leaf-vertex $v$ of Type $2$ or $3$ and $G \setminus v$ is prime (~\ref{step:prime-case}).
Our proofs in these subsections will mostly rely on Lemma~\ref{lem:strong-pair}.

Let us first show how we can use Lemma~\ref{lem:strong-pair} in order to prove correctness of~\ref{step:prime-case}~\ref{step:many-neighbours}.
Note that since we are in the case when $G \setminus v$ is prime, we needn't prove it in the following Theorem~\ref{th:primeEasy}.

\begin{theorem}\label{th:primeEasy}
Let $G = (V,E)$ be a prime planar graph, let $v$ be a leaf-vertex of Type 2 or 3, and let $\Pi_v = (a_v,b_v,c_v)$ be as in Definition~\ref{def:leafVertex}.

Suppose that $|N(a_v) \cap N(c_v)| \geq 3$ in $G \setminus v$, or there exists a minimal separator $S\subseteq (N(a_v) \cap N(c_v)) \cup \{a_v,c_v\}$ in $G \setminus v$ and $\{a_v,c_v\} \subseteq S$.

Then, $tb(G) = 1$ if and only if $tb(G \setminus v) = 1$. 
\end{theorem}

\begin{proof}
First we prove that $tb(G) = 1$ implies that $tb(G \setminus v) = 1$, that is the easy part of the result.
Let $(T,{\cal X})$ be a tree decomposition of $G$ of breadth one, let $(T,{\cal X}')$ be such that for every node $t \in V(T)$, $X_t' = X_t \setminus v$.
Observe that $(T,{\cal X}')$ is a tree decomposition of $G \setminus v$.
Furthermore, we claim that it has breadth one, indeed, for every $t \in V(T)$ such that $X_t \subseteq N_G[v]$, $X_t' \subseteq N_G[b_v]$ because $N_G(v) \subseteq N_G[b_v]$.
As a result, $tb(G \setminus v ) = 1$. 

\smallskip
Conversely, we prove that $tb(G \setminus v) = 1$ implies that  $tb(G) = 1$.
%We first claim that for every $a_vc_v$-separator $S'$ of $G \setminus v$, if there is $z \notin \{a_v,b_v,v\}$ such that $S' \subseteq N[z]$ then $z \in N(a_v) \cap N(c_v) \setminus v$.
%To prove the claim, there are two cases depending on whether $|N(a_v) \cap N(c_v)| \geq 3$ in $G \setminus v$, or there exists a minimal separator $S\subseteq (N(a_v) \cap N(c_v)) \cup \{a_v,c_v\}$ in $G \setminus v$ and $\{a_v,c_v\} \subseteq S$.
%\begin{itemize}
%\item First, if $a_v$ and $c_v$ have at least three common neighbours, vertex $z$ dominates $N(a_v) \cap N(c_v) \subseteq S'$, and all the vertices dominating $N(a_v) \cap N(c_v)$ are amongst $(N(a_v) \cap N(c_v)) \cup \{a_v,c_v\}$ because $G$ is $K_{3,3}$-minor-free by the hypothesis.
%\item Else, there exists a minimal separator $S\subseteq (N(a_v) \cap N(c_v)) \cup \{a_v,c_v\}$ in $G \setminus v$ and $\{a_v,c_v\} \subseteq S$.
%Then for any full component $C$ of $G \setminus S$, the $a_vc_v$-separator $S'$ must contain a vertex $x \in C$.
%Since there are at least two full components of $G \setminus S$, then $z \in S \setminus (a_v,c_v) \subseteq N(a_v) \cap N(c_v)$.
%\end{itemize}
%Therefore, the claim is proved.
Let $(T',{\cal X}')$ be a star-decomposition of $G \setminus v$ minimizing the distance in $T'$ between the subtrees $T'_{a_v}$ and $T'_{c_v}$.
There are two cases.
If $T'_{a_v} \cap T'_{c_v} \neq \emptyset$, then the subtrees $T'_{a_v},T'_{b_v},T'_{c_v}$ are pairwise intersecting, hence by the Helly property (Lemma~\ref{lem:helly}) $T'_{a_v} \cap T'_{b_v} \cap T'_{c_v} \neq \emptyset$, and so it suffices to make adjacent to any bag of $T'_{a_v} \cap T'_{b_v} \cap T'_{c_v}$ the new bag $N_G[v] \subseteq \{a_v,b_v,c_v,v\}$ so as to obtain a star-decomposition of $G$.
Else $T'_{a_v} \cap T'_{c_v} = \emptyset$ and so, by Corollary~\ref{cor:strong-three} if $|N(a_v) \cap N(c_v)| \geq 3$ in $G \setminus v$ or by Corollary~\ref{cor:strong-sep} else, there are two adjacent bags $B'_{a_v},B'_{c_v}$ such that $a_v \in B'_{a_v} \setminus B'_{c_v}, b_v  \in B'_{a_v} \cap B'_{c_v} \subseteq N(a_v) \cap N(c_v)$ and $c_v \in B'_{c_v} \setminus B'_{a_v}$.
Furthermore, $a_v$ dominates $B'_{a_v}$ while $c_v$ dominates $B'_{c_v}$.
One obtains a star-decomposition of $G$ simply by adding vertex $v$ into bags $B'_{a_v}$ and $B'_{c_v}$.
\end{proof}

\subsubsection{Proof of~\ref{step:prime-case}~\ref{step:few-neighbours}~\ref{step:one-neighbour}}
\label{sec:step-4-b-i}

The proof of this step is more involved than the proof of previous~\ref{step:prime-case}~\ref{step:many-neighbours}.
We will need the following intermediate lemma.

\begin{lemma}\label{lem:disconnection}
Let $G=(V,E)$ be a prime graph with $tb(G) = 1$, let $v$ be a leaf-vertex of Type 2 or 3 and let $\Pi_v=(a_v,b_v,c_v)$ be as in Definition~\ref{def:leafVertex}.
Suppose that $N(a_v) \cap N(c_v) = \{v,b_v\}$ and $V \neq \Pi_v \cup \{v\}$.
Then, $N_G[b_v] \setminus (a_v,c_v,v)$ is an $a_vc_v$-separator of $G \setminus v$.
\end{lemma}

\begin{proof}
Let $(T,{\cal X})$ be a star-decomposition of $G$, that exists by Lemma~\ref{lem:dominator-in-bag}, minimizing the distance in $T$ between the subtrees $T_{a_v}$ and $T_{c_v}$.
We claim that $T_{a_v} \cap T_{c_v} \neq \emptyset$, {\it i.e.}, $a_v,c_v$ are in a same bag of the decomposition.
By contradiction, let $T_{a_v} \cap T_{c_v} = \emptyset$.
Since $G$ is prime and $\Pi_v$ is a separator of $G$, therefore, one of $\Pi_v$ or $\Pi_v \setminus b_v$ is a minimal separator of $G$.
Since $(T,{\cal X})$ minimizes the distance in $T$ between $T_{a_v}$ and $T_{c_v}$, therefore, by Corollary~\ref{cor:strong-sep} there exist two adjacent bags $B_{a_v},B_{c_v}$ such that $a_v \in B_{a_v} \setminus B_{c_v}$ and $c_v \in B_{c_v} \setminus B_{a_v}$.
Furthermore, vertices $a_v$ and $c_v$ respectively dominate the bags $B_{a_v}$ and $B_{c_v}$.
This implies $B_{a_v} \cap B_{c_v} = N_G(a_v) \cap N_G(c_v)$ and so, $N_G(a_v) \cap N_G(c_v)$ is a minimal $a_vc_v$-separator of $G$ by the properties of the tree decomposition.
However, let $C$ be any component of $G \setminus (\Pi_v \cup \{v\})$, that exists because $V \neq \Pi_v \cup \{v\}$ by the hypothesis.
Since $G$ is prime, therefore, $a_v,c_v \in N(C)$ (or else, one of the cliques $a_v$ or $b_v$ or $c_v$ or $(a_v,b_v)$ or $(b_v,c_v)$ would be a clique-separator of $G$, thus contradicting the assumption that $G$ is prime).
Then, the $a_vc_v$-separator $N_G(a_v) \cap N_G(c_v)$ must contain some vertex of $C$, which contradicts the fact that $N_G(a_v) \cap N_G(c_v) = \{v,b_v\}$ by the hypothesis. 
As a result, we proved that $T_{a_v} \cap T_{c_v} \neq \emptyset$.

Let $H$ be the chordal supergraph of $G$ such that $(T,{\cal X})$ is a clique-tree of $H$.
Equivalently, every two vertices $x,y \in V$ are adjacent in $H$ if and only if they are in a same bag of ${\cal X}$.
In particular, $a_v,c_v$ are adjacent in $H$.
Let $S := N_H(a_v) \cap N_H(c_v)$.
We claim that $S$ is an $a_vc_v$-separator of $G$.
By contradiction, if it is not an $a_vc_v$-separator of $G$, then there exists an $a_vc_v$-path $P_{a_vc_v}$ of $G$ which does not intersect $S$.
Furthermore, $P_{a_vc_v}$ is a path of $H$ because $H$ is a supergraph of $G$, and it has length at least two because $a_v,c_v$ are non-adjacent in $G$.
So, let $Q_{a_vc_v}$ be taken of minimum length amongst all $a_vc_v$-paths of length at least two in $H$ that do not intersect $S$ (the existence of such a path follows from the existence of $P_{a_vc_v}$).
Observe that $Q_{a_vc_v}$ may be not a path in $G$.
By minimality of $Q_{a_vc_v}$, the vertices of $Q_{a_vc_v}$ induce a cycle of $H$ because $a_v,c_v$ are adjacent in $H$.
Therefore, the vertices of $Q_{a_vc_v}$ induce a triangle because $H$ is chordal.
However, this contradicts the fact that $Q_{a_vc_v}$ does not intersect $S = N_H(a_v) \cap N_H(c_v)$, so, the claim is proved.

Finally, let us prove that $S \setminus v \subseteq N_G[b_v] \setminus (a_v,c_v,v)$, that will conclude the proof that $N_G[b_v] \setminus (a_v,c_v,v)$ is an $a_vc_v$-separator of $G \setminus v$.
For every vertex $x \in S \setminus v$, $x \in N_H(a_v) \cap N_H(c_v)$, therefore, $T_{a_v} \cap T_x \neq \emptyset$ and $T_{c_v} \cap T_x \neq \emptyset$ by construction of $H$.
Since the subtrees $T_{a_v},T_{c_v},T_x$ are pairwise intersecting, by the Helly property (Lemma~\ref{lem:helly}) $T_{a_v} \cap T_{c_v} \cap T_x \neq \emptyset$, or equivalently there is some bag $B \in T_{a_v} \cap T_{c_v} \cap T_x$.
Let $z \in B$ dominate the bag.
Clearly, $x \in N_G[z]$.
Furthermore, $z \in N_G(a_v) \cap N_G(c_v)$ because $a_v,c_v$ are non-adjacent in $G$.
As a result, either $z = b_v$ or $z=v$.
Since $x \neq v$ and $N_G(v) \subseteq N_G[b_v]$, we have that $x \in N_G[b_v]$ in both cases.
\end{proof}

\begin{theorem}\label{thm:force-clique}
Let $G = (V,E)$ be a prime planar graph with $tb(G) = 1$, $v$ be a leaf-vertex of Type 2 or 3, and let $\Pi_v = (a_v,b_v,c_v)$ be as in Definition~\ref{def:leafVertex}.
Suppose $N(a_v) \cap N(c_v) = \{b_v,v\}$, and $G \setminus v$ is prime, and there is no minimal separator $S\subseteq (N(a_v) \cap N(c_v)) \cup \{a_v,c_v\}$ in $G \setminus v$ such that $\{a_v,c_v\} \subseteq S$.
Then, $G = C_4$, a cycle with four vertices.
\end{theorem}

\begin{proof}
By contradiction, assume $G \neq C_4$.
Since $G$ is prime by the hypothesis, $G$ has at least five vertices (the single other graph with four vertices and a leaf-vertex of Type 2 or 3 is the diamond, which is not prime).
Equivalently, $V \neq \Pi_v \cup \{v\}$.
By Lemma~\ref{lem:disconnection}, this implies that $N[b_v] \setminus (a_v,c_v,v)$ is an $a_vc_v$-separator of $G \setminus v$.
Since $G \setminus v$ is prime by the hypothesis, and so, biconnected, therefore, $G \setminus (b_v,v)$ is connected, and so, $N(b_v) \setminus (a_v,c_v,v) \neq \emptyset$ is an $a_vc_v$-separator of $G \setminus (b_v,v)$.
In particular, $a_v,b_v,c_v \in N(V \setminus (\Pi_v \cup \{v\}))$. 

Moreover, we claim that $V \setminus (\Pi_v \cup \{v\})$ induces a connected subgraph (note that the latter implies that $V \setminus (\Pi_v \cup \{v\})$ is a full component of $G \setminus \Pi_v$).
By contradiction, let $C_1,C_2$ be distinct components of $V \setminus (\Pi_v \cup \{v\})$.
Since $G$ is prime, $a_v,c_v \in N_G(C_1) \cap N_G(C_2)$ (or else, one of the cliques $a_v$ or $b_v$ or $c_v$ or $(a_v,b_v)$ or $(b_v,c_v)$ would be a clique-separator of $G$, thus contradicting the assumption that $G$ is prime).
Therefore, $b_v \in N_G(C_1) \cap N_G(C_2)$ because $N[b_v] \setminus (a_v,c_v,v)$ is an $a_vc_v$-separator of $G \setminus v$.
It follows that $\Pi_v$ is a minimal separator of $G \setminus v$, that contradicts the hypothesis that there is no minimal separator $S\subseteq (N(a_v) \cap N(c_v)) \cup \{a_v,c_v\}$ in $G \setminus v$ and $\{a_v,c_v\} \subseteq S$.
Consequently, $V \setminus (\Pi_v \cup \{v\})$ induces a connected subgraph.

Let $S' \subseteq N(b_v) \setminus (a_v,c_v,v)$ be a minimal $a_vc_v$-separator of $G \setminus (b_v,v)$.
By Lemma~\ref{lem:P3exists}, there exist $x,y \in N(b_v) \setminus (a_v,c_v,v)$ non-adjacent such that $S' = \{x,y\}$.
Finally, let $\Pi'=(x,b_v,y)$ and let $A,C$ be the respective components of $a_v,c_v$ in $G \setminus (\Pi' \cup \{v\})$.
Note that $x,y \in N(A) \cap N(C)$ because $G \setminus v$ is prime by the hypothesis (indeed, neither $x$ nor $b_v$ nor $y$ nor $(b_v,x)$ nor $(b_v,y)$ can be a separator of $G \setminus v$).
Let $P$ be an $xy$-path of $V \setminus (\Pi_v \cup \{v\})$, that exists because $V \setminus (\Pi_v \cup \{v\})$ is connected.
Also, let $A' \subseteq A$ and $C' \subseteq C$ be the respective components of $a_v,c_v$ in $G \setminus (P \cup \Pi' \cup \{v\})$.
Note that the subpath $P \setminus (x,y)$ lies onto a unique component of $G \setminus (\Pi' \cup \{v\})$ because it does not intersect $\Pi_v \cup \{v\}$ by construction, so, $A' = A$ or $C' = C$.
By symmetry, assume that $C' = C$.
There are two cases to consider.

\begin{itemize}
\item Assume $A' = A$ (see Figure~\ref{fig:c4-case-1} for an illustration).
Let us contract the internal vertices of $P$ so as to make vertices $x,y$ adjacent.
Then, let us contract the components $A,C$ to the two vertices $a_v,c_v$, respectively.
Finally, let us contract $v$ to either $a_v$ or $c_v$.
By construction, the five vertices $a_v,b_v,c_v,x,y$ now induce a $K_5$, that contradicts the fact that $G$ is planar by the hypothesis.

\begin{figure}[h!]
 \centering
 \includegraphics[width=0.45\textwidth]{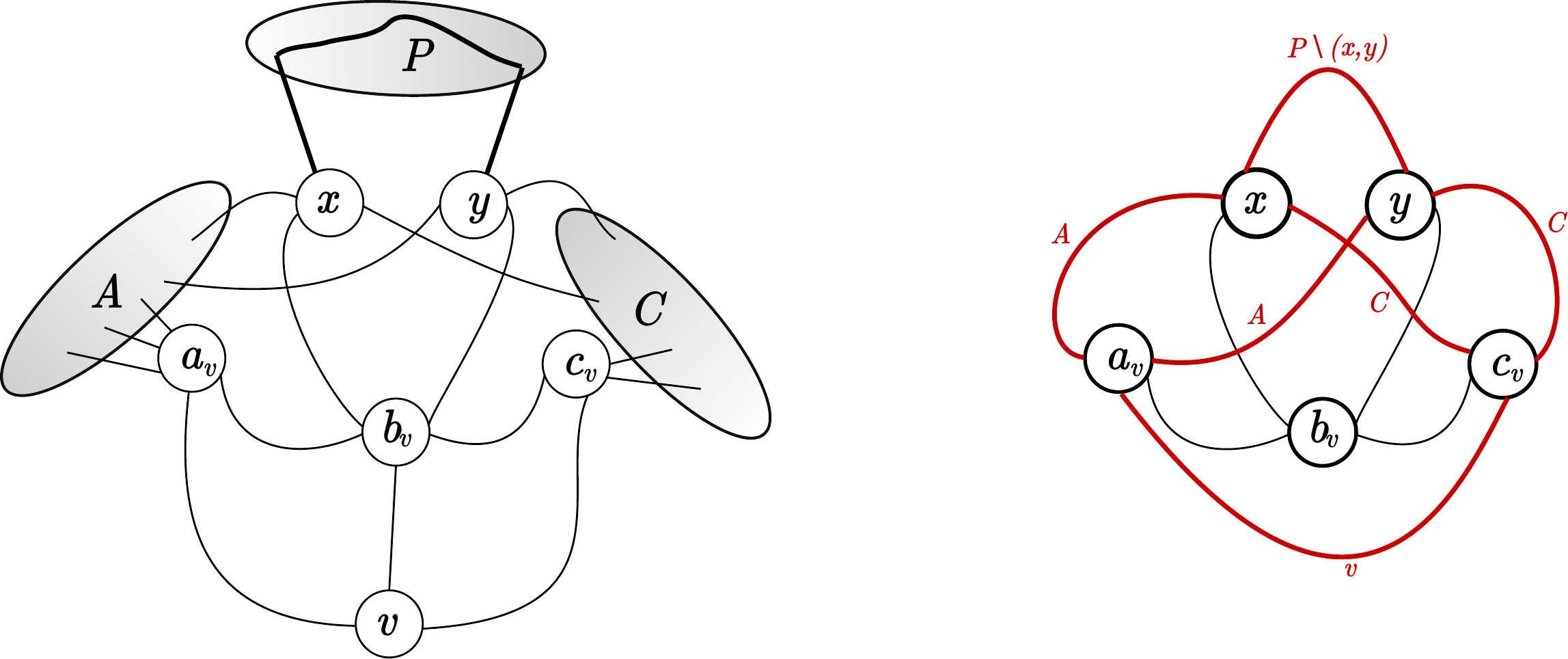}
 \caption{Case $A'=A$ (left). A $K_5$-minor is drawn (right), with edges resulting from contractions labeled in red.}
 \label{fig:c4-case-1}
\end{figure}

\item Else, $A' \neq A$. Equivalently, $P \subseteq A \cup \{x,y\}$ (see Figure~\ref{fig:c4-case-2} for an illustration).
Since $A$ is connected, $N(A') \cap ( P \setminus (x,y) ) \neq \emptyset$.
Let $z \in N(A') \cap P$.
Let us contract the internal vertices of $P$ to vertex $z$.
Then, let us contract the components $A'$ and $C' = C$ to the two vertices $a_v,c_v$, respectively.
Finally let us contract $v$ to either $a_v$ or $c_v$.
By construction, there is a $K_{3,3}$-minor whose sides of the bipartition are $\{a_v,x,y\}$ and $\{b_v,c_v,z\}$, respectively, that contradicts the fact that $G$ is planar by the hypothesis.

\begin{figure}[h!]
 \centering
 \includegraphics[width=0.55\textwidth]{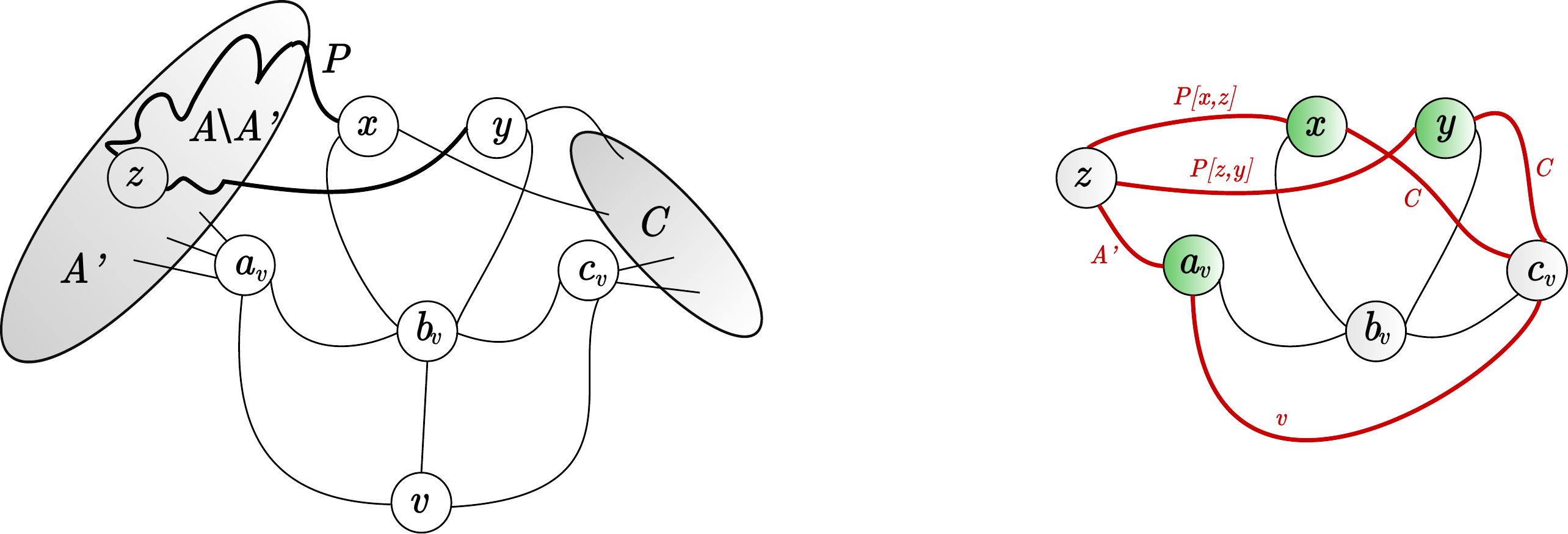}
 \caption{Case $A'\neq A$ (left). A $K_{3,3}$-minor is drawn (right), with each side of the bipartition being coloured differently. Edges resulting from contractions are labeled in red.}
 \label{fig:c4-case-2}
\end{figure}

\end{itemize}

Since both cases contradict the hypothesis that $G$ is planar, therefore, $G = C_4$.
\end{proof}

\subsubsection{Proof of~\ref{step:prime-case}~\ref{step:few-neighbours}~\ref{step:two-neighbours}}
\label{sec:step-4-b-ii}

\begin{theorem}\label{th:primeDifficult}
Let $G = (V,E)$ be a prime planar graph, let $v$ be a leaf-vertex of Type 2 or 3, and let $\Pi_v = (a_v,b_v,c_v)$ be as in Definition~\ref{def:leafVertex}. 

Suppose that all of the following statements hold:
\begin{itemize}
\item $N(a_v) \cap N(c_v) = \{v,b_v,u_v\}$ with $u_v \notin \{v\} \cup \Pi_v$;
\item $V \neq \{a_v,b_v,c_v,u_v,v\}$;
%\item $v$ is an isolated vertex of $G \setminus P$;
%\item $G \setminus (P \cup v)$ is a full component of $G \setminus P$;
%\item $G \setminus v$ has no clique-separator;
%\item finally, $P' = (a,u,c) \notin \mathcal{P}_3(G \setminus v)$.
\item there is no minimal separator $S\subseteq (N(a_v) \cap N(c_v)) \cup \{a_v,c_v\}$ in $G \setminus v$ and $\{a_v,c_v\} \subseteq S$.
\end{itemize}

Let $G'$ be the graph obtained from $G$ by adding edges $\{v,u_v\}$ and $\{b_v,v\}$, then $tb(G) = 1$ if and only if $tb(G') = 1$.
Moreover, $G'$ is planar and prime.
\end{theorem}

\begin{figure}[h!]
 \centering
 \includegraphics[width=0.25\textwidth]{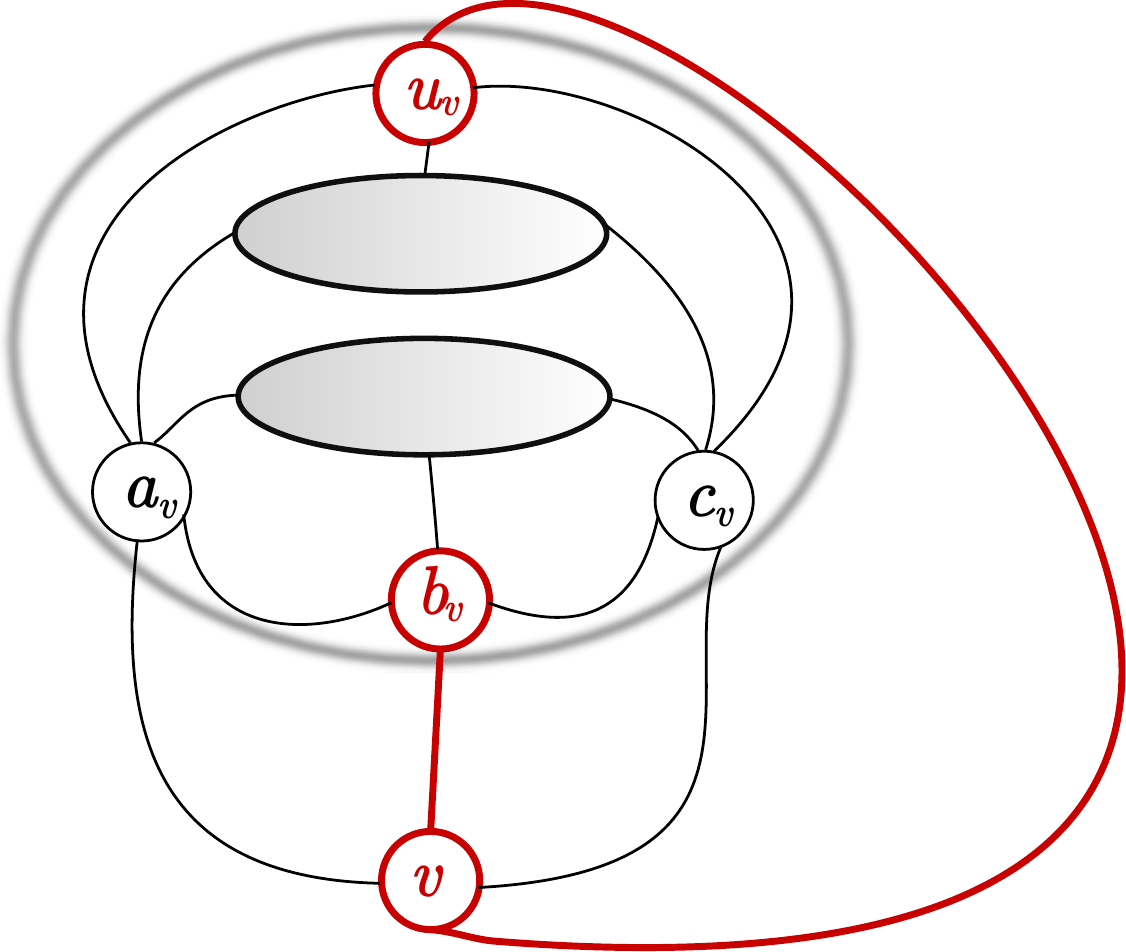}
 \caption{Case of a leaf-vertex $v$ of Type $3$ with $G \setminus v$ is prime and $N(a_v) \cap N(c_v) = \{ u_v, b_v, v \}$. Red edges are those added by Algorithm \texttt{Leaf-BottomUp}.}
 \label{fig:unique-neighbour}
\end{figure}

\begin{proof}
We will first prove that $G \setminus v$ is prime.
By contradiction, let $S'$ be a minimal clique-separator of $G \setminus v$.
By Theorem~\ref{th:GminusVprime?}, there is $w_v \neq v$ such that $S' = \{b_v,w_v\}$, and by Lemma~\ref{lem:separatorInCompo}, $S'$ must be an $a_vc_v$-separator of $G \setminus v$.
Then, it follows that $w_v = u_v \in N(a_v) \cap N(c_v)$, whence $V = \{a_v,b_v,c_v,u_v,v\}$ by Theorem~\ref{th:GminusVprime?}, that contradicts the hypothesis. 
Therefore, $G \setminus v$ is prime.

\medskip
Let us prove that $tb(G) = 1$ implies that $tb(G') = 1$.
Let $(T,{\cal X})$ be a star-decomposition of $G$, which exists by Lemma~\ref{lem:dominator-in-bag}, minimizing the distance in $T$ between the subtrees $T_{a_v}$ and $T_{c_v}$.
Since $N_G(v) \subseteq N_G[b_v]$ then removing $v$ from all bags leaves a tree decomposition of $G \setminus v$ of breadth one.
Up to reducing the tree decomposition, let $(T',{\cal X}')$ be any reduced tree decomposition of $G \setminus v$ that is obtained from $(T,{\cal X})$ by first removing $v$ from the bags.
Note that $(T',{\cal X}')$ is a star-decomposition of $G \setminus v$ by Lemma~\ref{lem:dominator-in-bag}.
Now, there are two cases.
\begin{itemize}

\item Suppose $T'_{a_v} \cap T'_{c_v} \neq \emptyset$.
We will need to prove in this case that the two subtrees $T'_{a_v} \cap T'_{b_v} \cap T'_{c_v}$ and $T'_{a_v} \cap T'_{u_v} \cap T'_{c_v}$ are nonempty and disjoint.

\begin{claim}
	\label{claim:bv-non-adjacent}
	$b_v,u_v$ are non-adjacent in $G$.
\end{claim}

\begin{proofclaim}
	By contradiction, if it were the case that $b_v,u_v$ are adjacent, then by Lemma~\ref{lem:common-neighbour}, either $u_v$ is an isolated vertex of $G \setminus (\Pi_v \cup \{v\})$ --- in which case, $\Pi_v \in {\cal P}_3(G \setminus v)$ because we assume $V \neq \{a_v,b_v,c_v,u_v,v\}$ by the hypothesis ---, or $(a_v,u_v,b_v) \in {\cal P}_3(G \setminus v)$.
	Since $G \setminus v$ is prime, it follows that one of $\Pi_v$ or $\Pi_v \setminus b_v$ must be a minimal separator of $G \setminus v$, similarly one of $(a_v,u_v,c_v)$ or $(a_v,c_v)$ must be a minimal separator of $G \setminus v$.
	Therefore, both cases contradict the hypothesis that there is no minimal separator $S\subseteq (N(a_v) \cap N(c_v)) \cup \{a_v,c_v\}$ in $G \setminus v$ and $\{a_v,c_v\} \subseteq S$, which proves that $b_v,u_v$ are non-adjacent.
\end{proofclaim}

Recall that we are in the case when $T'_{a_v} \cap T'_{c_v} \neq \emptyset$. 
The subtrees $T'_{a_v},T'_{c_v},T'_{u_v}$ are pairwise intersecting, similarly the subtrees $T'_{a_v},T'_{c_v},T'_{b_v}$ are pairwise intersecting.
Therefore, by the Helly property (Lemma~\ref{lem:helly}) $T'_{a_v} \cap T'_{b_v} \cap T'_{c_v} \neq \emptyset$ and $T'_{a_v} \cap T'_{u_v} \cap T'_{c_v} \neq \emptyset$.
Furthermore, $T_{a_v} \cap T_{b_v} \cap T_{c_v} \cap T_{u_v} = \emptyset$, because since $b_v,u_v$ are non-adjacent by Claim~\ref{claim:bv-non-adjacent} no vertex dominates all of $\{a_v,b_v,c_v,u_v\}$ in $G$, and so, $T'_{a_v} \cap T'_{b_v} \cap T'_{c_v} \cap T'_{u_v} = \emptyset$.

\begin{claim}
\label{claim:subtrees-adjacent}
The subtrees $T'_{a_v} \cap T'_{b_v} \cap T'_{c_v}$ and $T'_{a_v} \cap T'_{u_v} \cap T'_{c_v}$ are adjacent in $T'$.
\end{claim}

\begin{proofclaim}
By contradiction, let $B$ be an internal bag onto the path between both subtrees in $T'$, let $z \in B$ dominate the bag.
Note that $a_v,c_v \in B$ by the properties of the tree decomposition, $z \notin \{a_v,c_v\}$ because $a_v,c_v$ are non-adjacent, and so, $z \in N(a_v) \cap N(c_v) \setminus v = \{ u_v, b_v\}$.
This contradicts the fact that $B \notin T'_{a_v} \cap T'_{b_v} \cap T'_{c_v}$ and $B \notin T'_{a_v} \cap T'_{u_v} \cap T'_{c_v}$, therefore, the subtrees $T'_{a_v} \cap T'_{b_v} \cap T'_{c_v}$ and $T'_{a_v} \cap T'_{u_v} \cap T'_{c_v}$ are adjacent in $T'$.
\end{proofclaim}

Finally, let $B \in T'_{a_v} \cap T'_{b_v} \cap T'_{c_v}, B' \in T'_{a_v} \cap T'_{u_v} \cap T'_{c_v}$ be adjacent, that exist by Claim~\ref{claim:subtrees-adjacent}.
Observe that $b_v$ dominates $B$, $u_v$ dominates $B'$.
To obtain a star-decomposition of $G'$ from $(T',{\cal X}')$, it now suffices to add vertex $v$ in $B$ and $B'$, whence $tb(G') = 1$.

\item Else, $T'_{a_v} \cap T'_{c_v} = \emptyset$.
This implies $T_{a_v} \cap T_{c_v} = \emptyset$.
%Moreover, we claim that for every $a_vc_v$-separator $S'$ of $G$, if there is $z \notin \{a_v,c_v\}$ such that $S' \subseteq N[z]$ then $z \in N(a_v) \cap N(c_v) \setminus v$.
%Indeed, vertex $z$ dominates $N(a_v) \cap N(c_v) \subseteq S'$, and all the vertices dominating $N(a_v) \cap N(c_v)$ are amongst $(N(a_v) \cap N(c_v)) \cup \{a_v,c_v\}$ because $G$ is $K_{3,3}$-minor-free by the hypothesis.
Since the tree decomposition $(T,{\cal X})$ minimizes the distance in $T$ between $T_{a_v}$ and $T_{c_v}$, $G$ is planar and $|N(a_v) \cap N(c_v)| \geq 3$, therefore by Corollary~\ref{cor:strong-three}, the subtrees $T_{a_v}$ and $T_{c_v}$ are adjacent in $T$, whence the subtrees $T'_{a_v},T'_{c_v}$ are also adjacent in $T'$.
In particular, by Corollary~\ref{cor:strong-three} there exist two adjacent bags $B'_{a_v},B'_{c_v} \in {\cal X}'$ such that $a_v \in B'_{a_v} \setminus B'_{c_v}, B'_{a_v} \cap B'_{c_v} = N_G(a_v) \cap N_G(c_v) \setminus v = \{u_v,b_v\}, c_v \in B'_{c_v} \setminus B'_{a_v}$.
Furthermore, $a_v$ dominates $B'_{a_v}$ while $c_v$ dominates $B'_{c_v}$.
Therefore, in order to obtain a star-decomposition of $G'$ from $(T',{\cal X}')$, it now suffices to add vertex $v$ in $B'_{a_v}$ and $B'_{c_v}$ -- that yields exactly $(T,{\cal X})$ --, whence $tb(G') = 1$.
\end{itemize}

\medskip
Before we can prove the equivalence, {\it i.e.}, $tb(G) = 1$ if and only if $tb(G') = 1$, we need to prove first that $G'$ is prime and planar.

\begin{claim}
\label{claim:gprime-prime}
$G'$ is prime.
\end{claim}

\begin{proofclaim}
Let $S'$ be a clique-separator of $G'$.
Note that $v \in S'$ by construction of $G'$.
Therefore, $S' \setminus v$ is a clique-separator of $G \setminus v$, that contradicts the fact that $G \setminus v$ is prime.
Consequently, $G'$ is prime.
\end{proofclaim}

\begin{claim}
\label{claim:gprime-planar}
$G'$ is planar.
\end{claim}

\begin{proofclaim}
Let us fix a plane embedding of $G$.
By Jordan Theorem, the cycle induced by $(a_v,b_v,c_v,u_v)$ separates the plane into two regions.
Let $G_1,G_2$ be respectively the subgraphs of $G$ that are induced by all the vertices in each region.	

\begin{figure}[h!]
	\centering
	\includegraphics[width=0.25\textwidth]{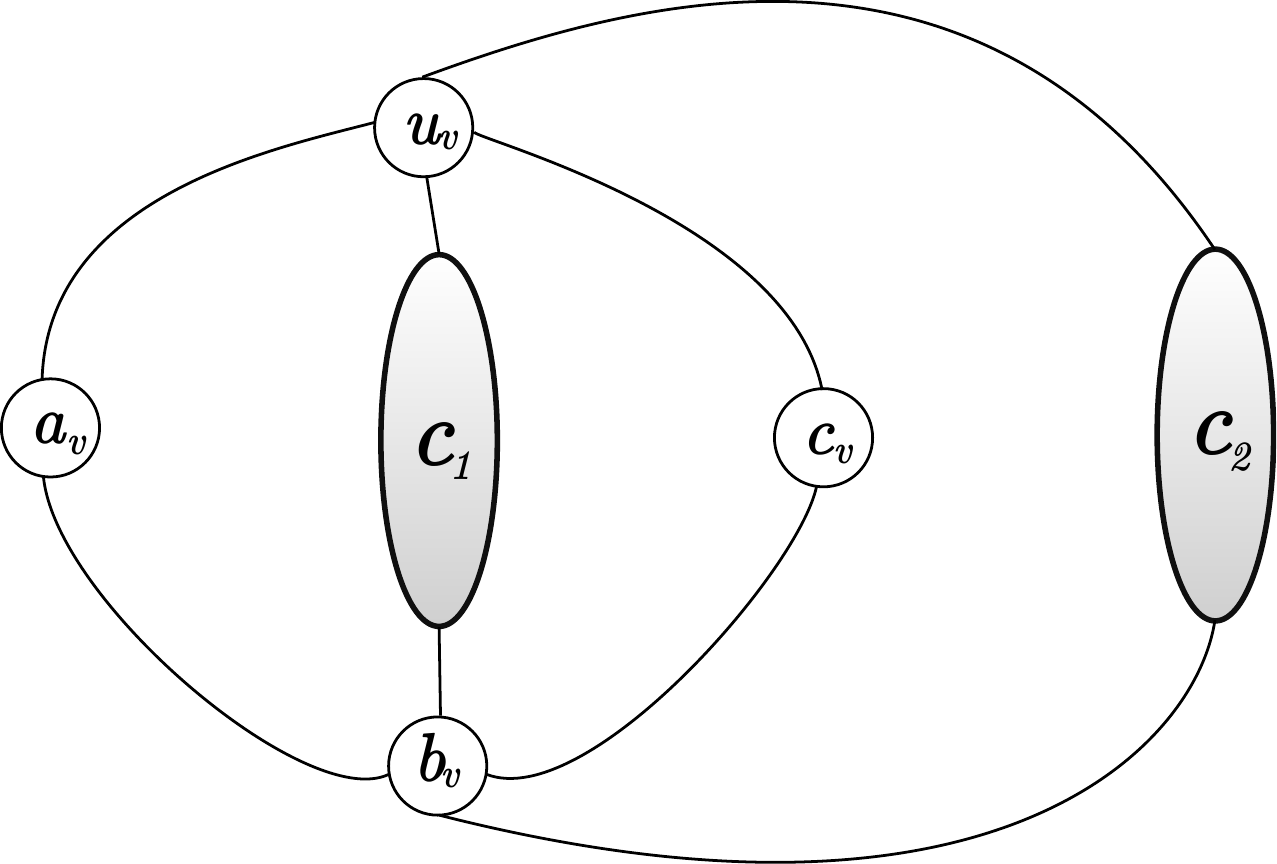}
	\caption{Proof that the graph $G'$ of Theorem~\ref{th:primeDifficult} is planar.}
	\label{fig:planar-one-neighbour}
\end{figure}

We claim that either $V \setminus (a_v,b_v,c_v,u_v,v) \subseteq V(G_1)$, or $V \setminus (a_v,b_v,c_v,u_v,v) \subseteq V(G_2)$.
Note that it will prove that $G'$ is planar, because then drawing vertex $v$ onto the region that does not contain the set $V \setminus (a_v,b_v,c_v,u_v,v)$ yields a planar embedding of $G'$.
By contradiction, let $C_1 \subseteq V(G_1), C_2 \subseteq V(G_2)$ be connected components of $V \setminus (a_v,b_v,c_v,u_v,v)$.
Let $\Pi'_v = (a_v,u_v,c_v)$.
If one of $\Pi_v$ or $\Pi'_v$ belongs to ${\cal P}_3(G \setminus v)$, then, there exists a minimal separator $S\subseteq (N(a_v) \cap N(c_v)) \cup \{a_v,c_v\}$ in $G \setminus v$ and since $G \setminus v$ is prime, $\{a_v,c_v\} \subseteq S$.
This would contradict the hypothesis, so, $\Pi_v,\Pi'_v \notin \mathcal{P}_3(G \setminus v)$. 
As a result, since $(a_v,b_v,c_v,u_v) = \Pi_v \cup \Pi'_v$ separates $C_1$ from $C_2$, therefore, $u_v,b_v \in N(C_1) \cap N(C_2)$ (or else, $\Pi_v \in {\cal P}_3(G \setminus v)$ or $\Pi'_v \in {\cal P}_3(G \setminus v)$).
Let us remove all other components of $V \setminus (a_v,b_v,c_v,u_v,v)$ but $C_1$ and $C_2$, and let us remove all edges between $\{a_v,c_v\}$ and $C_1 \cup C_2$ if any (see Figure~\ref{fig:planar-one-neighbour}).
Finally, let us contract $C_1,C_2$ to the two vertices $x_1,x_2$.
The cycle induced by $(u_v,x_1,b_v,x_2)$ separates the plane into two regions with $a_v,c_v$ being into different regions by construction.
Vertex $v$ must belong to one of the regions, but then it is a contradiction because $v \in N(a_v) \cap N(c_v)$ by the hypothesis.
\end{proofclaim}

To conclude the proof, let us prove that conversely, $tb(G') = 1$ implies that $tb(G) = 1$.
Let $(T',{\cal X}')$ be a star-decomposition of $G'$ minimizing the distance in $T'$ between the subtrees $T'_{a_v}$ and $T'_{c_v}$.
As an intermediate step, we claim that if removing vertex $v$ from all bags of ${\cal X}'$ leaves a tree decomposition of $G \setminus v$ of breadth one, then it implies that $tb(G) = 1$.
To prove the claim, there are two cases to be considered.
\begin{itemize}
\item If $T'_{a_v} \cap T'_{c_v} \neq \emptyset$, then the subtrees $T'_{a_v}, T'_{b_v}, T'_{c_v}$ are pairwise intersecting, hence by the Helly property (Lemma~\ref{lem:helly}) $T'_{a_v} \cap T'_{b_v} \cap T'_{c_v} \neq \emptyset$.
Equivalently there is bag containing $\Pi_v$, and so it suffices to remove $v$ from all bags and then to make any bag containing $\Pi_v$ adjacent to the new bag $N_G[v]$ in order to obtain a tree decomposition of $G$ of breadth one.

\item Else, $T'_{a_v} \cap T'_{c_v} = \emptyset$.
Since $(T',{\cal X}')$ minimizes the distance in $T'$ between the subtrees $T'_{a_v}$ and $T'_{c_v}$, $G'$ is planar by Claim~\ref{claim:gprime-planar} and $a_v,c_v$ have three common neighbours in $G'$, therefore, by Corollary~\ref{cor:strong-three} there must exist two adjacents bags $B'_{a_v},B'_{c_v}$ such that $a_v \in B'_{a_v} \setminus B'_{c_v}, B'_{a_v} \cap B'_{c_v} = N(a_v) \cap N(c_v)$ and $c_v \in B'_{c_v} \setminus B'_{a_v}$.
Furthermore, vertex $a_v$ dominates the bag $B'_{a_v}$, while vertex $c_v$ dominates the bag $B'_{c_v}$.
As a result, removing vertex $v$ from all bags but $B'_{a_v},B'_{c_v}$ leads to a tree decomposition of $G$ of breadth one.
\end{itemize}

\medskip
\noindent
Consequently, we are left to modify the tree decomposition $(T',{\cal X}')$ so as to ensure that none of the bags is only dominated by vertex $v$ in $G'$, for if it is the case then removing $v$ from all bags does leave a tree decomposition of $G \setminus v$ of breadth one.
We will call the latter property the \emph{removal property}.
Observe that if it is the case that $(T',{\cal X}')$ does not satisfy the removal property, then there must be a bag $B$ fully containing $N_{G'}(v)$ because any strict subset of $N_{G'}(v)$ is dominated by some vertex of $G \setminus v$.
In particular, $B = N_{G'}[v]$ because only vertex $v$ dominates $N_{G'}(v)$ in $G'$, and so we can further assume that $T'_v = \{B\}$ without violating the property for $(T',{\cal X}')$ to be a tree decomposition of $G'$ of breadth one.
Therefore in the following, assume that $(T',{\cal X}')$ is a reduced star-decomposition of $G'$ and $T'_v = \{B\}$, that is always possible to achieve by Lemma~\ref{lem:dominator-in-bag} and above remarks.

Since $V \neq \{a_v,b_v,c_v,u_v,v\} = N_{G'}[v]$ by the hypothesis, therefore, ${\cal X}' \setminus B \neq \emptyset$.
Let $B'$ be adjacent to $B$ in $T'$. 
Note that $B \cap B' \neq \{a_v,b_v,c_v,u_v\}$ because no other vertex than $v$ dominates the subset $\{a_v,b_v,c_v,u_v\}$ in $G'$.
By the properties of a tree decomposition, $B \cap B'$ is a separator of $G'$.
Consequently, $B \cap B'$ is not a clique because $G'$ is prime by Claim~\ref{claim:gprime-prime}.
Furthermore, since $B \cap B' \neq \{a_v,b_v,c_v,u_v\}$ it holds that $B \setminus (B' \cup \{v\}) \neq \emptyset$, consequently $B \cap B'$ is also a separator of $G \setminus v$.
Since $G \setminus v$ is prime, $B \cap B'$ cannot be any of $(a_v,c_v)$, $\Pi_v$ or $\Pi'_v$ because by the hypothesis there is no minimal separator $S\subseteq (N(a_v) \cap N(c_v)) \cup \{a_v,c_v\}$ in $G \setminus v$ and $\{a_v,c_v\} \subseteq S$. 
It follows that $B \cap B' \subseteq \{a_v,b_v,u_v\}$ or $B \cap B' \subseteq \{b_v,c_v,u_v\}$.
Let us substitute the bag $B$ with the two adjacent bags $B_1 = \{a_v,u_v,b_v,v\}, B_2 = \{b_v,c_v,u_v,v\}$, then we make adjacent all bags $B''$ that were formerly adjacent to $B$ to some bag amongst $B_1,B_2$ containing $B \cap B''$.
Note that $B_1 \subseteq N[a_v]$ and that $B_2 \subseteq N[c_v]$.
Therefore, the resulting tree decomposition is a tree decomposition of $G$ of breadth one such that $v$ dominates no bag.
\end{proof}

\subsubsection{Case of leaf-vertex $v$ of Type $2$ or $3$ and $G \setminus v$ not prime}

The remaining subsections will be devoted to the proof of correctness of~\ref{step:new-clique-sep}.
In particular, this subsection is devoted to the proof that when $G \setminus v$ is not prime one can only consider the case when the leaf-vertex $v$ is of Type 2, {\it i.e.}, $v$ and $b_v$ are adjacent in $G$.
Note that when $v$ is of Type 3, then in general one cannot add an edge between $v$ and $b_v$ without violating the property for the graph $G$ to be planar, as shown in Figure~\ref{fig:planar-unique-neighbour}.
We will now prove that whenever we are in the conditions of~\ref{step:new-clique-sep}, it is always possible to do so while preserving the planarity of the graph $G$ and the property to be of tree-breadth one.

\begin{figure}[h!]
 \centering
 \includegraphics[width=0.35\textwidth]{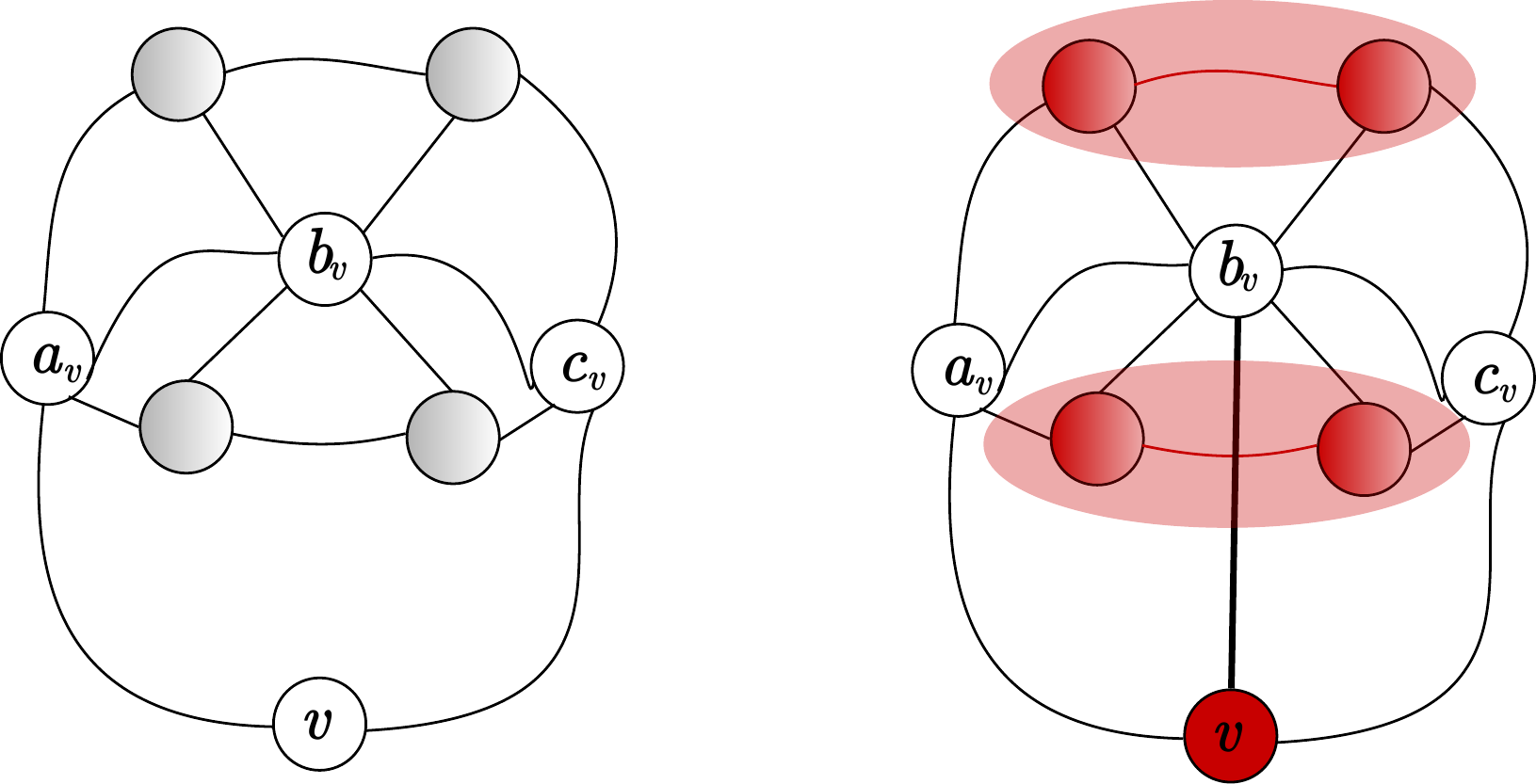}
 \caption{A planar graph $G$ with $tb(G) = 1$ (left), and a leaf-vertex $v$ of Type 3 so that adding an edge between $v$ and $b_v$ violates the property for the graph to be planar (right). In the latter case, one side of the bipartition of the $K_{3,3}$-minor is coloured red.}
 \label{fig:planar-unique-neighbour}
\end{figure}

\begin{theorem}\label{th:addEdgebv}
Let $G$ be a prime planar graph. Let $v$ be a leaf-vertex of Type $3$ such that $G \setminus v$ is not prime. Finally, let $\Pi_v=(a_v,b_v,c_v)$ be as in Definition~\ref{def:leafVertex}. Let $G'$ be obtained from $G$ by adding the edge $\{v,b_v\}$.

Then, $G'$ is prime and planar, and $tb(G)=1$ if and only if $tb(G')=1$.
\end{theorem}

\begin{proof}
First, we prove that $G'$ is prime and planar.
\begin{itemize}

\item In order to prove that $G'$ is prime, by contradiction let $S$ be a clique-separator of $G'$.
Since $G'$ is a supergraph of $G$, therefore $S$ is a separator of $G$ but it does not induce a clique in $G$.
Hence, $S$ contains the edge $\{v,b_v\}$, and so either $S \subseteq \{a_v,b_v,v\}$ or $S \subseteq \{b_v,c_v,v\}$.
Let $C = V \setminus (\Pi_v \cup \{v\})$, by Theorem~\ref{th:GminusVprime?}, $C$ is a full component of $G \setminus \Pi_v$ because $G \setminus v$ is not prime.
In particular, $C$ is connected and $a_v,c_v \in N(C)$, that contradicts the fact that $G' \setminus S$ is unconnected.

\item Then in order to prove that $G'$ is planar, let us fix a plane embedding of $G$.
The cycle induced by $(a_v,b_v,c_v,v)$ separates the plane into two regions.
To prove that $G'$ is planar, we claim that it suffices to prove that all vertices in $C = V \setminus (\Pi_v \cup \{v\})$ are in the same region, for then drawing the edge $\{b_v,v\}$ in the other region leads to a plane embedding of $G'$.
By contradiction, let $x,y \in C$ be in different regions.
By~\cite[Proposition 8]{Bouchitte2003}, the cycle $(a_v,b_v,c_v,v)$ is an $xy$-separator of $G$, that contradicts the fact that $C$ is connected.

\end{itemize}

\medskip
\noindent
Let us now prove that $tb(G) = 1$ implies that $tb(G') = 1$.
Let $(T,{\cal X})$ be a star-decomposition of $G$, that exists by Lemma~\ref{lem:dominator-in-bag}, minimizing the distance in $T$ between the subtrees $T_{a_v}$ and $T_{c_v}$.
Let us remove vertex $v$ from all bags, that leads to a tree decomposition $(T, {\cal X}_{-v})$ of $G \setminus v$ of breadth one because $N_G(v) \subseteq N_G(b_v)$. 
Then, let $(T',{\cal X}')$ be any reduced tree decomposition that is obtained from $(T,{\cal X}_{-v})$, that is a star-decomposition of $G \setminus v$ by Lemma~\ref{lem:dominator-in-bag}. 
Now, there are two cases.
If $T'_{a_v} \cap T'_{c_v} \neq \emptyset$, then the subtrees $T'_{a_v},T'_{b_v},T'_{c_v}$ are pairwise intersecting and so, by the Helly property (Lemma~\ref{lem:helly}) $T'_{a_v} \cap T'_{b_v} \cap T'_{c_v} \neq \emptyset$.
Hence one obtains a star-decomposition of $G'$ simply by making some bag of $T'_{a_v} \cap T'_{b_v} \cap T'_{c_v}$ adjacent to the new bag $N_{G'}[v] = \{a_v,b_v,c_v,v\}$.
Else, $T'_{a_v} \cap T'_{c_v} = \emptyset$, so, $T_{a_v} \cap T_{c_v} = \emptyset$.
Since $\Pi_v \in {\cal P}_3(G)$ and $G$ is prime by the hypothesis, therefore, one of $\Pi_v$ or $\Pi_v \setminus b_v$ must be a minimal separator of $G$.
As a result, since $(T,{\cal X})$ is assumed to minimize the distance in $T$ between the subtrees $T_{a_v}$ and $T_{c_v}$, by Corollary~\ref{cor:strong-sep} there exist two adjacent bags $B_{a_v}, B_{c_v} \in {\cal X}$ so that $a_v \in B_{a_v} \setminus B_{c_v}$ and $c_v \in B_{c_v} \setminus B_{a_v}$ respectively dominate the bags $B_{a_v}$ and $B_{c_v}$.
In such case, $B_{a_v} \cap B_{c_v} = N_G(a_v) \cap N_G(c_v)$ and so, since $b_v,v \in B_{a_v} \cap B_{c_v}$, $(T,{\cal X})$ is also a star-decomposition of $G'$.
So, in conclusion, $tb(G') = 1$ in both cases.

Conversely, let us prove that $tb(G') = 1$ implies that $tb(G) = 1$.
Let $(T',{\cal X}')$ be a star-decomposition of $G'$, that exists by Lemma~\ref{lem:dominator-in-bag}, minimizing the distance in $T'$ between the subtrees $T'_{a_v}$ and $T'_{c_v}$.
Let us remove vertex $v$ from all bags, that leads to a tree decomposition $(T', {\cal X}'_{-v})$ of $G' \setminus v = G \setminus v$ of breadth one because $N_{G'}[v] \subseteq N_{G'}[b_v]$. 
Then, let $(T,{\cal X})$ be any reduced tree decomposition that is obtained from $(T',{\cal X}'_{-v})$, that is a star-decomposition of $G \setminus v$ by Lemma~\ref{lem:dominator-in-bag}. 
There are two cases.
If $T_{a_v} \cap T_{c_v} \neq \emptyset$, then one obtains a star-decomposition of $G$ simply by making some bag of $T_{a_v} \cap T_{c_v}$ adjacent to the new bag $N_G[v] = \{a_v,c_v,v\}$.
Else, $T_{a_v} \cap T_{c_v} = \emptyset$, so, $T'_{a_v} \cap T'_{c_v} = \emptyset$.
Since $\Pi_v \in {\cal P}_3(G')$ and $G'$ is also prime, therefore, one of $\Pi_v$ or $\Pi_v \setminus b_v$ must be a minimal separator of $G'$.
As a result, since $(T',{\cal X}')$ is assumed to minimize the distance in $T'$ between the subtrees $T'_{a_v}$ and $T'_{c_v}$, by Corollary~\ref{cor:strong-sep} there exist two adjacent bags $B'_{a_v}, B'_{c_v} \in {\cal X}'$ so that $a_v \in B'_{a_v} \setminus B'_{c_v}$ and $c_v \in B'_{c_v} \setminus B'_{a_v}$ respectively dominate the bags $B'_{a_v}$ and $B'_{c_v}$.
In such case, one obtains a star-decomposition of $G$ by adding $v$ in the two bags $B'_{a_v},B'_{c_v}$.
So, in conclusion, $tb(G) = 1$ in both cases.

\end{proof}

\subsubsection{Proof of~\ref{step:new-clique-sep}~\ref{step:no-separation}}

\begin{theorem}\label{th:contractva}
Let $G$ be a prime planar graph, let $v$ be a leaf-vertex of Type 2, $\Pi_v = (a_v,b_v,c_v)$ be as in Definition~\ref{def:leafVertex}, and let $u_v \notin \Pi_v \cup \{v\}$ be such that $(b_v,u_v)$ is an edge-separator of $G \setminus v$.  

Suppose $a_v$ and $u_v$ are non-adjacent, and either $c_v$ and $u_v$ are non-adjacent or the subset $N_G(a_v) \cap N_G(u_v)$ is \underline{not} an $a_vu_v$-separator in the subgraph $G \setminus (c_v,v)$.

Then, $G/ va_v$ (obtained by contracting $\{v,a_v\}$) is planar and prime and $tb(G) = 1$ if and only if $tb(G/ va_v) = 1$.
\end{theorem}

\begin{figure}[h!]
 \centering
 \includegraphics[width=0.5\textwidth]{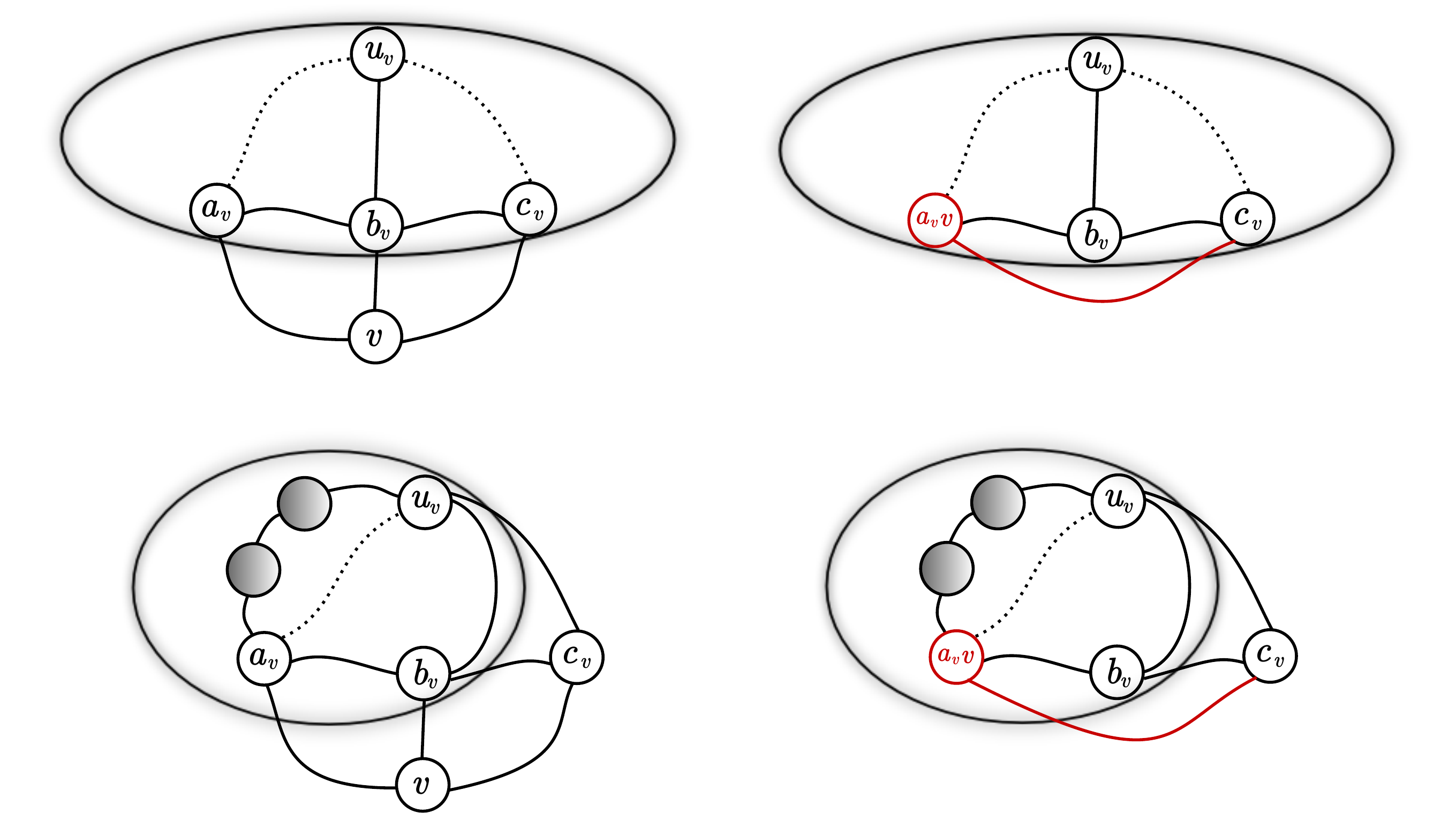}
 \caption{Cases when Theorem~\ref{th:contractva} applies and the edge $\{v,a_v\}$ can be contracted to $a_v$.}
 \label{fig:planar-clique-contraction}
\end{figure}

\begin{proof}
The graph $G/va_v$ is a contraction of the planar graph $G$, therefore it is planar.
Let us prove that $G/va_v$ is prime.
By contradiction, let $S$ be a minimal clique-separator of $G/va_v$.
Since $G/va_v$ is a supergraph of $G \setminus v$, $S$ is also a separator of $G \setminus v$.
Furthermore, it is not an $a_vc_v$-separator because $a_v,c_v$ are adjacent in $G/va_v$, therefore, by Lemma~\ref{lem:separatorInCompo} $S$ is a separator of $G$.
Since $G$ is prime by the hypothesis, $S$ does not induce a clique of $G$, whence $a_v,c_v \in S$.
However, since $(b_v,u_v)$ is not a separator of $G$ because $G$ is prime by the hypothesis, therefore by Lemma~\ref{lem:separatorInCompo} $(u_v,b_v)$ is an $a_vc_v$-separator of $G \setminus v$.
So, $N_G(a_v) \cap N_G(c_v) \subseteq \{v,b_v,u_v  \}$, that implies $N_G(a_v) \cap N_G(c_v) = \{v,b_v\}$ because $a_v$ and $u_v$ are non-adjacent by the hypothesis. 
In such a case $S \subseteq \Pi_v$, but then $V \setminus (\Pi_v \cup \{v\})$ cannot be a full component of $G \setminus \Pi_v$, thus contradicting Theorem~\ref{th:GminusVprime?}.
As a result, the graph $G/ va_v$ is planar and prime.

\medskip
\noindent
If $tb(G) = 1$ then $tb(G/va_v) = 1$ because tree-breadth is contraction-closed by Lemma~\ref{lem:contraction-closed}.
Conversely, let us prove that $tb(G/va_v) = 1$ implies $tb(G) = 1$.
To show this, let $(T,{\cal X})$ be a star-decomposition of $G/va_v$, that exists by Lemma~\ref{lem:dominator-in-bag}, minimizing the number of bags $|{\cal X}|$ (in particular, $(T,{\cal X})$ is a reduced tree decomposition).
Assume moreover $(T,{\cal X})$ to minimize the number of bags that are not contained into the closed neighbourhood of some vertex in $G$ w.r.t. this property.
Note that there is a bag of $(T,{\cal X})$ containing $\Pi_v$, because since it is a clique of $G/va_v$ the subtrees $T_{a_v},T_{b_v},T_{c_v}$ are pairwise intersecting and so, by the Helly property (Lemma~\ref{lem:helly}), $T_{a_v} \cap T_{b_v} \cap T_{c_v} \neq \emptyset$.
So, we can add in $(T,{\cal X})$ a new bag $N_G[v]$, and by making this bag adjacent to any bag of $T_{a_v} \cap T_{b_v}\cap T_{c_v}$ one obtains a tree decomposition of $G$ (not necessarily a star-decomposition). 
Consequently, we claim that to prove that $tb(G) = 1$, it suffices to prove that $(T,{\cal X})$ is a star-decomposition of $G \setminus v$, for then the above construction leads to a star-decomposition of $G$.

By contradiction, suppose it is not the case that $(T,{\cal X})$ is a star-decomposition of $G\setminus v$.
Since $G/va_v$ and $G\setminus v$ only differ in the edge $\{a_v,c_v\}$, there must be a bag $B$ of $T_{a_v} \cap T_{c_v}$ that is only dominated by some of $a_v,c_v$.
We make the stronger claim that the bag $B$ has a unique dominator, that is either $a_v$ or $c_v$.
Since $B$ is only dominated by some of $a_v,c_v$, then in order to prove the claim by contradiction we only need to consider the case when $B \subseteq N_{G/va_v}[a_v] \cap N_{G/va_v}[c_v]$.
Recall that $N_{G/va_v}[a_v] \cap N_{G/va_v}[c_v] = \{a_v,b_v,c_v\}$ by the above remarks (because $(u_v,b_v)$ is an $a_vc_v$-separator of $G \setminus v$), therefore either $B = \{a_v,b_v,c_v\}$ or $B = \{a_v,c_v\}$.
In the first case ($B = \{a_v,b_v,c_v\}$) we have that $B \subseteq N[b_v]$, thus contradicting the fact that $B$ is only dominated by some of $a_v,c_v$.
However in the second case ($B = \{a_v,c_v\}$) the bag $B$ is strictly contained in any bag of the nonempty subtree $T_{a_v} \cap T_{b_v} \cap T_{c_v}$, thus contradicting the fact that $(T,{\cal X})$ is a reduced tree decomposition by minimality of $|{\cal X}|$. 
Therefore, the claim is proved and so, the bag $B$ has a unique dominator, that is either $a_v$ or $c_v$.
Note that if $B \subseteq N_{G/va_v}[c_v]$ then we may further assume that $c_v,u_v$ are nonadjacent, or else by Theorem~\ref{th:GminusVprime?} $N_{G/va_v}[c_v] = \{a_v,b_v,c_v,u_v\}  \subseteq N[b_v]$ and so, $B \subseteq N[b_v]$, that would contradict the claim that $B$ is only dominated by some of $a_v,c_v$.
In addition, since $a_v$ and $c_v$ play symmetrical roles in the case when $u_v,c_v$ are nonadjacent, let us assume w.l.o.g. that vertex $a_v$ is the sole dominator of the bag $B$.

In such a case, $N_{G/va_v}(a_v) \cap N_{G/va_v}(c_v) = \{b_v\}$ because $(u_v,b_v)$ is an $a_vc_v$-separator of $G \setminus v$, so, since $N(c_v) \setminus (\Pi_v \cup \{v\}) \neq \emptyset$ because $G$ is prime by the hypothesis, the existence of a bag $B'$ containing vertex $c_v$ and adjacent to $B$ follows. 
By the properties of a tree decomposition, $B \cap B'$ is a separator of $G/va_v$
Now, let $C_a$ be the component of vertex $a_v$ in $G \setminus (b_v,u_v,v)$.
Observe that $c_v \notin C_a$ because $(u_v,b_v)$ is an $a_vc_v$-separator of $G \setminus v$.
Since $B \cap B' \subseteq N_{G/va_v}[a_v] \subseteq C_a \cup \Pi_v$, therefore, $B \cap B' \cap C_a \neq \emptyset$ or else $B \cap B'$ would be a clique-separator in $G/va_v$ (impossible since it is a prime graph). 
There are several cases to be considered depending on the dominators of bag $B'$.

\begin{itemize}
\item If $a_v$ dominates $B'$ then $B,B'$ can be merged into one, thus contradicting the minimality of $|{\cal X}|$;
\item Else, $B'$ must be dominated by one of $b_v$ or $u_v$ because $B \cap B' \cap C_a \neq \emptyset$, $c_v \in (B \cap B') \setminus C_a$ and $(b_v,u_v)$ separates $c_v$ from $C_a$.
In fact, we claim that it cannot be dominated by vertex $u_v$.
By contradiction, suppose that it is the case.
Since $a_v$ and $u_v$ are non-adjacent, therefore, $a_v \in B \setminus B'$ and $u_v \in B' \setminus B$.
So, it follows by the properties of a tree decomposition that $B \cap B'$ is an $a_vu_v$-separator of $G/va_v$.
However, $B \cap B' \subseteq N(a_v) \cap N(u_v)$, that contradicts the hypothesis that $N_G(a_v) \cap N_G(u_v)$ is not an $a_vu_v$-separator in the subgraph $G \setminus (c_v,v)$.

Therefore, $b_v \in B'$ dominate the bag.
Observe that if it were the case that there are at least two bags that are both adjacent to $B$ and dominated by $b_v$, then they could all be merged into one without violating the property for $(T,{\cal X})$ to be a star-decomposition.
As a result, by minimality of $|{\cal X}|$, $B'$ is the unique bag that is both adjacent to $B$ and dominated by $b_v$, whence it is also the unique bag adjacent to $B$ containing vertex $c_v$. 
Let us substitute the two bags $B,B'$ with $B \setminus c_v, B' \cup \{a_v\}$.
Since $N_{G/va_v}(a_v) \cap N_{G/va_v}(c_v) = \{b_v\}$, it is still a star-decomposition of $G / va_v$ with equal number of bags $|{\cal X}|$.
Furthermore, there is one less bag that is not contained in the closed neighbourhood of some vertex in $G$, thus contradicting the minimality of $(T,{\cal X})$.
\end{itemize}
\end{proof}

\subsubsection{Proof of~\ref{step:new-clique-sep}~\ref{step:separation}~\ref{step:contraction-case} and~\ref{step:new-clique-sep}~\ref{step:separation}~\ref{step:diamond-case}}

In order to deal with all remaining cases, it will require us to further study the neighbourhood of vertex $b_v$ in the graph.
Observe that in the following Theorem~\ref{claim:connect-vertex-b} we needn't prove that the resulting graph $G'$ is prime because it will be proved in Theorem~\ref{lem:clique-case-2}.

\begin{theorem}\label{claim:connect-vertex-b}
Let $G$ be a prime planar graph, let $v$ be a leaf-vertex of Type 2, $\Pi_v = (a_v,b_v,c_v)$ be as in Definition~\ref{def:leafVertex}, and let $u_v \notin \Pi_v \cup \{v\}$ be such that $(b_v,u_v)$ is an edge-separator of $G \setminus v$.   

Suppose $u_v \in N(c_v) \setminus N(a_v)$, $N(a_v) \cap N(u_v)$ is an $a_vu_v$-separator of $G \setminus (c_v,v)$, and $N(b_v) = \{a_v,c_v,u_v,v\}$.

Then, there exists $x \in (N(a_v) \cap N(u_v)) \setminus b_v$ such that the graph $G'$, obtained from $G$ by adding the edge $\{b_v,x\}$, is planar and satisfies $tb(G') = 1$ if $tb(G) = 1$.
Moreover, the vertex $x$ can be found in linear-time.
\end{theorem}

\begin{figure}[h!]
 \centering
 \includegraphics[width=0.45\textwidth]{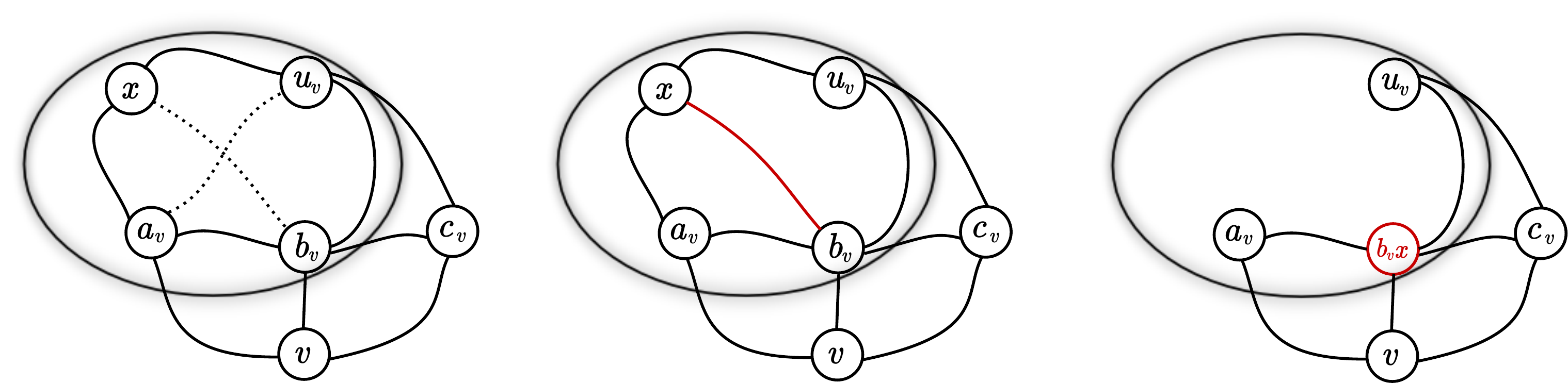}
 \caption{Cases when one of Theorem~\ref{claim:connect-vertex-b} or Theorem~\ref{lem:clique-case-2} applies and vertex $b_v$ can be eventually contracted to another vertex.}
 \label{fig:planar-clique-contraction}
\end{figure}

\begin{proof}
First, we claim that $(a_v,u_v)$ is a minimal $2$-separator of $G$.
Indeed, by the hypothesis $c_v$ and $u_v$ are adjacent, therefore, by Theorem~\ref{th:GminusVprime?} $N_G(c_v) = \{b_v,u_v,v\}$.
In addition, $N(b_v) = \{a_v,c_v,u_v,v\}$ by the hypothesis.
Last, since $G$ is prime by the hypothesis, therefore, $N(a_v) \setminus (\Pi_v \cup \{v\}) \neq \emptyset$, and so, since $a_v$ and $u_v$ are non-adjacent by the hypothesis, $V(G) \setminus (a_v,b_v,c_v,u_v,v) \neq \emptyset$.
As a result, $(a_v,u_v)$ is a minimal $2$-separator of $G$ with $\{b_v,c_v,v\}$ being a full component of $G \setminus (a_v,u_v)$.

Since $N(a_v) \cap N(u_v)$ is an $a_vu_v$-separator of $G \setminus (c_v,v)$ by the hypothesis, therefore, $N(a_v) \cap N(u_v) \setminus b_v \neq \emptyset$, for it has to contain a vertex from every component of $G \setminus (a_v,b_v,c_v,u_v,v)$.
For now, let $x \in N(a_v) \cap N(u_v) \setminus b_v$ be arbitrary.
Let us prove that $tb(G) = 1$ implies that $tb(G') = 1$ where $G'$ is obtained by adding an edge between $b_v$ and $x$ (for now, $G'$ may not be planar, depending on the choice for $x$).
To prove this, let $(T,{\cal X})$ be a star-decomposition of $G$, that exists by Lemma~\ref{lem:dominator-in-bag}, minimizing the distance in $T$ between the subtrees $T_{a_v}$ and $T_{u_v}$.
We claim that $T_{a_v} \cap T_{u_v} \neq \emptyset$.
By contradiction, if  $T_{a_v} \cap T_{u_v} = \emptyset$, then by Corollary~\ref{cor:strong-sep}, there are two bags $B_{a_v},B_{u_v}$ that are adjacent in $T$ and such that $a_v \in B_{a_v} \setminus B_{u_v}, u_v \in B_{u_v} \setminus B_{a_v}$ respectively dominate $B_{a_v},B_{u_v}$.
However, this implies by the properties of a tree decomposition that $B_{a_v} \cap B_{u_v} \subseteq N(a_v) \cap N(u_v)$ is an $a_vu_v$-separator of $G$.
Since the $a_vu_v$-path $(a_v,v,c_v,u_v)$ does not intersect $N(a_v) \cap N(u_v)$, that is clearly a contradiction, and so, $T_{a_v} \cap T_{u_v} \neq \emptyset$.  

Furthermore, since there is a full component of $G \setminus (a_v,u_v)$ in the subgraph $G \setminus (b_v,c_v,v)$, therefore, by Lemma~\ref{lem:sep-keep} the removal of vertices $b_v,c_v,v$ from all bags in ${\cal X}$ leads to a tree decomposition $(T,{\cal X}^-)$ of breadth one of $G \setminus (b_v,c_v,v)$. 
Let $(T',{\cal X}')$ be a reduced star-decomposition obtained from $(T,{\cal X}^-)$, thar exists by Lemma~\ref{lem:dominator-in-bag}.
Since the subtrees $T'_{a_v},T'_x,T'_{u_v}$ are pairwise intersecting (because $x \in N(a_v) \cap N(u_v)$ and $T_{a_v} \cap T_{u_v} \neq \emptyset$), therefore by the Helly property (Lemma~\ref{lem:helly}) $T'_{a_v} \cap T'_x \cap T'_{u_v} \neq \emptyset$.
Let $B \in T'_{a_v} \cap T'_x \cap T'_{u_v}$.
To obtain a star-decomposition of $G'$, it now suffices to make the bag $B$ adjacent to the new bag $N_{G'}[b_v] = \{a_v,b_v,c_v,u_v,v,x\}$.

\medskip
\noindent
The above result holds for any choice of vertex  $x \in (N_G(a_v) \cap N_G(u_v)) \setminus b_v$.
Let us finally prove that one such a vertex $x$ exists so that $G'$ is planar.
Indeed, since $N(a_v) \cap N(u_v)$ is an $a_vu_v$-separator of $G \setminus (c_v,v)$ by the hypothesis, therefore, $S := (N(a_v) \cap N(u_v)) \cup \{v\}$ is an $a_vu_v$-separator of $G$, and in particular it is a minimal $a_vu_v$-separator (because for every vertex $s \in S$, there is an $a_vu_v$-path that intersects $S$ only in $s$).
By Corollary~\ref{cor:make-sep-cyc}, it can be computed in linear-time a planar supergraph $G_S$ of $G$ so that $S$ induces a cycle of $G_S$. 
Then, let $N_{G_S}(b_v) \cap S = \{x,v\}$, by construction the graph $G'$ is planar for such a choice of vertex $x$.
\end{proof}

In Theorem~\ref{claim:connect-vertex-b}, we show conditions so that vertex $b_v$ can be made adjacent to some other vertex of $N_G(a_v) \cap N_G(u_v)$.
Lemma~\ref{claim:sep-b} completes the picture by proving that if it is the case that $N_G(a_v) \cap N_G(u_v) \cap N_G(b_v) \neq \emptyset$, then $|N_G(a_v) \cap N_G(u_v) \cap N_G(b_v)|=1$ and vertex $b_v$ has exactly five neighbours. 

\begin{lemma}\label{claim:sep-b}
Let $G$ be a prime planar graph, let $v$ be a leaf-vertex of Type 2, $\Pi_v = (a_v,b_v,c_v)$ be as in Definition~\ref{def:leafVertex}, and let $u_v \notin \Pi_v \cup \{v\}$ be such that $(b_v,u_v)$ is an edge-separator of $G \setminus v$. 

Suppose $u_v \in N_G(c_v) \setminus N_G(a_v)$ and there exists $x \in N_G(a_v) \cap N_G(u_v) \cap N_G(b_v)$.

Then, $N_G(b_v) = \{a_v,c_v,u_v,v,x\}$. 
\end{lemma}

\begin{figure}[h!]
	\centering
	\includegraphics[width=0.15\textwidth]{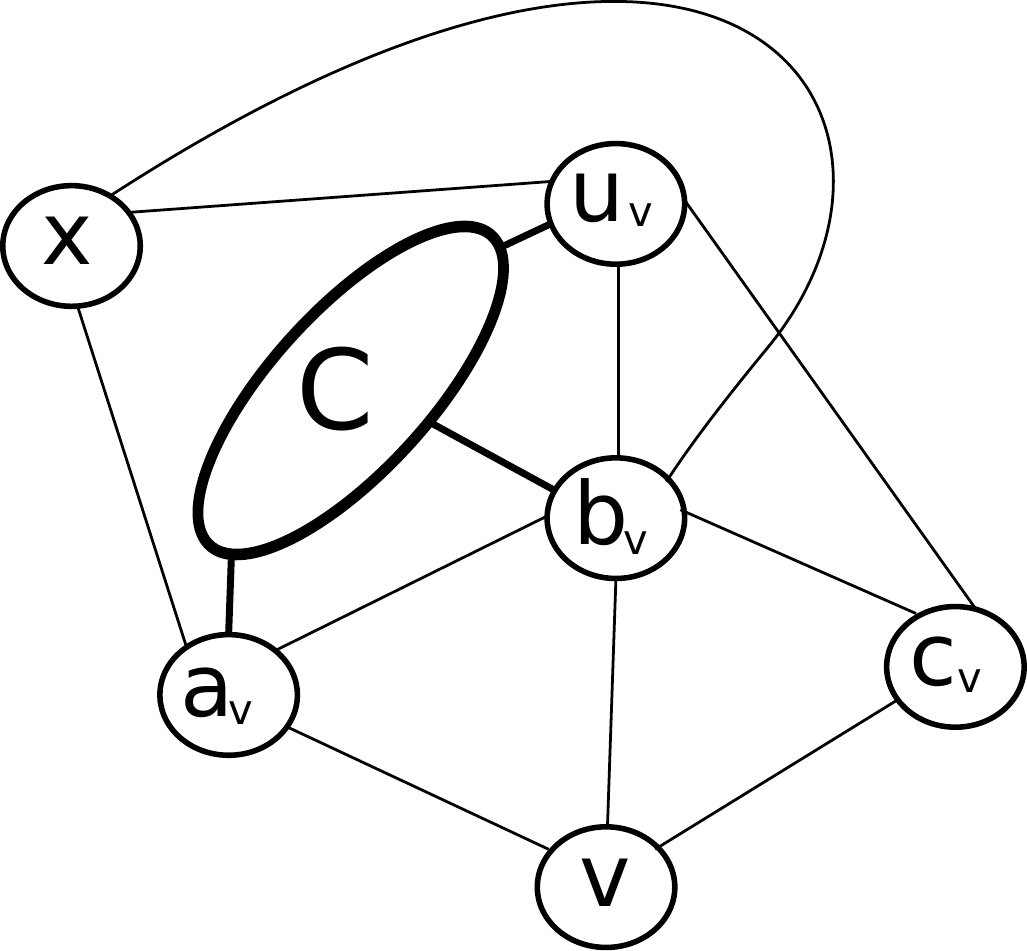}
	\caption{Case when $N_G(b_v) \neq \{a_v,c_v,u_v,v,x\}$.}
	\label{fig:single-neighbour-b}
\end{figure}

\begin{proof}
By contradiction, let $C$ be a component of $G \setminus (a_v,b_v,c_v,u_v,v,x)$ such that $b_v \in N(C)$ (see Figure~\ref{fig:single-neighbour-b} for an illustration).
By Theorem~\ref{th:GminusVprime?} $N_G(c_v) = \{b_v,u_v,v\}$, therefore, $c_v,v \notin N(C)$.
It follows that $N(C)$ is a separator of $G$.
In particular, $N(C) \subseteq \{a_v,b_v,u_v,x\}$, so, $a_v,u_v \in N(C)$ or else $N(C)$ should be a clique-separator of the prime graph $G$.
As a result, there is a $K_{3,3}$-minor with $\{a_v,b_v,u_v\}$ and $\{C,x,\{c_v,v\}\}$ being the two sides of the bipartition.
It contradicts the fact that $G$ is planar by the hypothesis.
\end{proof}

\begin{theorem}\label{lem:clique-case-2}
Let $G$ be a prime planar graph, let $v$ be a leaf-vertex of Type 2, $\Pi_v = (a_v,b_v,c_v)$ be as in Definition~\ref{def:leafVertex}, and let $u_v \notin \Pi_v \cup \{v\}$ be such that $(b_v,u_v)$ is an edge-separator of $G \setminus v$. 

Suppose $u_v \in N(c_v) \setminus N(a_v)$, $N(a_v) \cap N(u_v)$ is an $a_vu_v$-separator of $G \setminus (c_v,v)$, and either $N(b_v) = \{a_v,c_v,u_v,v\}$ or $N(b_v) \cap N(a_v) \cap N(u_v) \neq \emptyset$.

Then, there is $x \in N(a_v) \cap N(u_v)$ such that one of the following must hold:
\begin{itemize}
\item $V(G)= \{a_v,b_v,c_v,u_v,v,x\}$, and $G$ admits a star-decomposition with two bags $N_G[b_v],N_G[x]$;
\item or $\Pi'=(a_v,x,u_v) \in {\cal P}_3(G)$, and let $G'$ be obtained from $G$ by adding the edge $\{b_v,x\}$ (if it is not already present) then contracting this edge.
The graph $G'$ is planar and prime, furthermore $tb(G) = 1$ if and only if $tb(G') = 1$.
\end{itemize}
Moreover, vertex $x$ can be computed in linear-time.
\end{theorem}

\begin{proof}
There are two cases.
If $N_G(b_v) = \{a_v,c_v,u_v,v\}$, then let $x$ be set as in the statement of Theorem~\ref{claim:connect-vertex-b}.
Else, let $x$ be the unique vertex of $N(b_v) \cap N(a_v) \cap N(u_v)$, that is well-defined by Lemma~\ref{claim:sep-b}.
Note that in both cases, vertex $x$ can be computed in linear-time.
In addition, $N(b_v) \subseteq \{a_v,c_v,u_v,v,x\}$ (the latter property following from Lemma~\ref{claim:sep-b} when $b_v$ and $x$ are adjacent, and being trivial else).
Suppose for the proof that $V(G) \neq \{a_v,b_v,c_v,u_v,v,x\}$ (else, Theorem~\ref{lem:clique-case-2} is trivial).
We claim that $\{b_v,c_v,v\}$ is a component of $G \setminus \Pi'$.
Indeed, $N(b_v) \subseteq \Pi' \cup \{c_v,v\}$ by the hypothesis, and by Theorem~\ref{th:GminusVprime?} $N_G(c_v) = \{b_v,u_v,v\}$.
Since $V(G) \neq \{a_v,b_v,c_v,u_v,v,x\}$, then it indeed follows that $\Pi' \in {\cal P}_3(G)$, with $\{b_v,c_v,v\}$ being a component of $G \setminus \Pi'$.

Let us prove that $G'$ is prime and planar.
By Theorem~\ref{claim:connect-vertex-b}, adding an edge between $b_v$ and $x$ if it is not already present does not violate the property for the graph $G$ to be planar.
Therefore, $G'$ is planar because it is obtained by an edge-contraction from some planar graph.
To prove that $G'$ is prime, by contradiction suppose the existence of a minimal clique-separator $S'$ of $G'$.

Let us denote by $x'$ the vertex resulting from the contraction of the edge $\{b_v,x\}$.
Let $S := S'$ if $x' \notin S'$, $S := (S' \setminus x') \cup \{b_v,x\}$ else.
By construction, $S$ is a separator of $G$.
In particular, $S$ is not a clique because $G$ is prime by the hypothesis.
Therefore, $S \neq S'$, whence $x' \in S'$ or equivalently, $x,b_v \in S$.
We now claim that $c_v \in S \cap S'$ or $v \in S \cap S'$ (possibly, $v,c_v \in S \cap S')$.
There are two cases.

\begin{itemize}
\item Suppose that $S \setminus b_v$ is a separator of $G$.
Then, $S \setminus b_v$ is not a clique because $G$ is prime by the hypothesis.
Since $S \setminus (b_v,x) = S' \setminus x'$ is a clique, there must be some vertex of $S \setminus (b_v,x) = S \cap S'$ that is adjacent to $x'$ in $G'$ but non-adjacent to $x$ in $G$.
Consequently, $v \in S \cap S'$ or $c_v \in S \cap S'$.
\item Else, $S \setminus b_v$ is not a separator of $G$.
Recall that by construction, $S$ is a separator of $G$.
In particular, there must be two neighbours of $b_v$ in $G$ that are separated by $S$ in $G$.
Since $N_G(b_v) \setminus x$ induces the path $(a_v,v,c_v,u_v)$, it follows that $S$ must contain an internal node of the path, whence $c_v \in S \cap S'$ or $v \in S \cap S'$.
\end{itemize}
However, in such case $S'$ must be contained in one of $(a_v,x',v), \ (v,x',c_v)$ or $(c_v,x',u_v)$, for it is a clique of $G'$.
In particular, let $z \in \{a_v,u_v\} \setminus S'$.
Since $z$ has a neighbour in every component $C'$ of $G \setminus \Pi'$, $\{z\} \cup C'$ is not disconnected by $S'$ in $G'$. 
Furthermore, let us contract $C'$ to $z$ so as to make $a_v$ and $u_v$ adjacent, $S'$ intersects the resulting cycle $(a_v,u_v,c_v,v)$ either in an edge (different from $\{a_v,u_v\}$) or a single vertex because it is a clique of $G'$, therefore, $(a_v,u_v,c_v,v) \setminus S'$ is not disconnected by $S'$.
Altogether, this contradicts the fact that $S'$ is a separator of $G'$, and so, $G'$ is prime.

\medskip
\noindent
Finally, let us prove that $tb(G') = 1$ if and only if $tb(G) = 1$.
If $tb(G) = 1$, then let us assume $b_v$ and $x$ to be adjacent (if they are not, then Theorem~\ref{claim:connect-vertex-b} ensures we can add the edge without violating the property for the graph to be of tree-breadth one).
Then, $tb(G') = 1$ because it is obtained by an edge-contraction from some graph with tree-breadth one and that tree-breadth is contraction-closed by Lemma~\ref{lem:contraction-closed}.

Conversely, let us prove that $tb(G') = 1$ implies that $tb(G) = 1$.
To prove this, let $(T,{\cal X})$ be a star-decomposition of $G'$, that exists by Lemma~\ref{lem:dominator-in-bag}, minimizing the distance in $T$ between the subtrees $T_{a_v}$ and $T_{u_v}$.
We claim that $T_{a_v} \cap T_{u_v} \neq \emptyset$.
By contradiction, suppose $T_{a_v} \cap T_{u_v} = \emptyset$.
Recall that $(a_v,x',u_v) \in {\cal P}_3(G')$ (because $\Pi' \in {\cal P}_3(G)$)  and $G'$ is prime, therefore one of $(a_v,x',u_v)$ or $(a_v,u_v)$ is a minimal separator of $G'$.
Since we assume the distance in $T$ between $T_{a_v}$ and $T_{u_v}$ to be minimized, by Corollary~\ref{cor:strong-sep}, there are two bags $B_{a_v},B_{u_v}$ that are adjacent in $T$ so that $a_v \in B_{a_v} \setminus B_{u_v}, u_v \in B_{u_v} \setminus B_{a_v}$ respectively dominate $B_{a_v},B_{u_v}$.
However, by the properties of a tree decomposition $B_{a_v} \cap B_{u_v} \subseteq N(a_v) \cap N(u_v)$ is an $a_vu_v$-separator of $G'$, that is impossible due to the existence of the path $(a_v,v,c_v,u_v)$ in $G'$ that does not intersect $N(a_v) \cap N(u_v)$.
Therefore, $T_{a_v} \cap T_{u_v} \neq \emptyset$. 
Hence the subtrees $T_{a_v},T_{x'},T_{u_v}$ are pairwise intersecting and so, by the Helly Property (Lemma~\ref{lem:helly}), $T_{a_v} \cap T_{x'} \cap T_{u_v} \neq \emptyset$.
Furthermore, $N_{G'}[c_v] \cup N_{G'}[v] \subseteq N_{G'}[x']$ by construction. 
So, let us construct a tree decomposition of $G$ of breadth one as follows.
First, let us remove $c_v$ and $v$ from all bags in ${\cal X}$.
Since $N_{G'}[c_v] \cup N_{G'}[v] \subseteq N_{G'}[x']$, one obtains a tree decomposition of $G' \setminus (c_v,v)$ of breadth one.
Then let us replace $x'$ with $x$ in all bags.
Note that in so doing, one obtains a tree decomposition of $G \setminus (b_v,c_v,v)$ of breadth one.
Finally, let us make adjacent the new bag $N_G[b_v]$ with any bag of $T_{a_v} \cap T_{x} \cap T_{u_v}$.
The result is indeed a tree decomposition of $G'$ because $N_G[b_v] \subseteq \{a_v,b_v,c_v,u_v,v,x\}$ and $N_G[c_v] \cup N_G[v] \subseteq N_G[b_v]$.  
\end{proof}

\subsubsection{Proof of~\ref{step:new-clique-sep}~\ref{step:separation}~\ref{step:no-diamond-case}}

\begin{theorem}\label{lem:final-case-1}
Let $G$ be a prime planar graph, let $v$ be a leaf-vertex of Type 2, $\Pi_v = (a_v,b_v,c_v)$ be as in Definition~\ref{def:leafVertex}, and let $u_v \notin \Pi_v \cup \{v\}$ be such that $(b_v,u_v)$ is an edge-separator of $G \setminus v$. 

Suppose $u_v \in N(c_v) \setminus N(a_v)$, $N(a_v) \cap N(u_v)$ is an $a_vu_v$-separator in the subgraph $G \setminus (c_v,v)$, $N(b_v) \neq \{a_v,c_v,u_v,v\}$ and $N(a_v) \cap N(b_v) \cap N(u_v) = \emptyset$.

Then it can be computed in linear-time (a unique) $x \in N(a_v) \cap N(u_v)$ such that if $tb(G) = 1$, $N(b_v) \cap N(x)$ is a $b_vx$-separator and $|N(b_v) \cap N(x)| \geq 3$.
\end{theorem}

\begin{proof}
Let $W = (N(a_v) \cap N(u_v)) \cup \{a_v,c_v,u_v,v\}$.
By the hypothesis, $N(b_v) \neq \{a_v,c_v,u_v,v\}$ and $N(a_v) \cap N(b_v) \cap N(u_v) = \emptyset$, therefore, it exists a component $C_0$ of $G \setminus W$ such that $b_v \in N(C_0)$.
We claim that there is $x \in N(a_v) \cap N(u_v) \cap N(C_0)$ satisfying that $N(C_0) \subseteq (a_v,b_v,x)$ or $N(C_0) \subseteq (u_v,b_v,x)$.
Indeed, first observe that $v,c_v \notin N(C_0)$ because by Theorem~\ref{th:GminusVprime?} $N_G(c_v) = \{b_v,u_v,v\}$. 
Furthermore, $a_v \notin N(C_0)$ or $u_v \notin N(C_0)$ because $N(a_v) \cap N(u_v)$ is an $a_vu_v$-separator of $G \setminus (c_v,v)$ by the hypothesis.
So, let $\{z,z'\} = \{a_v,u_v\}$ satisfy $z' \notin N(C_0)$.
Since  $G$ is prime by the hypothesis, and so, biconnected, $G \setminus z$ is connected.
Furthermore, $N(C_0) \setminus z \subseteq N(z')$ (by the definition of $W$), therefore, $N(C_0) \setminus z$ is a minimal separator of $G \setminus z$.
By Lemma~\ref{lem:P3exists} there exist $s,t \in N(a_v) \cap N(u_v)$ non-adjacent such that $N(C_0) \setminus z = \{s,t\}$.
Since $b_v \in N(C_0)$ by construction, therefore, let us set $\{s,t\} = \{b_v,x\}$, that finally proves the claim.

We claim in addition that $x$ does not depend on the choice of the component $C_0$.
By contradiction, let $C,C'$ be two components of $G \setminus W$ such that $b_v \in N(C) \cap N(C')$ and let $x,x' \in N(a_v) \cap N_G(u_v)$ be distinct and such that $x \in N(C), \ x' \in N(C')$.
Then, there exists a $K_{3,3}$-minor with $\{a_v,b_v,u_v\}$ and $\{ \{c_v,v\}, C \cup \{x\}, C' \cup \{x'\} \}$ being the sides of the bipartition, that contradicts the hypothesis that $G$ is planar.
Thus from now on, let $x \in N(a_v) \cap N(u_v) \setminus b_v$ be the unique vertex satisfying that for every component $C$ of $G \setminus W$, if $b_v \in N(C)$ then $x \in N(C)$.

\begin{figure}[h!]
	\centering
	\includegraphics[width=0.25\textwidth]{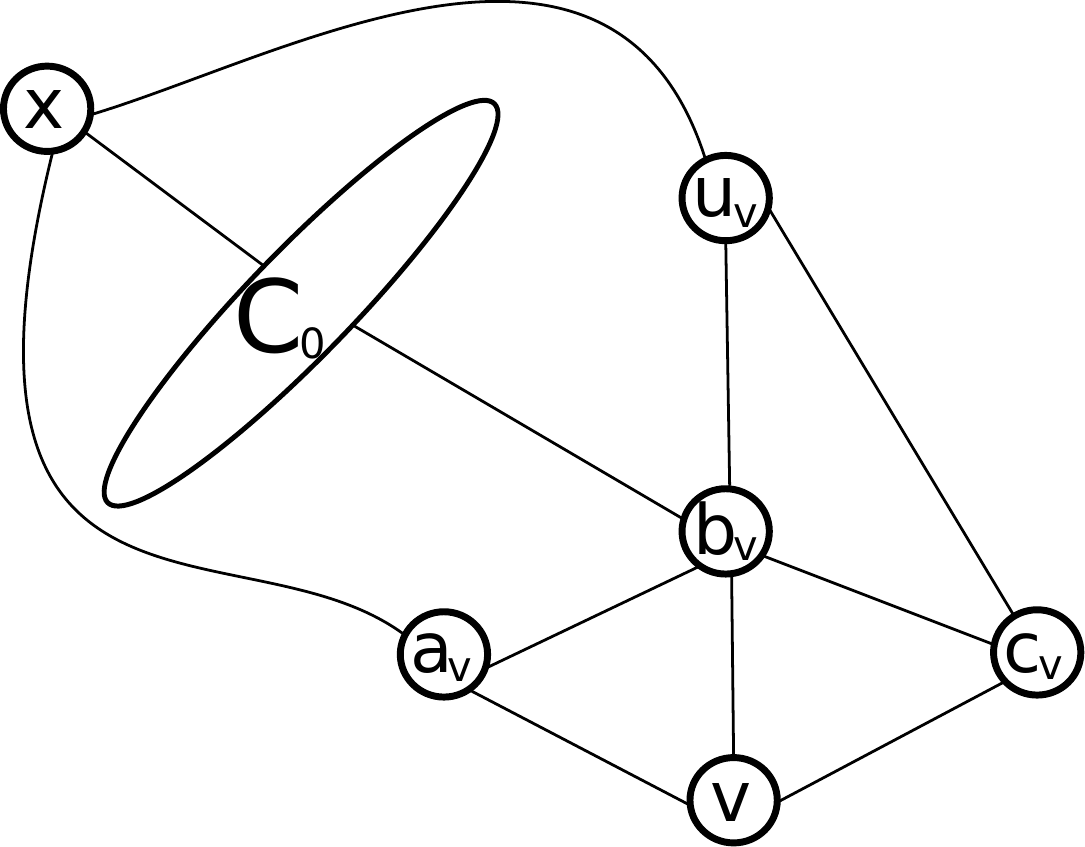}
	\caption{Component $C_0$ such that $b_v,x \in N(C_0)$.}
	\label{fig:add-chord-to-quadrangle}
\end{figure}

\smallskip
Recall that $C_0$ is a fixed component of $G \setminus W$ such that $b_v,x \in N(C_0)$ (see Figure~\ref{fig:add-chord-to-quadrangle} for an illustration).
Finally, assume for the remaining of the proof that $tb(G) = 1$ and let us prove that $N(b_v) \cap N(x)$ is a $b_vx$-separator and $|N(b_v) \cap N(x)| \geq 3$.
To prove it, we will only need to prove that $N(b_v) \cap N(x)$ is a $b_vx$-separator of $G$.
Indeed, in such a case $N(b_v) \cap N(x) \cap C_0 \neq \emptyset$, and so, $|N(b_v) \cap N(x)| \geq 3$ because $a_v,u_v \in N(b_v) \cap N(x)$ and $a_v,u_v \notin C_0$. 
%Observe that one of $(b_v,x), (b_v,a_v,x)$ or $(b_v,u_v,x)$ is a minimal separator $S$ of $G$ by construction.

Let $(T,{\cal X})$ be star-decomposition of $G$, that exists by Lemma~\ref{lem:dominator-in-bag}, minimizing the distance in $T$ between the subtrees $T_{b_v}$ and $T_x$.
%Assume furthermore that $(T,{\cal X})$ minimizes the distance in $T$ between $T_{a_v}$ and $T_{u_v}$ w.r.t. this property.
We claim that $T_{b_v} \cap T_x = \emptyset$.
By contradiction, suppose $T_{b_v} \cap T_x \neq \emptyset$.
Let us prove as an intermediate subclaim that $T_{a_v} \cap T_{u_v} \neq \emptyset$. 
By contradiction, let $T_{a_v} \cap T_{u_v} = \emptyset$. 
By the properties of a tree decomposition, every bag $B$ onto the path in $T$ between $T_{a_v}$ and $T_{u_v}$ must contain $N(a_v) \cap N(u_v)$ and at least one of $v$ or $c_v$.
If $c_v \in B$ then $B \subseteq N[u_v]$ since $N(c_v) = \{b_v,u_v,v\}$ and $x \in B$.
Similarly if $v \in B$ then $B \subseteq N[a_v]$ since $N(v) = \{a_v,b_v,c_v\}$ and $x \in B$.
Consequently, there are two adjacent bags $B_{a_v},B_{u_v}$ such that $a_v \in B_{a_v} \setminus B_{u_v}$ and $u_v \in B_{u_v} \setminus B_{a_v}$ respectively dominate $B_{a_v}$ and $B_{u_v}$.
However, by the properties of a tree decomposition, $B_{a_v} \cap B_{u_v} = N(a_v) \cap N(u_v)$ is an $a_vu_v$-separator of $G$, thus contradicting the existence of the path $(a_v,v,c_v,u_v)$ in $G$.
Therefore, it follows that $T_{a_v} \cap T_{u_v} \neq \emptyset$, that proves the subclaim.

If $T_{a_v} \cap T_{u_v} \neq \emptyset$ and $T_{b_v} \cap T_x \neq \emptyset$ then the subtrees $T_{a_v},T_{b_v},T_{u_v},T_x$ are pairwise intersecting and so, it implies $T_{a_v} \cap T_{u_v} \cap T_{b_v} \cap T_x \neq \emptyset$ by the Helly property (Lemma~\ref{lem:helly}).
However, let $B' \in T_{a_v} \cap T_{u_v} \cap T_{b_v} \cap T_x$, no vertex in $G$ can dominate $B'$ because $N(b_v) \cap N(a_v) \cap N(u_v) = \emptyset$ by the hypothesis, thus contradicting the fact that $(T,{\cal X})$ is a star-decomposition.
As a result, we proved the claim that $T_{b_v} \cap T_x = \emptyset$.

Finally, since there exists $S \subseteq (N(b_v) \cap N(x)) \cup \{b_v,x\}$ a minimal separator of $G$ such that $b_v,x \in S$ (namely, $S := N(C_0)$), and $(T,{\cal X})$ is assumed to minimize the distance in $T$ between $T_{b_v}$ and $T_x$, by Corollary~\ref{cor:strong-sep} there exist two adjacent bags $B_{b_v},B_x$ such that $b_v \in B_{b_v} \setminus B_x, \ x \in B_x \setminus B_{b_v}$ respectively dominate $B_{b_v}$ and $B_x$.
By the properties of a tree decomposition, $B_{b_v} \cap B_x = N(b_v) \cap N(x)$ is indeed a $b_vx$-separator of $G$. 
\end{proof}

\begin{lemma}\label{claim:final-case-2}
Let $G$ be a prime planar graph, let $v$ be a leaf-vertex of Type 2, $\Pi_v = (a_v,b_v,c_v)$ be as in Definition~\ref{def:leafVertex}, and let $u_v \notin \Pi_v \cup \{v\}$ be such that $(b_v,u_v)$ is an edge-separator of $G \setminus v$. 

Suppose $u_v \in N(c_v) \setminus N(a_v)$, $N(a_v) \cap N(u_v)$ is an $a_vu_v$-separator in the subgraph $G \setminus (c_v,v)$, $N(b_v) \neq \{a_v,c_v,u_v,v\}$ and $N(a_v) \cap N(b_v) \cap N(u_v) = \emptyset$.
Assume furthermore that there is $x \in N(a_v) \cap N(u_v)$, and there exists a leaf-vertex $l \in N(b_v) \cap N(x)$.

Then, $l$ is a leaf-vertex of Type 1, or $l$ is a leaf-vertex of Type 2 or 3 and $G \setminus l$ is prime.
\end{lemma}

\begin{figure}[h!]
 \centering
 \includegraphics[width=0.65\textwidth]{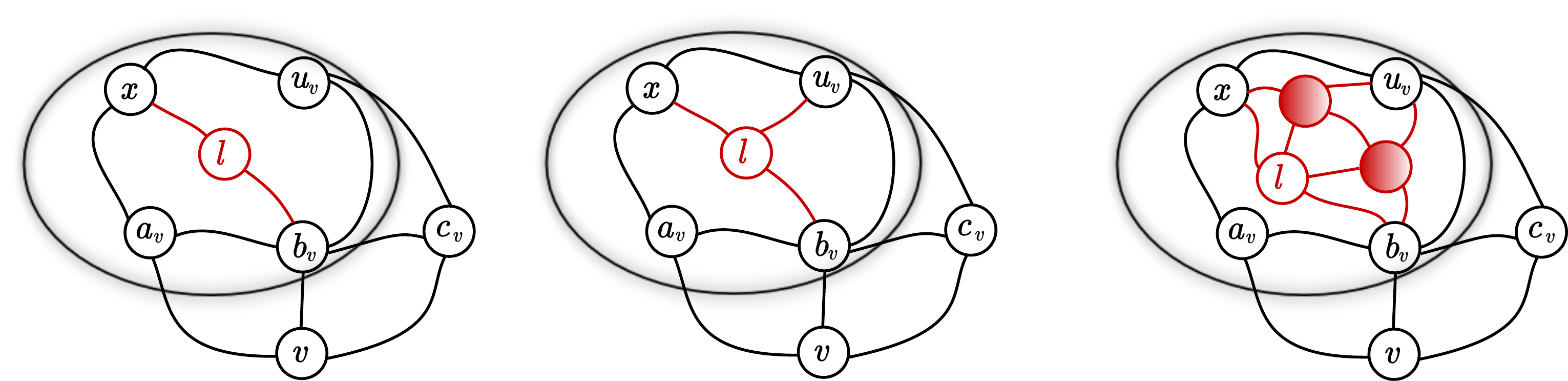}
 \caption{Existence of a leaf-vertex in $N(b_v) \cap N(x)$.}
 \label{fig:planar-clique-new-leaf}
\end{figure}

\begin{proof}
Suppose for the proof that $l$ is not of Type 1 (else, Lemma~\ref{claim:final-case-2} is trivial).
Then, $l$ is of Type 2 or 3, let $\Pi_l$ be as in Definition~\ref{def:leafVertex}.
Note that $l \neq a_v$ because $v,b_v,x \in N(a_v)$ do not induce a path, similarly $l \neq u_v$ because $b_v,c_v,x \in N(u_v)$ do not induce a path.
Furthermore by the hypothesis, $b_v$ and $x$ are the two endpoints of $\Pi_l$.
Suppose by way of contradiction that there is a minimal clique-separator $S$ of $G \setminus l$.
Since $G$ is prime by the hypothesis, by Lemma~\ref{lem:separatorInCompo} $S$ is a $b_vx$-separator of $G \setminus l$.
However, it implies that $a_v,u_v \in S$, that contradicts the fact that $S$ is a clique.
As a result, $G \setminus l$ is prime.   
\end{proof}

Equipped with Lemma~\ref{claim:final-case-2}, we can assume from now on that there is no leaf-vertex that is adjacent to both vertices $b_v,x$, or else it could be immediately processed by the algorithm.

\begin{theorem}\label{lem:final-case-3}
Let $G$ be a prime planar graph, let $v$ be a leaf-vertex of Type 2, $\Pi_v = (a_v,b_v,c_v)$ be as in Definition~\ref{def:leafVertex}, and let $u_v \notin \Pi_v \cup \{v\}$ be such that $(b_v,u_v)$ is an edge-separator of $G \setminus v$. 

Suppose $u_v \in N(c_v) \setminus N(a_v)$, $N(a_v) \cap N(u_v)$ is an $a_vu_v$-separator in the subgraph $G \setminus (c_v,v)$, $N(b_v) \neq \{a_v,c_v,u_v,v\}$ and $N(a_v) \cap N(b_v) \cap N(u_v) = \emptyset$.
Assume furthermore that there is $x \in N(a_v) \cap N(u_v)$ such that $N(b_v) \cap N(x)$ is a $b_vx$-separator, $|N(b_v) \cap N(x)| \geq 3$, and there is no leaf-vertex in $N(b_v) \cap N(x)$.

Then, there exist $y,z \in N(b_v) \cap N(x)$ non-adjacent such that the graph $G'$, obtained from $G$ by making $y,z$ adjacent, is planar and prime, and it holds $tb(G) = 1$ if and only if $tb(G') = 1$.
Furthermore, the pair $y,z$ can be computed in linear-time.
\end{theorem}

\begin{figure}[h!]
 \centering
 \includegraphics[width=0.25\textwidth]{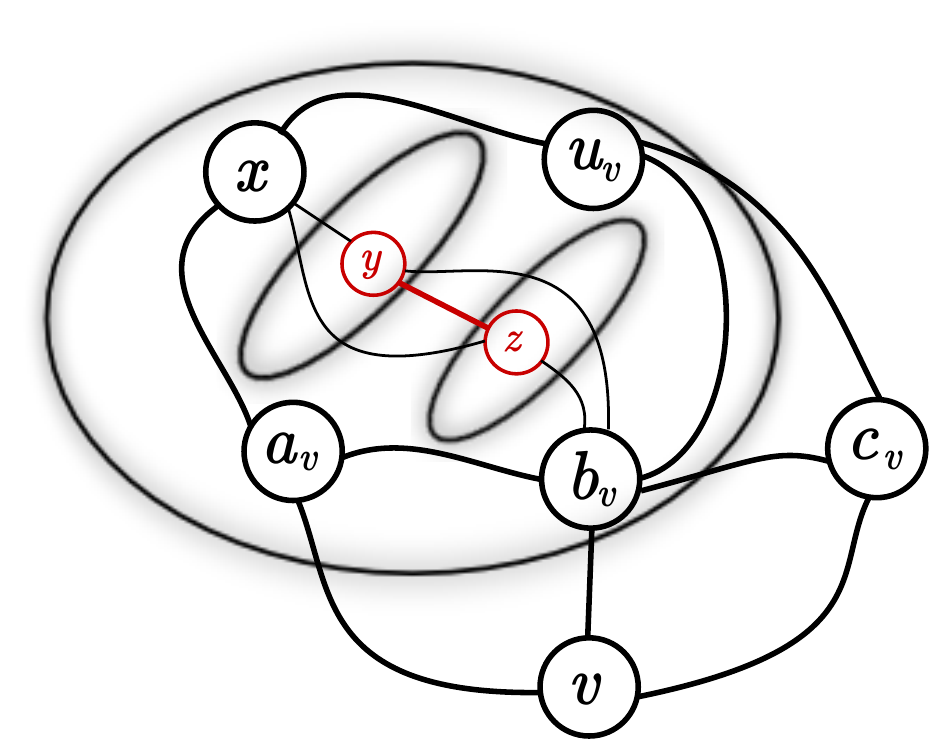}
 \caption{Illustration of Theorem~\ref{lem:final-case-3}.}
 \label{fig:planar-clique-add-edge}
\end{figure}

\begin{proof}
Let us first show how to find the pair $y,z$.
Let $W = \{a_v,c_v,v,u_v\} \cup (N(a_v) \cap N(u_v))$. 
Choose any component $C_0$ of $G \setminus W$ such that $b_v,x \in N(C_0)$ and $N(C_0) \subseteq (a_v,b_v,x)$ or $N(C_0) \subseteq (u_v,b_v,x)$ (the existence of such a component has been proved in Theorem~\ref{lem:final-case-1}).
Note that $N(b_v) \cap N(x) \cap C_0 \neq \emptyset$ since $N(b_v) \cap N(x)$ is a $b_vx$-separator of $G$ by the hypothesis.
%As $|N(b_v) \cap N(x)| \geq 3$ by the hypothesis, $N(b_v) \cap N(x) \cap N(a_v) \cap N(u_v) = \emptyset$ because $N(b_v) \cap N(a_v) \cap N(u_v) = \emptyset$ by the hypothesis, and $N(b_v) \cap N(x) \cap \{a_v,c_v,u_v,v\} = \{a_v,u_v\}$ because, by Theorem~\ref{th:GminusVprime?}, $N(c_v) = \{b_v,u_v,v\}$, therefore, there exists a component $C_0$ of $G \setminus W$ such that $N(b_v) \cap N(x) \cap C_0 \neq \emptyset$.
%We will prove as an intermediate claim that 
%Indeed, $c_v,v \notin N(C_0)$ because $N(v) = \{a_v,b_v,c_v\}$ by the hypothesis, and by Theorem~\ref{th:GminusVprime?} $N(c_v) = \{b_v,u_v,v\}$.
%Furthermore, $a_v \notin N(C_0)$ or $c_v \notin N(C_0)$, or else $N(a_v) \cap N(c_v) \subseteq W$ could not be an $a_vu_v$-separator of $G \setminus (v,c_v)$, thus contradicting the hypothesis.
%Therefore, the claim is proved, that is, $N(C_0) \subseteq (a_v,b_v,x)$ or $N(C_0) \subseteq (u_v,b_v,x)$.
Then, let $S := N(b_v) \cap N(x)$. 
By the hypothesis $S$ is a minimal separator of $G$ and $|S| \geq 3$, therefore, by Corollary~\ref{cor:make-sep-cyc} there is a planar supergraph $G_S$ of $G$ so that $S$ induces a cycle of $G_S$.
Furthermore, $G_S$ can be computed in linear-time.
Let $P$ be an $a_vu_v$-path of the cycle $G_S[S]$ that intersects $C_0$.
Since by the above claim $a_v \notin N_G(C_0)$ or $u_v \notin N_G(C_0)$, therefore, there is $y \in C_0 \cap V(P)$, there is $z$ adjacent to vertex $y$ in $P$ so that either $z \in C_1$ for some component $C_1$ of $G \setminus (W \cup C_0)$ or $z \in \{a_v,u_v\} \setminus N_G(C_0)$. 
In particular, $z \notin N_G[C_0] = C_0 \cup N_G(C_0)$.
Moreover, the graph $G'$, obtained from $G$ by adding an edge between $y$ and $z$, is planar by construction.

\begin{claim}
\label{claim:gprime-prime-2}
$G'$ is prime.
\end{claim}

\begin{proofclaim}
By contradiction, let $X$ be a minimal clique-separator of $G'$.
Since $G'$ is a supergraph of $G$, $X$ is a separator of $G$.
As a result, $y,z \in X$ because $G$ is prime by the hypothesis.
%Furthermore, $z \notin N[C_0]$ by construction.
Let us prove as an intermediate step that $N(y) \cap N(z) = \{b_v,x\}$.
There are two cases.
If $z \in \{a_v,u_v\}$, then let $\{z,z'\} = \{a_v,u_v\}$.
Since $N(C_0) \subseteq (z',b_v,x)$ (because $z \notin N_G(C_0)$) and $z,z'$ are non-adjacent by the hypothesis, therefore, the claim immediately follows in this case.
Else, $z \notin \{a_v,u_v\}$.
Let $C_1$ be the component of $G \setminus (W \cup C_0)$ containing $z$.
In such case, $b_v,x \in N(C_1)$ because $z \in S$ by construction.
Therefore, $N(C_1) \subseteq (a_v,b_v,x)$ or $N(C_1) \subseteq (u_v,b_v,x)$ because the respective roles of components $C_0,C_1$ are symmetrical in this case.
Suppose by way of contradiction $N(C_0) = N(C_1) = \{s,b_v,x\}$ for some $s \in \{a_v,u_v\}$ and let $\{s,t\} = \{a_v,u_v\}$. 
Then, there is a $K_{3,3}$-minor of $G$ with $\{b_v,x,s\}$ and $\{C_0,C_1,\{c_v,v,t\}\}$ being the sides of the bipartition, that contradicts the hypothesis that $G$ is planar.
As a result, $N(C_0) \cap N(C_1) = \{b_v,x\}$, that finally proves the claim.

Since $X$ is assumed to be a clique of $G'$ and $y,z \in X$, it follows $X \subseteq \{b_v,y,z\}$ or $X \subseteq \{x,y,z\}$.
Consequently, $G[W \setminus X]$ is connected because $a_v,u_v \notin X$ is a dominating pair of $W$, $b_v,x \in N(a_v) \cap N(u_v)$ and $b_v \notin X$ or $x \notin X$. 
However, since $y,z \in X$ and $N_G(y) \subseteq W \cup C_0$, then in such case there must be a component $A$ of $G \setminus X$ so that $A \subset C_0$.  
Since $z \notin N_G[C_0]$ by construction, $N_G(A) \subseteq X \setminus z$ is a clique-separator of $G$, thus contradicting the hypothesis that $G$ is prime.
As a result, $G'$ is prime.
\end{proofclaim}

Now, let us prove $tb(G) = 1$ if and only if $tb(G') = 1$.

If $tb(G) = 1$, then let $(T,{\cal X})$ be star-decomposition of $G$, that exists by Lemma~\ref{lem:dominator-in-bag}, minimizing the distance in $T$ between $T_{a_v}$ and $T_{u_v}$.
Let us prove that $(T,{\cal X})$ is a star-decomposition of $G'$, whence $tb(G') = 1$.
To prove it, it is sufficient to prove $T_y \cap T_z \neq \emptyset$.
We will prove as an intermediate claim that $T_{a_v} \cap T_{u_v} \neq \emptyset$.
By contradiction, assume $T_{a_v} \cap T_{u_v} = \emptyset$.
Observe that $\Pi' = (a_v,b_v,u_v) \in {\cal P}_3(G)$ with $\{c_v,v\}$ being a full component of $G \setminus \Pi'$.
Therefore, since $G$ is prime by the hypothesis, one of $\Pi'$ or $\Pi' \setminus b_v$ is a minimal separator of $G$.
Since $(T,{\cal X})$ is assumed to minimize the distance in $T$ between $T_{a_v}$ and $T_{u_v}$, therefore, by Corollary~\ref{cor:strong-sep} there are two adjacent bags $B_{a_v},B_{u_v}$ such that $a_v \in B_{a_v} \setminus B_{u_v}$ and $u_v \in B_{u_v} \setminus B_{a_v}$ respectively dominate $B_{a_v}$ and $B_{u_v}$.
However, by the properties of a tree decomposition $B_{a_v} \cap B_{u_v} = N(a_v) \cap N(u_v)$ is an $a_vu_v$-separator of $G$, that contradicts the existence of the path $(a_v,v,c_v,u_v)$ in $G$.
Consequently the claim is proved, hence $T_{a_v} \cap T_{u_v} \neq \emptyset$. 
The latter claim implies $T_{b_v} \cap T_{x} = \emptyset$, for if $T_{b_v} \cap T_x \neq \emptyset$ then the subtrees $T_{a_v},T_{b_v},T_{u_v},T_x$ are pairwise intersecting, hence $T_{a_v} \cap T_{b_v} \cap T_x \cap T_{u_v} \neq \emptyset$ by the Helly property (Lemma~\ref{lem:helly}), that would contradict the fact that $(T,{\cal X})$ is a star-decomposition because $N(a_v) \cap N(u_v) \cap N(b_v) = \emptyset$ by the hypothesis.
Finally, since $(x,y,b_v,z)$ induces a cycle of $G$ and $T_{b_v} \cap T_x = \emptyset$, therefore, by the properties of a tree decomposition $T_y \cap T_z \neq \emptyset$, and so, $(T,{\cal X})$ is indeed a star-decomposition of $G'$.

\smallskip
Conversely, let us prove $tb(G') = 1$ implies $tb(G) = 1$.
To prove it, let $(T',{\cal X}')$ be a star-decomposition of $G'$, that exists by Lemma~\ref{lem:dominator-in-bag}, minimizing the number $|{\cal X}'|$ of bags.
Assume furthermore $(T',{\cal X}')$ to minimize the number of bags $B \in {\cal X}'$ that are not contained into the closed neighbourhood of some vertex in $G$ w.r.t. the minimality of $|{\cal X}'|$. 
In order to prove $tb(G) = 1$, it suffices to prove that $(T',{\cal X}')$ is a star-decomposition of $G$.
We will start proving intermediate claims.

\begin{claim}
\label{claim:neighbours-y}
$a_v,u_v \notin N_G(y)$.
\end{claim}

\begin{proofclaim}
By contradiction, assume the existence of $z' \in \{a_v,u_v\}$ so that $z'$ and $y$ are adjacent in $G$.
In particular, $z' \neq z$ (since $z \notin N_G[C_0]$) and $N_G(C_0) = \{b_v,x,z'\}$ since either $N_G(C_0) \subseteq \{b_v,x,a_v\}$ or $N_G(C_0) \subseteq \{b_v,x,u_v\}$. 
Hence, the path $(b_v,z',x)$ is a separator of $G$.
Since $y \in N_G(b_v) \cap N_G(z') \cap N_G(x)$, by Lemma~\ref{lem:common-neighbour} either $C_0$ is reduced to $y$ or $(b_v,y,x) \in  {\cal P}_3(G)$ separates $z'$ from $C_0 \setminus y$.
The case $C_0 \setminus y = \emptyset$ implies that $y$ is a leaf-vertex of Type 2, that contradicts the hypothesis that there is no leaf-vertex in $N_G(b_v) \cap N_G(x)$.
Therefore, let $(b_v,y,x) \in {\cal P}_3(G)$ separates $z'$ from $C_0 \setminus y$ in $G$, and let $C_0' \subseteq C_0 \setminus y$ be a component of $G \setminus (b_v,y,x)$ (such a component $C_0'$ exists because $N_G(C_0) = \{z',b_v,x\}$).
Since $G$ is prime, $b_v,x \in N_G(C_0')$ (indeed, neither $b_v$ nor $y$ nor $x$ nor $(b_v,y)$ nor $(y,x)$ can be a separator of $G$).
Therefore, $N_G(b_v) \cap N_G(x) \cap C_0' \neq \emptyset$ because $N_G(b_v) \cap N_G(x)$ is a $b_vx$-separator of $G$ by the hypothesis.
Furthermore, $y \in N_G(C_0')$ because $C_0$ is connected.
%In particular, let $y' \in N_G(b_v) \cap N_G(x) \cap C_0'$ and let $(y,Q,y')$ be an $yy'$-path of $G[C_0' \cup \{y\}]$.
%By construction, $V(Q) \subseteq C_0'$, so, $a_v,b_v,u_v,x \notin C_0'$.
However in such case, there is a $K_{3,3}$-minor of $G'$ with $\{b_v,x,y\}$ and $\{a_v,C_0',u_v\}$ being the sides of the bipartition, that contradicts the fact that $G'$ is planar.
\end{proofclaim}

\begin{claim}
\label{claim:no-dominator-for-cycle}
There is no vertex dominating the cycle $(a_v,b_v,u_v,x)$ in $G'$.
\end{claim}

\begin{proofclaim}
By contradiction, if it were the case that such a vertex exists, then, since $N_G(b_v) \cap N_G(a_v) \cap N_G(u_v) = \emptyset$ by the hypothesis, the dominator should be $y$ and furthermore $z \in \{a_v,u_v\}$.
In particular, $y \in N_G(b_v) \cap N_G(x) \cap N_G(z')$ with $\{z,z'\} = \{a_v,u_v\}$, thus contradicting Claim~\ref{claim:neighbours-y}. 
%this dominator should be one of $y$ or $z$.
%In particular, the dominator could be $z$ only if $z \notin \{a_v,u_v\}$ because $a_v$ and $u_v$ are non-adjacent in $G'$. 
%Therefore, assume by symmetry that $y$ dominates $(a_v,b_v,u_v,x)$ in $G'$ (for if $z \in C_1$ with $C_1$ being a component of $G \setminus (W \cup C_0)$, vertices $y,z$ play symmetrical roles).
%The latter assumption implies $a_v \in N_G(y)$ or $u_v \in N_G(y)$ because $G'$ and $G$ only differ in the edge $\{y,z\}$.
%Assume by symmetry $a_v \in N_G(y)$.
%Since $N_G(C_0) \subseteq (a_v,b_v,x)$ or $N_G(C_0) \subseteq (u_v,b_v,x)$, therefore, $N_G(C_0) = \{a_v,b_v,x\}$, hence $(b_v,a_v,x) \in {\cal P}_3(G)$.
%In addition, $y \in N_G(a_v) \cap N_G(b_v) \cap N_G(x) \cap C_0$ with $C_0$ being a component of $G \setminus (b_v,a_v,x)$, so
%As a result, $a_v \notin N_G(y)$, that finally proves that no vertex dominates the cycle $(a_v,b_v,u_v,x)$ in $G'$.
\end{proofclaim} 

\begin{claim}
\label{claim:crossing-subtrees}
$T'_{a_v} \cap T'_{u_v} \neq \emptyset$.
\end{claim}

\begin{proofclaim}
By contradiction, let $T'_{a_v} \cap T'_{u_v} = \emptyset$. 
By the properties of a tree decomposition, every bag $B$ onto the path in $T'$ between $T'_{a_v}$ and $T'_{u_v}$ (including the endpoints) must contain $N_{G'}(a_v) \cap N_{G'}(u_v)$ and at least one of $v$ or $c_v$.
Then, if $c_v \in B$ then $B \subseteq N_{G'}[u_v]$ and so, $B \in T'_{u_v}$, since $N_{G'}(c_v) = \{b_v,u_v,v\}$ and $x \in B$.
Similarly if $v \in B$ then $B \subseteq N_{G'}[a_v]$ and so, $B \in T'_{a_v}$, since $N_{G'}(v) = \{a_v,b_v,c_v\}$ and $x \in B$.
Consequently, there are two adjacent bags $B_{a_v},B_{u_v}$ such that $a_v \in B_{a_v} \setminus B_{u_v}$ and $u_v \in B_{u_v} \setminus B_{a_v}$ respectively dominate $B_{a_v}$ and $B_{u_v}$ in $G'$.
However, by the properties of a tree decomposition, $B_{a_v} \cap B_{u_v} = N_{G'}(a_v) \cap N_{G'}(u_v)$ is an $a_vu_v$-separator of $G'$, thus contradicting the existence of the path $(a_v,v,c_v,u_v)$ in $G'$.
Therefore, it follows that $T'_{a_v} \cap T'_{u_v} \neq \emptyset$, that proves the claim.
\end{proofclaim}

\begin{claim}
\label{claim:disjoint-subtrees}
$T'_{b_v} \cap T'_x = \emptyset$
\end{claim}

\begin{proofclaim}
Suppose for the sake of contradiction that $T'_{b_v} \cap T'_x \neq \emptyset$.
By Claim~\ref{claim:crossing-subtrees}, $T'_{a_v} \cap T'_{u_v} \neq \emptyset$, and so, the subtrees $T'_{a_v},T'_{b_v},T'_{u_v},T'_x$ are pairwise intersecting.
Hence by the Helly property (Lemma~\ref{lem:helly}), $T'_{a_v} \cap T'_{b_v} \cap T'_{u_v} \cap T'_x \neq \emptyset$.
However in such case, since $(T',{\cal X}')$ is a star-decomposition of $G'$ there must be a vertex dominating the cycle $(a_v,b_v,u_v,x)$ in $G'$, thereby contradicting Claim~\ref{claim:no-dominator-for-cycle}. 
\end{proofclaim}

As a result, $T'_{a_v} \cap T'_{u_v} \neq \emptyset$ by Claim~\ref{claim:crossing-subtrees} and $T'_{b_v} \cap T'_x = \emptyset$ by Claim~\ref{claim:disjoint-subtrees}.

Finally, suppose by way of contradiction $(T',{\cal X}')$ is not a star-decomposition of $G$.
In such case, since $G$ and $G'$ only differ in the edge $\{y,z\}$, there must exist $B \in T'_y \cap T'_z$ that is uniquely dominated by some of $y,z$ in $G'$.
More precisely, let us prove that only one of $y,z$ can dominate $B$.
By contradiction, suppose $B \subseteq N_{G'}[y] \cap N_{G'}[z] = \{b_v,x,y,z\}$ (indeed, $y \in C_0$ whereas $z \notin N_G[C_0]$).
Since $T'_{b_v} \cap T'_x = \emptyset$ by Claim~\ref{claim:disjoint-subtrees}, $B \subseteq \{b_v,y,z\}$ or $B \subseteq \{x,y,z,\}$.
Therefore, $B$ is a clique of $G'$.
However, since $B \neq V(G')$, there is a bag $B'$ adjacent to $B$ and by the properties of a tree decomposition $B \cap B'$ is a clique-separator of $G'$, thus contradicting the fact that $G'$ is prime by Claim~\ref{claim:gprime-prime-2}.
Consequently, either $B \subseteq N_{G'}[y]$ or $B \subseteq N_{G'}[z]$, and either $B \not\subseteq N_{G'}[y]$ or $B \not\subseteq N_{G'}[z]$.
In the following, let $\{s,t\} = \{y,z\}$ satisfy $B \subseteq N_{G'}[s]$, that is well-defined.
Let $B'$ be any bag adjacent to $B$ so that $t \in B'$ (such bag exists because $y,z \in N(b_v) \cap N(x)$, and $b_v \notin B$ or $x \notin B$ because $T'_{b_v} \cap T'_x = \emptyset$).
There are three cases.

\begin{itemize}
\item Suppose no vertex of $b_v,x,y,z$ dominates $B'$ in $G'$ (see Figure~\ref{fig:very-last-case} for an illustration).
Since $b_v \notin B$ or $x \notin B$ because $T'_{b_v} \cap T'_x = \emptyset$, therefore, $(B \cap B') \setminus (b_v,x,y,z) \neq \emptyset$, or else by the properties of a tree decomposition that would be a clique-separator of $G'$, thus contradicting the fact that $G'$ is prime by Claim~\ref{claim:gprime-prime-2}.  
Let $t' \in (B \cap B') \setminus (b_v,x,y,z)$.
Note that $t'$ and $t$ are non-adjacent in $G'$ because $t' \in N_{G'}[s]$ and $N_{G'}[y] \cap N_{G'}[z] = \{b_v,x,y,z\}$.
Let $s' \in B'$ dominate this bag.
Note that $s'$ and $s$ are non-adjacent in $G'$ because we assume $s' \notin \{b_v,x,y,z\}$, $t \in N_{G'}(s')$ and $N_{G'}[y] \cap N_{G'}[z] = \{b_v,x,y,z\}$.
In particular, $s' \neq t'$ and $(s,t',s',t)$ induces a path in $G$. 
By construction, $y \in C_0$ and $z \notin N_G[C_0]$, hence there must be some of $s',t'$ in $N_G(C_0)$. 
Since $N_G(C_0) \subseteq \{a_v,b_v,u_v,x\}$ and $s',t' \notin \{b_v,x,y,z\}$, therefore the pairs $\{s',t'\}$ and $\{a_v,u_v\}$ intersect.
However, by Claim~\ref{claim:neighbours-y} $a_v,u_v \notin N_G(y)$, similarly $a_v,u_v \notin N_G(z)$, that contradicts the existence of the path $(s,t',s',t)$ in $G$.
Consequently, assume in the remaining cases that there is some vertex of $b_v,x,y,z$ dominating $B'$ in $G'$.

\begin{figure}[h!]
	\centering
	\includegraphics[width=0.25\textwidth]{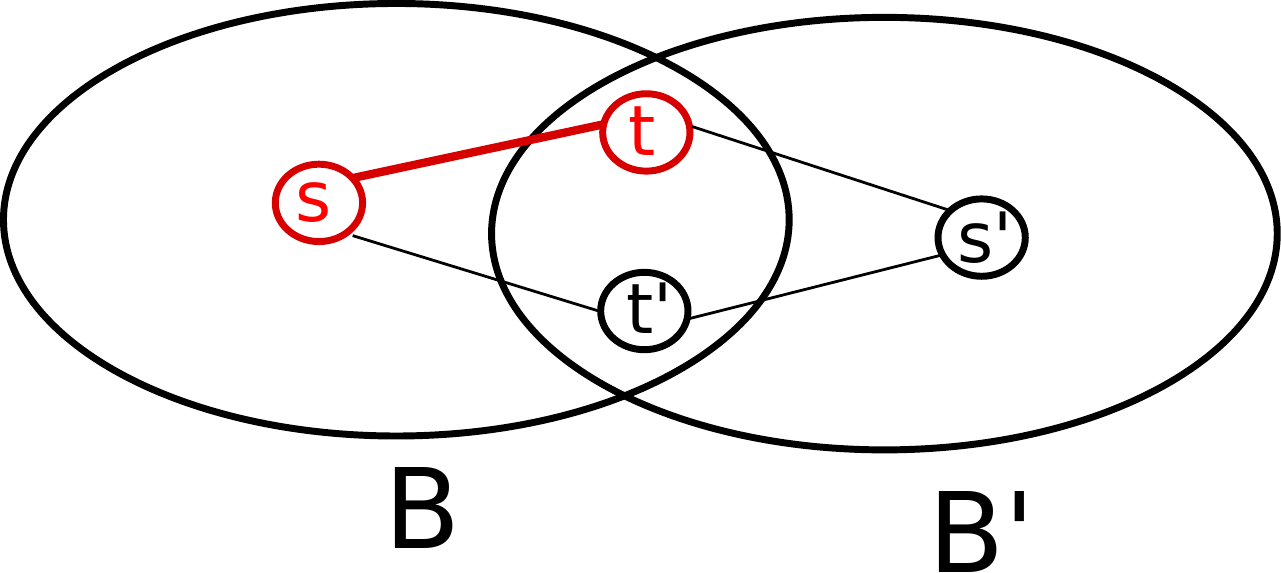}
	\caption{$B \subseteq N_{G'}[s]$, $B' \subseteq N_{G'}[s']$.}
	\label{fig:very-last-case}
\end{figure}

\item Suppose $B'$ is dominated by one of $y,z$ in $G'$.
We claim that $B$ and $B'$ are dominated by the same vertex of $y,z$, for if it were not the case $B \cap B' \subseteq N_{G'}[y] \cap N_{G'}[z] = \{b_v,x,y,z\}$, and so, since by the properties of a tree decomposition $B \cap B'$ is a separator of $G'$, and $b_v \notin B$ or $x \notin B$ because $T'_{b_v} \cap T'_x = \emptyset$, $B \cap B'$ should be a clique-separator of $G'$, thus contradicting the fact that $G'$ is prime by Claim~\ref{claim:gprime-prime-2}. 
However, in such a case bags $B,B'$ could be merged into one while preserving the property for the tree decomposition to be a star-decomposition of $G'$, that would contradict the minimality of $|{\cal X}'|$.
\item Therefore, $B'$ is dominated by some of $b_v,x$. 
We claim that there is a unique such bag $B'$ that is adjacent to $B$.
By contradiction, let $B'' \neq B'$ be adjacent to $B$ and such that $B''$ is also dominated by some of $b_v,x$.
In particular, if $B''  \cup B' \subseteq N[b_v]$ or $B'' \cup B' \subseteq N[x]$ then both bags $B',B''$ could be merged into one without violating the property for the tree decomposition to be a star-decomposition of $G'$, that would contradict the minimality of $|{\cal X}'|$.
Else, w.l.o.g. $B' \subseteq N[b_v]$ and $B'' \subseteq N[x]$.
Since $T'_{b_v} \cap T'_x = \emptyset$, $a_v,u_v \in N(b_v) \cap N(x)$ and $B',B''$ are adjacent to $B$, therefore, by the properties of the tree decomposition $a_v,u_v \in B$.
However, $a_v,u_v \notin N_G[y]$ by Claim~\ref{claim:neighbours-y}, $a_v$ and $u_v$ are non-adjacent, and either $z \in \{a_v,u_v\}$ or $a_v,u_v \notin N_G[z]$ (by the same proof as for Claim~\ref{claim:neighbours-y}), thus contradicting the fact that either $B \subseteq N_{G'}[y]$ or $B \subseteq N_{G'}[z]$.
Hence, the claim is proved and $B'$ is assumed to be the unique bag adjacent to $B$ such that $B \subseteq N[b_v]$ or $B \subseteq N[x]$.   
In particular, $B'$ is the unique bag adjacent to $B$ containing vertex $t$ (recall that $\{s,t\} = \{y,z\}$ and $B \subseteq N_{G'}[s]$, $B \not\subseteq N_{G'}[t]$).

Let us substitute the bags $B,B'$ with $B \setminus t, B' \cup \{s\}$.
We claim that this operation keeps the property for $(T',{\cal X}')$ to be a star-decomposition of $G'$.
To prove the claim, first note that the operation only modifies bags $B$ and $B'$, furthermore $B \setminus t \subseteq N[s]$ and $B' \cup \{s\} \subseteq N[b_v]$ or $B' \cup \{s\} \subseteq N[x]$.
Consequently, to prove the claim, it suffices to prove that the operation keeps the property for $(T',{\cal X}')$ to be a tree decomposition of $G'$ (for in such a case, it is always a star-decomposition).
Since $T'_t \setminus B$ is connected because $B'$ is the only bag containing vertex $t$ that is adjacent to the bag $B$, therefore, we are left to prove that there is no $w \in N_{G'}(t) \setminus s$ such that $T'_w \cap T'_t = \{B\}$.
By contradiction, let $w \in N_{G'}(t) \setminus s$ satisfy $T'_w \cap T'_t = \{B\}$.
Since $w \in B \subseteq N_{G'}[s]$, therefore $w \in N_{G'}(s) \cap N_{G'}(t) = N_{G}(y) \cap N_G(z) = \{b_v,x\}$.
Moreover, $w \notin B'$ because $t \in B'$ and we assume that $T'_w \cap T'_t = \{B\}$.
In such a case, since it is assumed that $B' \subseteq N[b_v]$ or $B' \subseteq N[x]$, and in addition $T'_x \cap T'_{b_v} = \emptyset$, let us write $\{w,w'\} = \{b_v,x\}$ such that $w \in B \setminus B', \ w' \in B' \setminus B$ and $B' \subseteq N_{G'}[w']$.
By the properties of a tree decomposition, $B \cap B'$ is a $b_vx$-separator of $G'$, so, $a_v,u_v \in B \cap B'$.
However, $a_v,u_v \notin N_G(y)$ by Claim~\ref{claim:neighbours-y} and similarly $a_v,u_v \notin N_G(z)$, that contradicts the fact that $B \subseteq N_{G'}[s]$ for some $s \in \{y,z\}$.
This finally proves the claim that substituting the bags $B,B'$ with $B \setminus t, B' \cup \{s\}$ keeps the property for $(T',{\cal X}')$ to be a star-decomposition of $G'$.

However, the above operation does not increase the number of bags $|{\cal X}'|$, furthermore there is one less bag that is not contained in the closed neighbourhood of some vertex in $G$.
This contradicts the minimality of $(T',{\cal X}')$ w.r.t. these two properties.
\end{itemize} 
As a result, we proved by contradiction that $(T',{\cal X}')$ is a star-decomposition of $G$, hence $tb(G) = 1$.
\end{proof}

\subsection{Complexity of Algorithm \texttt{Leaf-BottomUp}}

To complete this section, let us emphasize on some computational aspects of Algorithm \texttt{Leaf-BottomUp}, that will ensure the quadratic-time complexity of the algorithm.
%For simplicity, assume that there are $n > 4$ vertices in the graph (in fact, it can be checked that there are ten connected graphs with $n \leq 4$ and that they are all planar with tree-breadth one).
We here assume that the planar graph $G$ is encoded with adjacency lists.
Note that the adjacency lists can be updated in linear-time before each recursive call to the algorithm.

We will need as a routine to test whether two vertices are adjacent in \emph{constant-time}.
In order to achieve the goal, the following result (relying upon the bounded degeneracy of planar graphs) will be used:

\begin{lemma}[~\cite{Chrobak91}]
There exists a data structure such that each entry in the adjacency matrix of a planar graph can be looked up in constant time. 
The data structure uses linear storage, and can be constructed in linear time.
\end{lemma}

\subsubsection{Finding a leaf-vertex}
\label{sec:complexity-leaf-vertex}

At each call to the algorithm, it is first required to decide whether a leaf-vertex exists.
If that is the case, then one such a vertex must be computed.
Here is a way to achieve the goal in linear-time.
Let us start computing the degree sequence of $G$, then let us order the vertices of the graph $G$ by increasing degree.

\paragraph{Finding a leaf-vertex of Type 1.}
Let $v$ be any vertex of degree at least four.
We claim that a necessary condition for $v$ to be a leaf-vertex of Type 1 is that all but at most two neighbours of $v$ have degree four.
Indeed, if $v$ is a leaf-vertex of Type 1, then let $\Pi_v, d_v$ be defined as in Definition~\ref{def:leafVertex}.
By Lemma~\ref{lem:ends-separate}, either $V(G) = N[v] \cup \{d_v\}$ or $\Pi' = (a_v,d_v,c_v) \in {\cal P}_3(G)$ and $N[v] \setminus (a_v,c_v)$ is a full component of $G \setminus \Pi'$. 
In both cases, all neighbours in $N(v) \setminus (a_v,c_v)$ have degree four.

\begin{itemize}
\item Therefore, let us count the number of neighbours of degree four in $N(v)$, that can be done in ${\cal O}(deg(v))$-time simply by traversing the adjacency list of vertex $v$ (recall that the degree sequence of $G$ has been computed).

\item If there are all but at most two neighbours in $N(v)$ that have degree four, then we claim that one can construct the induced subgraph $G[N(v)]$ in ${\cal O}(deg(v))$-time.
Indeed, for every neighbour $u \in N(v)$ that has degree four, let us test in constant-time for each of its four neighbours whether they are adjacent to vertex $v$ --- we only keep those for which it is the case in the adjacency list of $u$ in $G[N(v)]$.
Then, for every $u \in N(v)$ that does not have degree four (there are at most two such vertices), let us construct the adjacency list of $u$ in $G[N(v)]$ simply by testing to which vertices in $N(v) \setminus u$ it is adjacent --- the latter takes constant-time by neighbour.

\item Once $G[N(v)]$ has been computed, it is easy to check whether it is a path in ${\cal O}(|N(v)|) = {\cal O}(deg(v))$-time.

\item Finally, let $u \in N(v)$ have degree four. Let us pick in constant-time any neighbour $d_v \in N(u) \setminus N(v)$ (note that such a vertex is unique if $G[N(v)]$ induces a path).
In order to decide whether $v$ is a leaf-vertex of Type 1, it is now sufficient to test whether vertex $d_v$ is adjacent to every vertex in $N(v)$ --- that takes constant-time by neighbour. 
\end{itemize}

\paragraph{Finding a leaf-vertex of Type 2.}
Recall that a vertex is a leaf-vertex of Type 2 if and only if it has degree three and its three neighbours induce a path.
Given any vertex of degree three, three adjacency tests are enough in order to determine whether its three neighbours induce a path --- and each adjacency tests takes constant-time.
Therefore, it can be checked in constant-time whether a vertex is a leaf-vertex of Type 2.

\paragraph{Finding a leaf-vertex of Type 3.}
By Definition~\ref{def:leafVertex}, a vertex $v$ is a leaf-vertex of Type 3 if and only if it has degree two and its two neighbours are non-adjacent and they have at least two common neighbours (including $v$).
Note that given a degree-two vertex, it can be checked whether its two neighbours are non-adjacent in constant-time.  
We now distinguish three cases.

\begin{enumerate}
\item First, suppose there is a vertex $v$ such that $N(v) = \{x,y\}$ and neighbour $x$ is a degree-two vertex.
In such case, let $N(x) = \{v,z\}$, in order to decide whether $v$ is a leaf-vertex of Type 3, it is sufficient to test in constant-time whether $y,z$ are adjacent.	
	
\item Second, suppose there are two degree-two vertices $v,v'$ that share the same two non-adjacent neighbours ({\it i.e.}, $N(v) = N(v') = \{x,y\}$ and $x,y$ are non-adjacent).
In such case, both vertices $v,v'$ are leaf-vertices of Type 3 (this case may happen if for instance, $G = K_{2,q}$ with $q \geq 2$).
In order to check whether this case happens, it is sufficient to sort the pairs $N(v)$ with $v$ being a degree-two vertex in linear-time (for instance, using a bucket-sort).

\item Else, let $V'$ contain every degree-two vertex $v$ with two non-adjacent neighbours of degree at least three (if one of the two neighbours of $v$ has degree two, we fall in the first case)
W.l.o.g., every vertex $v \in V'$ is uniquely determined by the pair $N(v)$ composed of its two neighbours (or else, we fall in the second case).
In such case, let us contract every $v \in V'$ to one of its two neighbours.
By doing so, we remove $v$ and we make the two vertices in $N(v)$ adjacent.
Note that all these edge-contractions are pairwise independent.
Let us call $G'$ the graph resulting from all edge-contractions, and let us call ``virtual edges'' any new edge resulting from an edge-contraction.
Then, let us list all triangles in the resulting graph $G'$, it can be done in linear-time~\cite{Papadimitriou81}.
By construction, $v \in V'$ is a leaf-vertex of Type 3 if and only if the virtual edge resulting from its contraction belongs to a triangle in which it is the unique virtual edge.
\end{enumerate}

Overall, finding a leaf-vertex in $G$ takes ${\cal O}(\sum_{v \in V} deg(v))$-time, that is ${\cal O}(n)$-time because $G$ is planar.

\subsubsection{Existence of a star-decomposition with two bags}

\begin{lemma}\label{lem:complexity-2bags}
Let $G$ be a planar graph, it can be decided in quadratic-time whether $G$ admits a star-decomposition with one or two bags.
\end{lemma}

\begin{proof}
$G$ admits a star-decomposition with one bag if and only if there is a universal vertex in $G$, hence it can be decided in linear-time.
Assume for the remaining of the proof that $G$ does not admit a star-decomposition with less than two bags.
We will consider two necessary conditions for some fixed pair $x,y$ to be the dominators of the only two bags in some star-decomposition of $G$.
For each of the two conditions, we will prove that all pairs satisfying the condition can be computed in quadratic-time.
Then, we will conclude the proof by showing that the two conditions are sufficient to ensure the existence of a star-decomposition of $G$ with two bags.

\begin{enumerate}
\item Recall that if it exists a star-decomposition of $G$ with two bags, then by the properties of a tree decomposition every vertex of $G$ must be contained in at least one bag. 
Therefore, if $x,y$ are the only two dominators of the bags in some star-decomposition of $G$, they must be a dominating pair of $G$.
It can be decided in ${\cal O}(deg(x) + deg(y))$-time whether a fixed pair $x,y$ is a dominating pair.
So, overall, it takes ${\cal O}(n^2)$-time to compute all dominating pairs of $G$ with $n$ being the order of the graph, for the graph is planar and so, it is a sparse graph.
\item Furthermore, recall that if it exists a star-decomposition of $G$ with two bags, then by the properties of a tree decomposition every edge of $G$ must be contained in at least one bag. 
Therefore, if there is a star-decomposition of $G$ with two bags that are respectively dominated by $x$ and $y$, then it must be the case that there does not exist any edge $e = \{u,v\}$ so that $u \in N[x] \setminus N[y]$ and $v \in N[y] \setminus N[x]$ (else, such an edge could not be contained in any of the two bags).
In order to decide whether the latter condition holds for some fixed pair $x,y$, it suffices to test whether every vertex of $N[x] \setminus N[y]$ is non-adjacent to all vertices in $N[y] \setminus N[x]$ --- it takes constant-time per test and so, ${\cal O}(deg(x) \cdot deg(y))$-time in total.
As a result, computing all pairs $x,y$ satisfying the condition requires ${\cal O}(\sum_{x,y} deg(x) \cdot deg(y)) = {\cal O}([\sum_x deg(x)][\sum_y deg(y)]) = {\cal O}(n^2)$-time because the graph $G$ is planar and so, it is a sparse graph.
\end{enumerate}

Finally, let $x,y$ satisfy the two above necessary conditions.
We claim that $(T,{\cal X})$ with $T$ being an edge and ${\cal X} = \{ N[x], N[y] \}$ is a star-decomposition of $G$.
Indeed, every vertex is contained into a bag because the pair $x,y$ satisfies the first necessary condition.
Furthermore, every edge has its both ends contained into a common bag because the pair $x,y$ satisfies the second necessary condition.
Last, all the bags containing a common vertex induce a subtree because there are only two bags.
As a result, $(T,{\cal X})$ is a tree decomposition of $G$.
Since each bag of ${\cal X}$ is respectively dominated by $x$ or $y$, therefore $(T,{\cal X})$ is indeed a star-decomposition of $G$, that proves the claim, hence the lemma.
\end{proof}

Note that in any execution of Algorithm \texttt{Leaf-BottomUp}, it is verified at most once whether some planar graph admits a star-decomposition with one or two bags.

\subsubsection{Upper-bound on the number of steps in the algorithm}

\begin{lemma}\label{lem:ub-steps}
Let $G$ be a prime planar graph with $n$ vertices and $m$ edges.
Then, there are at most $5n-m$ recursive calls to the Algorithm \texttt{Leaf-BottomUp}, that is ${\cal O}(n)$. 
\end{lemma}

\begin{proof}
First note that since $G$ is planar by the hypothesis, $5n-m \geq 0$ and $5n-m = {\cal O}(n)$. 
Let $G'$ with $n'$ vertices and $m'$ edges so that Algorithm \texttt{Leaf-BottomUp} is recursively applied on $G'$ when $G$ is the input.
Since there is at most one such a graph $G'$ ({\it i.e.}, there is no more than one recursive call at each call of the algorithm), furthermore $G'$ is prime and planar, therefore, in order to prove the lemma it suffices to prove that $5n' - m '< 5n -m$.
To prove it, let us consider at which step of the algorithm the recursive call occurs.
\begin{itemize}
\item If it is at~\ref{step:type-1}, then $G'$ is obtained by removing a leaf-vertex of Type 1, denoted by $v$, and then contracting all the internal vertices in the path $\Pi_v$ (induced by $N(v)$) to a single edge.
Therefore, $n' = n - deg(v) + 3$, $m' = m - 3deg(v) + 8$ and so, $5n'-m' =  5n - m - (2deg(v) - 7) < 5n - m$ because $deg(v) \geq 4$.

Thus, from now on let us assume we fall in~\ref{step:type-2-3}, {\it i.e.}, a leaf-vertex of Type 2 or 3 is considered, denoted by $v$.

\item If the recursion happens at~\ref{step:prime-case}~\ref{step:many-neighbours}, then $G'$ is obtained by removing $v$.
Therefore, $n' = n - 1$ and either $m' = m-3$ (if $v$ is of Type 2) or $m'=m-2$ (if $v$ is of Type 3), hence $m' \geq m - 3$ and so, $5n' - m' \leq 5n -m - 2 < 5n - m$.
\item If it is at~\ref{step:prime-case}~\ref{step:few-neighbours}, then we fall in~\ref{step:prime-case}~\ref{step:few-neighbours}~\ref{step:two-neighbours} (no recursion occurs in~\ref{step:prime-case}~\ref{step:few-neighbours}~\ref{step:one-neighbour}), thus $G'$ is obtained by making $v$ adjacent to the two vertices in $(N(a_v) \cap N(c_v)) \setminus v$ (including $b_v$ in the case when $v$ is of Type 3).
Therefore, $n' = n$ and either $m'=m+1$ (if $v$ is of Type 2) or $m'=m+2$ (if $v$ is of Type 3), hence $m' \geq m+1$ and so, $5n' -m' \leq 5n -m - 1 < 5n-m$.

\item Else, the recursion happens at~\ref{step:new-clique-sep}.
Recall that in such case, there exists a vertex $u_v$ such that $(b_v,u_v)$ is a clique-separator of $G \setminus v$.
Adding an edge between $v$ and $b_v$ if it does not exist, decreases $5n-m$ by $1$, therefore from now on let us assume that $v$ is a leaf-vertex of Type 2.

\begin{itemize}
\item if it is at~\ref{step:new-clique-sep}~\ref{step:no-separation}, then $G'$ is obtained by contracting the edge $\{v,a_v\}$.
Therefore, $n'=n-1$, $m'=m-2$, hence $5n'-m' = 5n-m-3 < 5n-m$.
\item If it is at~\ref{step:new-clique-sep}~\ref{step:separation}~\ref{step:contraction-case}, then $G'$ is obtained by adding an edge between $b_v$ and some vertex $x \in (N(a_v) \cap N(u_v)) \setminus b_v$ then contracting this edge. 
Furthermore, $N(b_v) = \{a_v,c_v,u_v,v\}$ in such case and $c_v,v \notin N(x)$.
Therefore, $n'=n-1$, $m'=m-2$ and so, $5n'-m' = 5n-m-3 < 5n-m$.
\item If it is at~\ref{step:new-clique-sep}~\ref{step:separation}~\ref{step:diamond-case}, then $G'$ is obtained by contracting the edge $\{b_v,x\}$ where $x \in N(a_v) \cap N(u_v) \cap N(b_v)$.
Furthermore, $N(b_v) = \{a_v,c_v,u_v,v,x\}$ in such case and $c_v,v \notin N(x)$.
Therefore, $n'=n-1$, $m' = m-3$ and so, $5n' - m' = 5n -m - 2 < 5n - m$.
\item Finally, in all other cases the recursive call happens at~\ref{step:new-clique-sep}~\ref{step:separation}~\ref{step:no-diamond-case}.
Then, $G'$ is obtained by adding an edge between two vertices $y,z \in N(b_v) \cap N(x)$ for some $x \in (N(a_v) \cap N(u_v)) \setminus b_v$.
Therefore, $n' = n$, $m' = m+1$ and so, $5n' -m' = 5n -m - 1 < 5n-m$.
\end{itemize}

\end{itemize}
\end{proof}

\section{Conclusion and Open questions}
\label{sec:conclusion}

%open questions resolved
On the negative side, we proved the NP-hardness of computing five metric graph invariants (namely, tree-breadth, path-length, path-breadth, $k$-good tree and path decompositions) whose complexity has been left open in several works~\cite{Dourisboure2007b,Dragan2014a,Dragan2014b}.
These results add up to the proof in~\cite{Lokshtanov2010} that it is NP-hard to compute the tree-length.
We leave as a future work further study on the border between tractable and intractable instances for the problem of computing the above metric graph invariants.
Especially, what are the graph classes for which it can be decided in polynomial-time whether a graph admits a star-decomposition ?
In this paper, we partially answer to this question by proving that it is the case for bipartite graphs and planar graphs.
Based on these two positive results, we conjecture that the problem is Fixed-Parameter Tractable when it is parameterized by the \emph{clique-number} of the graph (note that there is a large clique in all the graphs obtained from our polynomial-time reductions).
Intermediate challenges could be to determine whether the problem is Fixed-Parameter Tractable when it is parameterized by the genus, the tree-width or the Hardwiger number.

%sharp bounds on the inapproximability: 2-approx for tree-length
Finally, we notice that all our NP-hardness results imply that the above metric graph invariants cannot be approximated below some constant-factor.
There remains a gap between our inapproximability results and the constant-ratio of the approximation algorithms in~\cite{Dourisboure2007b,Dragan2014b}.
Therefore, we leave as an interesting open question whether we can fill in this gap.

\bibliographystyle{abbrv}
\bibliography{biblio}

\end{document}